%% file: sar.tex
\def\mode{0}
\if 0\mode
	\documentclass[a4paper,12pt]{article}
\else
	\documentclass[a4paper,12pt]{report}
\fi

\usepackage[utf8x]{inputenc}
\usepackage[english]{babel}
\usepackage{amsmath}
\usepackage{amsthm}
\usepackage{amssymb}
\usepackage{caption}
\usepackage{float}
\usepackage{mathtools}
\usepackage{subcaption}

\usepackage{graphicx}
\graphicspath{ {./images/} }
\usepackage{comment}
\usepackage{tocloft}
\usepackage{tikz}
\usepackage[strings]{underscore}
\usepackage[explicit]{titlesec}
\usepackage{listings}
\usepackage[linesnumbered,ruled,noend]{algorithm2e}
\usepackage{pdfpages}
\usepackage{ifoddpage}

\usepackage [english]{babel}
\usepackage [autostyle, english = american]{csquotes}

\usepackage{xcolor}
\usepackage{soul}

\usepackage[toc,page]{appendix}

\usepackage[numbers,sort&compress,merge]{natbib}
\usepackage{chapterbib}
\usepackage[hyphens]{url}
\usepackage[colorlinks = true,
linkcolor = blue,
urlcolor = blue,
citecolor = blue,
anchorcolor = black]{hyperref}
\hypersetup{breaklinks=true}

\usepackage{siunitx}
\usepackage{color}
\usepackage[shortlabels]{enumitem}
\usepackage{mathabx}
\usepackage{thmtools,thm-restate}
\usepackage{multirow}
\usepackage{xspace}
\usepackage{authblk}
\def\email#1{{\tt#1}}

\titleformat{\section}
{\normalfont\Large\bfseries}{\thesection}{1em}{\hypersetup{linkcolor = black}
    \hyperlink{toc}{#1}}
\titleformat{\subsection}
{\normalfont\large\bfseries}{\thesubsection}{1em}{\hypersetup{linkcolor = black}\hyperlink{toc}{#1}}

\newif\ifthesis

\if 0\mode
\thesisfalse
\else
\thesistrue
\fi

\ifthesis
\newtheorem{theorem}{Theorem}[chapter]
\usepackage{quotchap}
\let\oldsection\section
\let\subsubsection\subsection
\renewcommand{\subsection}{\oldsection}
\renewcommand{\section}{\chapter}
\newcommand{\paper}{thesis\xspace}
\newcommand{\Section}{Chapter\xspace}

\newcommand{\ssection}{chapter\xspace}
\newcommand{\Subsection}{Section\xspace}
\newcommand{\Subsections}{Sections\xspace}
\newcommand{\ssubsection}{section\xspace}
\newcommand{\ssubsections}{sections\xspace}
\else
\newtheorem{theorem}{Theorem}[section]
\newcommand{\paper}{paper\xspace}
\newcommand{\Section}{Section\xspace}

\newcommand{\ssection}{section\xspace}
\newcommand{\Subsection}{Section\xspace}
\newcommand{\Subsections}{Sections\xspace}
\newcommand{\ssubsection}{subsection\xspace}
\newcommand{\ssubsections}{subsections\xspace}
\fi

\newtheorem{lemma}[theorem]{Lemma}

\def\pagespace{\thispagestyle{empty}\ \clearpage}

\def\reals{{\mathbb R}}
\def\eps{{\varepsilon}}
\def\bd{{\partial}}
\newcommand{\set}[1]{\{{#1}\}}

\newif\ifcomments
\commentstrue

\ifcomments

\def\micha#1{\textcolor{red}{\textsc{(Micha says: }}\textsf{#1})}
\def\golan#1{\textcolor{blue}{\textsc{(Golan says: }}\textsf{#1})}
\def\danny#1{\textcolor{orange}{\textsc{(Danny says: }}\textsf{#1})}
\def\marc#1{\textcolor{magenta}{\textsc{(Marc says: }}\textsf{#1})}
\def\todo#1{\textcolor{blue}{\textbf{TODO:} #1}}
\else
\def\golan#1{}
\def\danny#1{}
\def\micha#1{}
\def\marc#1{}
\def\st#1{}
\def\todo#1{}
\fi

\newcommand{\R}{\mathbb{R}}

\newcommand{\A}{\mathcal{A}}
\newcommand{\B}{\mathcal{B}}
\newcommand{\D}{\mathcal{D}}
\newcommand{\W}{\mathcal{W}}
\newcommand{\V}{\mathcal{V}}
\newcommand{\Hv}{H_{\vec{v}}}
\newcommand{\Gv}{G_{\vec{v}}}
\newcommand{\kv}{k_{\vec{v}}}
\newcommand{\mv}{m_{\vec{v}}}
\newcommand{\VA}{\A(\V^\bd)}

\def\SB#1{\textsubscript{#1}}

\def\SED{{\sf SED}}
\def\FVD{{\sf FVD}}

\newcommand{\aabr}[1]{\text{\sf{AABR(}}#1\text{\sf{)}}}

\newcommand{\ust}{UST}
\newcommand{\lst}{LST}
\newcommand{\saust}{SA-UST}
\newcommand{\salst}{SA-LST}
\newcommand{\salstv}[1]{{\rm{SA-LST}\SB{$|\vec{v}|$}\rm{(}#1\rm{)}}}
\newcommand{\saustv}[1]{{\rm{SA-UST}\SB{$|\vec{v}|$}\rm{(}#1\rm{)}}}
\newcommand{\salsta}[1]{{\rm{SA-LST}\SB{{\sf AABR}}\rm{(}#1\rm{)}}}
\newcommand{\sausta}[1]{{\rm{SA-UST}\SB{{\sf AABR}}\rm{(}#1\rm{)}}}
\newcommand{\salstd}[1]{{\rm{SA-LST}\SB{\SED}\rm{(}#1\rm{)}}}
\newcommand{\saustd}[1]{{\rm{SA-UST}\SB{\SED}\rm{(}#1\rm{)}}}

\DeclarePairedDelimiter\ceil{\lceil}{\rceil}
\DeclarePairedDelimiter\floor{\lfloor}{\rfloor}

\newcommand\restr[2]{{
		\left.\kern-\nulldelimiterspace 
		#1 
		\vphantom{\big|} 
		\right|_{#2} 
}}

\newcommand{\keywords}[1]{\par\addvspace\baselineskip%
	\noindent{\textbf{Keywords:}}\enspace\ignorespaces#1}

\begin{document}

\ifthesis
	\pagenumbering{roman}
	\setcounter{tocdepth}{1}
	\input {general/thesis_title}
	\pagespace
	\input {general/acknowledgments}
	\pagespace
	\pagespace
	\input {general/abstract}
	\pagespace
	
	\hypersetup{linkcolor = black}
	\protect\hypertarget{toc}{}
	\setcounter{page}{1}
	\tableofcontents
	\hypersetup{linkcolor = blue}
	\pagebreak
	
	\pagenumbering{arabic}
	\pagespace
	\setcounter{page}{1}
	\input {general/introduction}
	\newpage
	\input{chapters/related_work}
	\newpage
	\input{chapters/labeled}	
	\newpage
	\input{chapters/labeled_sa}	
	\newpage
	\input{chapters/unlabeled}
	\newpage
	\input{chapters/unlabeled_sa/unlabeled_sa}

	\newpage
	\input{general/conclusion}
	\newpage
	\input{chapters/appendix}
\else
	\input {general/title}

\input {general/abstract}\input {general/introduction}\input{chapters/related_work}\input{chapters/labeled}\input{chapters/labeled_sa}\input{chapters/unlabeled}\input{chapters/unlabeled_sa/unlabeled_sa}\input{general/conclusion}\input{chapters/appendix}
\fi

\bibliographystyle{abbrv}
\bibliography{sar}

\ifthesis
	\checkoddpage\ifoddpage	
	\else
		\pagespace
		\pagespace
	\fi
	\includepdf[pages=-]{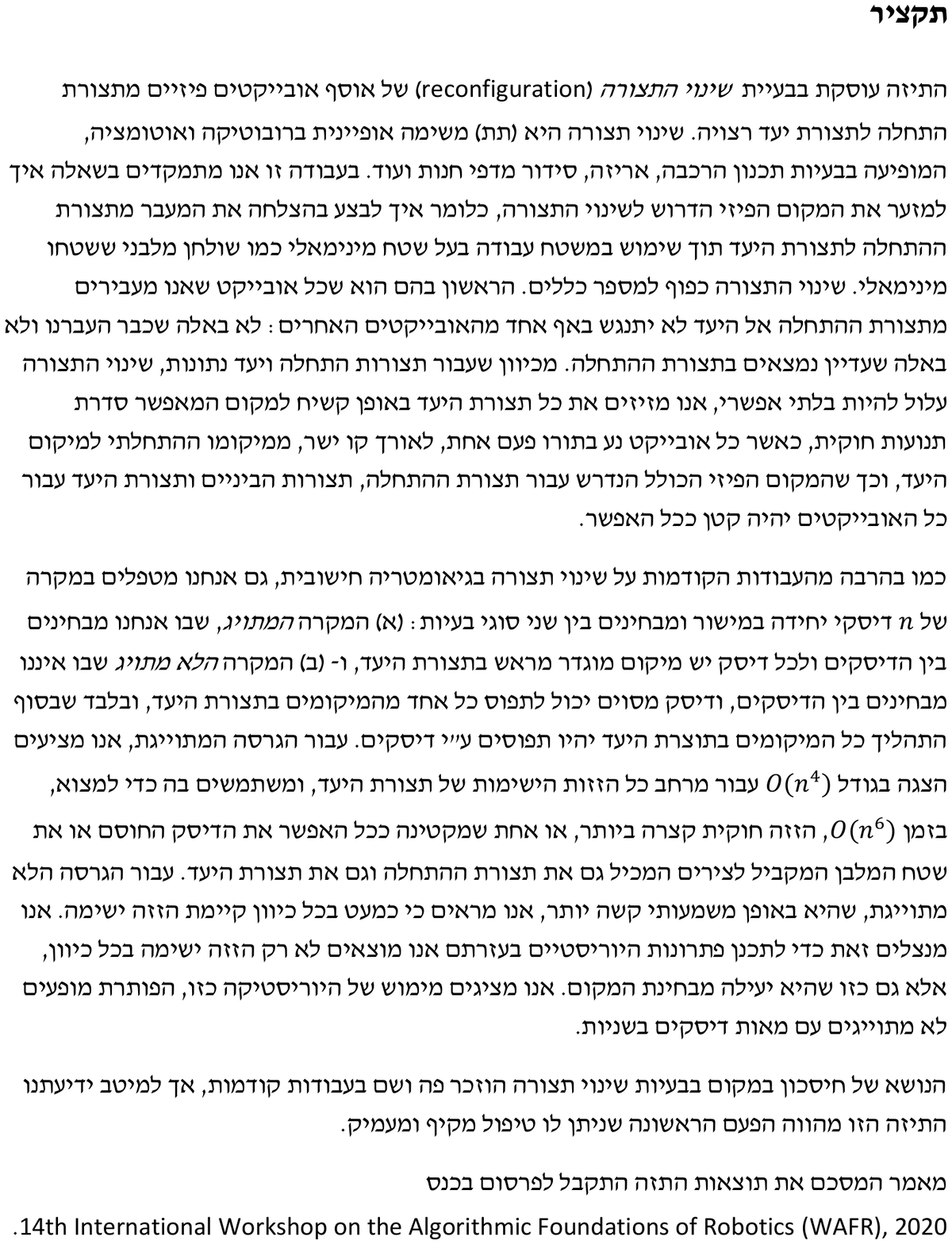}
\fi

\end{document}

%% file: general/thesis_title.tex
\begin{titlepage}
\begin{center}

	\begin{figure}[t]
	\centering 
		\includegraphics[width=0.99\textwidth]{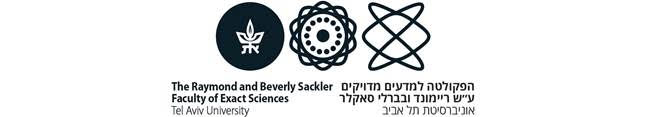}
	\end{figure}

  {\Huge{\sf Space-Aware Reconfiguration}}\\\vspace{1.4ex}

  {\large Thesis submitted in partial fulfillment of the requirements for the
  				M.Sc. degree in the \\\vspace{0.3cm}
  				Blavatnik School of Computer Science, Tel-Aviv University}\\\vspace{0.5cm}
  {\large  by}\\\vspace{0.5cm}
  {\LARGE\bf Golan Miglioli-Levy}\\\vspace{5cm}
  {\large This work has been carried out under the supervision
  				of}\\\vspace{0.2ex}
  {\large Prof. Dan Halperin and Prof. Micha Sharir}\\\vfill

  {\large September 2020}

\end{center}

\end{titlepage}

%% file: general/acknowledgments.tex
\subsection*{Acknowledgments}
First and foremost, I would like to thank my advisors:
Prof. Dan Halperin for guiding me and was always available to consult with.
This was an amazing journey that I would not have chosen to do with any other advisor.
Prof. Micha Sharir, who took on himself days as nights to propel our work and bring it to the maximum, making sure to be a remarkable example and a mentor.
His dedication and effort was an inspiration to me.
A special thanks to Marc v.Kreveld who initiated the main idea of this work and accompanied it along the way.
None of this work could not have been done without dear friends and a fun working environment, especially in the Computational Geometry and Robotics lab, which includes countless of insightful talks with Nir Goren, Michal Kleinbort, Efi Fogel and Tzvika Geft.
I would also like to emphasis the love and support I have received from my family, my friends, my roommate and my loving girlfriend while working on this thesis.

%% file: general/abstract.tex
\begin{abstract}
We consider the problem of \emph{reconfiguring} a set of physical objects into a desired target configuration, a typical (sub)task in robotics and automation, arising in product assembly, packaging, stocking store shelves, and more.
In this \paper we address a variant, which we call \emph{space-aware reconfiguration}, where the goal is to minimize the physical space needed for the reconfiguration, while obeying constraints on the allowable collision-free motions of the objects.
Since for given start and target configurations, reconfiguration may be impossible, we translate the entire target configuration rigidly into a location that admits a valid sequence of moves, where each object moves in turn just once, along a straight line, from its starting to its target location, so that the overall physical space required by the start, all intermediate, and target configurations for all the objects is minimized.

We investigate two variants of space-aware reconfiguration for the often examined setting of $n$ \emph{unit discs} in the plane, depending on whether the discs are distinguishable (labeled) or indistinguishable (unlabeled).
For the labeled case, we propose a representation of size $O(n^4)$ of the space of all feasible initial rigid translations, and use it to find, in $O(n^6)$ time, a shortest valid translation, or one that minimizes the enclosing disc or axis-aligned rectangle of both the start and target configurations.
For the significantly harder unlabeled case, we show that for almost every direction, there exists a translation in this direction that makes the problem feasible.
We use this to devise heuristic solutions, where we optimize the translation under stricter notions of feasibility.
We present an implementation of such a heuristic, which solves unlabeled instances with hundreds of discs in seconds.

\ifthesis
\else
\keywords{Motion planning, Disc reconfiguration, Smallest enclosing disc}
\fi
\end{abstract}

%% file: general/introduction.tex
\section{Introduction}
Consider a set of $n$ objects in the plane or in three-dimensional space and two configurations of these objects, a start configuration $S$ and a target configuration $T$, where in each configuration the objects are pairwise interior disjoint.
A typical \emph{reconfiguration} problem asks to efficiently move the objects from $S$ to $T$, subject to constraints on the allowable motions, the most notable of which is that all the moves be collision free.

In the specific problem studied in this \paper, we are given $n$ unit discs in the plane and we wish to move them from some start configuration to a target configuration.
A valid move is a translation of one disc in a fixed direction from one placement to another without colliding with the other discs. 
The goal in earlier works on this problem was to minimize the number of moves, and the goal in the present study is to find an initial rigid translation of the discs of $T$, that minimizes the size of the physical space needed for the reconfiguration, under the constraint that each disc moves exactly once.
This problem, like most problems in the domain of reconfiguration, comes in (at least) two flavors: \emph{labeled} and \emph{unlabeled}.
In the labeled version, each object has a unique label, which marks its start placement and its unique target placement.
In the unlabeled version the objects are indistinguishable, and we do not care which object finally gets to any specific target placement, as long as all the target placements are occupied at the end of the process; in particular all the objects are isothetic (as are the unit discs in our study).
For the unlabeled case, without an initial shift of the target configuration, Abellanas et al.~\cite{DBLP:journals/comgeo/AbellanasBHORT06} have shown that $2n-1$ moves are always sufficient.
Dumitrescu and Jiang~\cite{in-the-plane} have shown that $\lfloor 5n/3 \rfloor - 1$ moves are sometimes necessary, and that finding the minimum number of moves is NP-Hard.
For the labeled case, Abellanas et al.~\cite{DBLP:journals/comgeo/AbellanasBHORT06} have shown that $2n$ moves are always sufficient and sometimes necessary.
These are several examples of reconfiguration problems that have been studied in discrete and computational geometry; see, e.g., \cite{lifting,sliding,graphs-and-grids,d-mp-2013}.
Varying the type of objects, the ambient space, the constraints on the motion and the optimization criteria, we get a wide range of problems, many of which are hard.

Similar problems arise in robotics.
For example, such problems arise when a robot needs to arrange products on a shelf in a store, or when a robot needs to move objects around in order to access a
specific product that needs to be picked up; see, e.g.,~\cite{DBLP:conf/rss/HanSKBY17,DBLP:conf/icra/HavurOEP14,DBLP:conf/iros/LevihnIS12}.
In robotics, these problems are often referred to as object \emph{rearrangement} problems.
In this \paper, though, we will stick to the term reconfiguration, which is also in common use.

Another prominent example from robotics and automation is the \emph{assembly planning} problem (see, e.g., \cite{DBLP:journals/algorithmica/HalperinLW00}), in which the target configuration of the objects comprises their positions in the desired product.
The goal of assembly planning is to design a sequence of motions that will bring the parts together to form the desired product, and we want this sequence to be (collision-free and) optimal according to various criteria~\cite{DBLP:conf/case/GeftTGH19,goldwasser1996complexity}.
Yet another area where variants of the reconfiguration problem arise is in motion planning for a swarm of robots, where the goal is to minimize the total execution time of parallel collision-free motions of the robots; see, e.g., the recent work~\cite{DBLP:journals/siamcomp/DemaineFKMS19}.

We address a certain criterion, which, to the best of our knowledge, has hardly been studied earlier: \emph{minimizing the physical space} needed to carry out the desired assembly or reconfiguration.
Abellanas et al.~\cite{DBLP:journals/comgeo/AbellanasBHORT06} did study a similar set of problems, in which the discs are placed inside different types of confined spaces. Their technique shows how to minimize the number of moves, given a prescribed size for a bounding rectangle of $S$ and $T$.
We adopt a different approach.
We consider $T$ a rigid configuration that can be placed anywhere in the workspace, and the goal is to find a placement for $T$ for which there exists a feasible (collision-free) sequence of moves,
where each disc moves \emph{exactly once} along the straight segment that connects its start placement and to its target placement.
The region occupied by $S$ and by $T$ in its translated location, together with the space required for the reconfiguration motion of all the objects, should be minimal according to various possible criteria.
In this \paper we consider the setup where we allow $T$ only to be translated.
We call this problem \emph{space-aware reconfiguration};
we study it in this \paper for the case of unit discs in the plane.
Rigidly translating $T$ into a different location in the plane ensures that the target objects maintain the same positional relations between them.
This is a desired property in some reconfiguration problems, such as assembly planning.
In this approach we do not care where the position of the final product is, as long as the space required for the reconfiguration is minimized.
Moreover, as we will see, the variant where only translations of $T$ are allowed is already quite difficult to solve.
Tackling the general case, where we allow an arbitrary initial rigid motion of $T$, is left as a challenging open problem.

We say that a disc is placed at a point $p$ if its center is placed at $p$.
To avoid confusion between placeholder positions for discs (start or target) and the actual discs placed at these positions, we define a \emph{valid configuration} $P$ to be a finite set of points, such that every pair of points in the set lie at distance $\geq 2$ from one another, that is, we can place a unit disc at each point of $P$, so that the discs are pairwise interior disjoint.
For any point $p$, we denote by $D_r(p)$ the open disc centered at $p$ with radius $r$.
If $r$ is not specified, then $D(p)$ is a unit disc ($r=1$).
For any valid configuration $P$, we denote $D(P) = \set{D(p) \mid p \in P}$.

Let $S$ and $T$ be two valid configurations, of $n$ points each, where $S$ represents the centers of the start positions, at which $n$ unit discs are initially placed, and $T$ represents the centers of the target positions.
We look for a sequence of $n$ moves that bring the discs from $S$ to $T$.
A move consists of a single translation of one disc from $D(S)$ to a position in $T$, such that the disc does not collide with any other (stationary) disc on its way---neither with a disc in a start position, which has not been moved yet, nor with a disc that has already been moved to a target position.
Each disc has to perform exactly one such move.
We call such a sequence of moves an \emph{itinerary}.
We say that an itinerary is \emph{valid} if all of its moves are collision-free.
We denote such an Unlabeled Single Translation instance of the problem by \ust($S$,$T$), and a Labeled Single Translation instance by \lst($S$,$T$,$M$), where $M$ is the matching between $S$ and $T$ induced by the labels; that is, each position in $S$ is matched by $M$ to the position in $T$ with the same label.
We call an instance of the problem \emph{feasible} if it has a valid (collision-free) itinerary.

\begin{figure}
		\centering
		\includegraphics[trim=280 200 280 160, clip, width=0.3\textwidth]{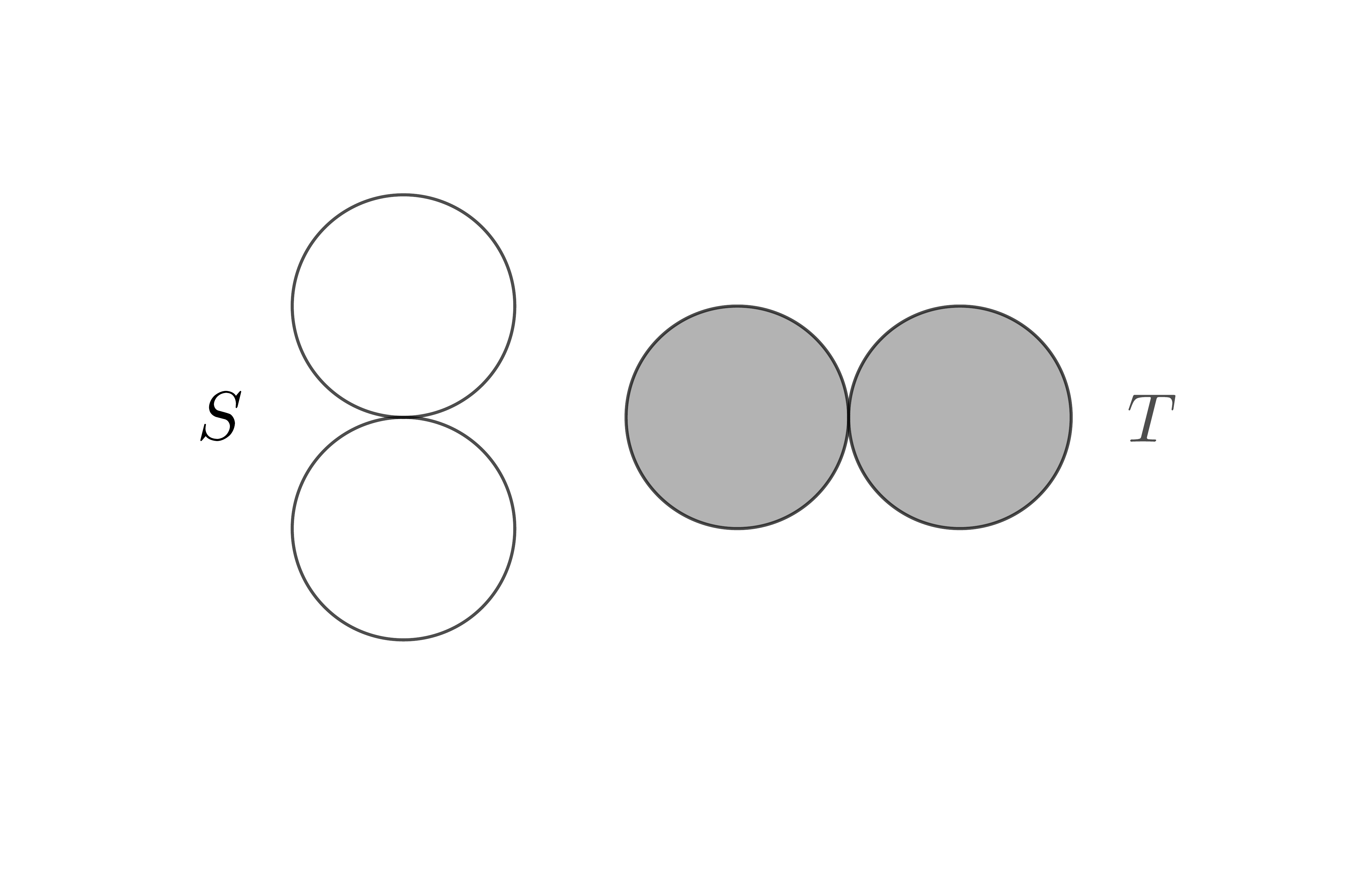}
		\caption{
			\sf An infeasible instance for the stationary unlabeled (or labeled) version.
			The discs of $D(S)$ are drawn empty while the discs of $D(T)$ are drawn shaded. 
		} \label{fig:unlabeled_infeasible}
\end{figure}

It is easy to see (consider Figure~\ref{fig:unlabeled_infeasible}) that even in the unlabeled version, the stationary version of this problem, in which $T$ cannot be translated, may not have a solution.
If the shaded discs in the figure were placed higher (so that their centers were collinear with the center of the top empty disc, say), the problem would have been feasible.
We therefore look for a vector $\vec{v}$ for which a valid itinerary exists from $S$ to $T + \vec{v}$ (i.e., $T$ translated by $\vec{v}$).%
\footnote{The translations from $T$ to $T+\vec{v}$ do not count as moves.
We often refer to it as \emph{the initial translation}.}
In the labeled case, the translated targets retain their
labels after the translation.
That is, if target position $i$ was at the point $t_i$, the point $t_i+\vec{v}$ is now the $i$th target position.
Observe that the initial location of $T$ (when $\vec{v} = (0,0)$) is now meaningless.
From now on, we assume that the input location of $T$ is placed to overlap with $S$ as much as possible, e.g., $S$ and $T$ share their centers of mass or the centers of their smallest enclosing discs.
A placement of this kind is ideal for the space-aware paradigm used in this work (although in practice it may not be valid).

In the space-aware variant studied in this \paper, we look for a translation $\vec{v}$ such that (a) $S$ and $T+\vec{v}$ admit a valid itinerary, and (b) some measure of `nearness' of $T+\vec{v}$ to $S$ is minimized. 
One typical variant of (b) is to require that some prescribed bounding shape (e.g., an axis-aligned rectangle or a disc), of $D(S)$ \emph{together} with $D(T + \vec{v})$, have minimum area.
We denote these space-aware variants as \saust($S$,$T$), and \salst($S$,$T$,$M$), for the unlabeled and labeled variants, respectively (ignoring in these notations the specific optimization criterion to be used).
We note again that our space-aware variant differs from the previous studies in that we insist on executing only $n$ moves, one move per disc.
This model raises several problems: 
The first challenge is to construct the space of \emph{valid} translations (those that have a valid itinerary), or, in the unlabeled case, to construct a sufficiently large subset thereof.
Second, we want to find a valid translation that minimizes some measure of optimality.

To the best of our knowledge, it has not been proven that deciding the existence of a valid itinerary for the unlabeled version (assuming $T$ is stationary, and allowing one move per disc) is NP-hard, but similar reconfiguration problems, such as those called OMC~\cite{DBLP:journals/comgeo/AbellanasBHORT06} and U-TRANS-RP~\cite{in-the-plane}, have been shown to be NP-hard.
Both problems seek to find a valid itinerary of translations of discs from a start configuration to a target configuration.
Specifically, in OMC (one move per coin), each of the discs is given a subset of possible final targets, and can move (as in our model) exactly once.
In U-TRANS-RP, the goal is to decide whether a valid itinerary of at most $k$ moves exists, for the unlabeled variant.
The NP-hardness of U-TRANS-RP is shown for the case where $k$ is smaller than $n$ (some discs are already at the target locations, and some may move more than once).
Although at the moment we do not know whether it is possible to reduce any of these problems to our setting, we believe that the unlabeled version of our setting is indeed NP-hard, and so minimizing some criterion of optimality for \saust($S$,$T$) is most likely even harder.

There is a large body of related research on algorithms for multi-robot motion planning and multi-agent path finding;
for recent reviews, see, e.g.,  \cite{ hks-r-18} and \cite{ DBLP:conf/socs/SternSFK0WLA0KB19} respectively.
Notice however that in contrast to these more general problems, the focus of our work here (as in \cite{DBLP:journals/comgeo/AbellanasBHORT06,in-the-plane,lifting,sliding,d-mp-2013}) is on the special case where every object is transferred by a very small number (typically one or two) of simple atomic moves (e.g., translation along a segment).
Therefore the planning techniques are of a rather different nature.

\paragraph{Contribution.}
The unlabeled case appears to be (and most likely is) much harder than the labeled case.
We therefore begin by studying the labeled case.
We present, in \Section~\ref{section:labeled}, an algorithm for the labeled case that runs in $O(n^6)$ time, for constructing the space of all valid translations.
We then show, in \Section~\ref{section:labeled_sa}, how to find a valid translation that minimizes some measure of space-aware optimality.
Specifically, we consider three such measures:
(i) minimizing the length of the translation vector $\vec{v}$, relative to some ideal placement, as discussed above;
(ii) minimizing the area of the axis-aligned bounding rectangle of $D(S) \cup D(T+\vec{v})$;
(iii) minimizing the area of the smallest enclosing disc of $D(S) \cup D(T+\vec{v})$.
We refer to the corresponding variants of the problem as \salstv{$S$,$T$,$M$}, \salsta{$S$,$T$,$M$}, and \salstd{$S$,$T$,$M$}.
All the variants that we study can be solved by algorithms that run in $O(n^6)$ time. 
(The minimization steps of the algorithms are actually faster; this bound is the cost of the first step, of constructing the space of all valid translations.)

The unlabeled case appears, as already noted, to be much harder.
We first show, in \Section~\ref{section:unlabeled}, that we can find a valid translation in almost any prescribed direction,%
\footnote{The only exceptional directions are those of common inner tangents of pairs of tangent discs. If no tangency between the discs is allowed, all directions admit a valid translation.}
if we translate $T$ sufficiently far away (see \Section~\ref{section:unlabeled} for a more precise statement).
Although this is a useful result, it suffers from two problems:
(i) It does not produce the space of all valid translations (which is likely a very difficult task).
(ii) It is contrary to our goal of achieving space-aware optimality.

We study in \Subsection~\ref{subsec:heuristics} practical heuristic techniques that aim to find shorter valid translations, at the cost of further restricting the notion of validity.
That is, the solutions that we obtain are valid, but we may be missing other, more optimal valid solutions.
Under our strictest notion of validity, we present an algorithm for finding an optimal translation for \salstv{$S$,$T$,$M$} and \salsta{$S$,$T$,$M$} in $O(n^2 \log n)$ time, which deteriorates as we consider other more relaxed and related (albeit still rather restrictive) notions of validity.

In \Subsection~\ref{subsec:sed_1d}, we consider the more involved problem of finding the optimal translation for \salstd{$S$,$T$,$M$} under the same strictest notion of validity.
We show that the problem can be solved in $O(n^2\alpha(n)\log n)$ time.

We also show, in \Subsection~\ref{subsec:diameter}, that we can always find a valid translation $\vec{v}$ for which the radius of the smallest enclosing disc of $D(S) \cup D(T + \vec{v})$ is at most $O(n)$ times the sum of the radii of the smallest enclosing discs of $D(S)$ and of $D(T)$ (clearly, the sum of the radii is asymptotically optimal).
The factor $O(n)$ reduces to a constant if every pair of discs of $D(S)$,
and every pair of discs of $D(T)$ are separated by at least some distance $\eps$ (where the above constant depends on $\eps$).

Finally, in \Subsection~\ref{subsec:implement}, we present experimental results of an implementation of the heuristic algorithm, for the most restrictive notion of validity, and show that it performs well in practice.
The algorithm solves unlabeled instances with hundreds of discs, of several different input types, in seconds.

We conclude the \paper with \Section~\ref{section:conclusion}, where we discuss our work and pose several open problems for further research.

\ifthesis
The results of the thesis are presented, in a compact form, at the 14th International Workshop on the Algorithmic Foundations of Robotics, 2020.
A full version is also available at \cite{arxiv}.
\fi

%% file: chapters/related_work.tex
\section{Related Work}
\vspace{-10pt}
Reconfiguration problems stand at the base of many algorithmic problems and have many different applications.
Abellanas et al.~\cite{DBLP:journals/comgeo/AbellanasBHORT06} consider it as a measure for the distance between various configurations, similarly to measuring the difference between two strings of text by their edit distance.
Some works~\cite{DBLP:journals/comgeo/AbellanasBHORT06,in-the-plane} consider the reconfiguration problem as a simplified version of multi-robot motion planning, such that the robots need to move, one at a time, from the start configuration to the target configuration, and there are no other obstacles but the robots themselves.
Another application of the reconfiguration problem is to move large objects in a warehouse~\cite{in-the-plane,DBLP:journals/siamcomp/DemaineFKMS19} --- one is clearly interested in minimizing the number of moves each object has to perform.
From a different perspective, reconfiguration problems (mainly on graphs) can abstract combinatorial puzzle games, for example the 15-puzzle~\cite{johnson1879notes}.

The typical workspace of reconfiguration problems is either the plane, three-dimensional space or discrete domains like (infinite) graphs and grids.
The simplest, and most widely studied, kind of moving objects are discs (coins), but other objects can be found in the literature such as segments~\cite{sliding, on-translating-a-set-of-rectangles},
rectangles or squares~\cite{moving-rectangles, on-translating-a-set-of-rectangles},
general convex objects~\cite{in-the-plane},
pseudodiscs~\cite{lifting} and chips on a graph~\cite{graphs-and-grids}.

Most papers consider three versions of the discs reconfiguration problem, similar to the versions discussed in the introduction:
\begin{description}
	\item[Unlabeled version] where the discs in both configurations are congruent (or isothetic) and indistinguishable, so each start disc can occupy any target disc,
	\item[Labeled version] where the discs are congruent but distinguishable by labels, so that each start disc has to occupy the target disc with the corresponding label, or
	\item[Arbitrary radii version] where the discs are not congruent.
	It is assumed that for each start disc there is at least one target disc of the same size.
	If there is more than one such disc, the start disc may occupy any of the targets of its size.
\end{description}
In the following \ssubsections, we survey three different models for the reconfiguration problem of discs in the plane. 
In each of the models, one must construct a valid itinerary (also called a schedule or a motion plan), i.e., a collision-free sequence of moves, moving the discs one-by-one, until every disc reaches one of its possible target locations.
The term ``move'' differs between the models.

Most of the papers focus on in the combinatorial aspects of the problem.
Thus, the main goal is to find an itinerary with minimal number of moves.
We present lower and upper bounds on the number of moves in each model.
For the lower bound, an input instance (construction) for the reconfiguration problem is provided such that every valid itinerary has to have at least some number of moves to perform the reconfiguration task.
For the upper bound, an algorithm is presented that outputs a valid itinerary with a bounded number of moves for every input instance.

We start by examining the most permissive model, the lifting model, in which each move consists of lifting a disc (in the third dimension, say), and placing it in a free position.
Then we examine the sliding model, in which the discs slide along continuous curves.
We finish with the model which is closest to our work, the translating model, in which the discs are moved along straight-line trajectories.
The results for the lower and upper bounds in each model and version are summarized in Table~\ref{tab:bounds}.

\begin{table}[h]
	\centering
	\caption{\sf Summary of the lower and upper bounds for the different models and versions.
		A similar table can be found in~\cite{d-mp-2013}.}
	\label{tab:bounds}
	\begin{tabular}{||c | c | c  | c||} 
		\hline
		Model & Version & Lower bound & Upper bound \\ [0.5ex] 
		\hline\hline
		Lifting & Unlabeled & $n +\Omega(n^{1/2})$ & $n + O(n^{2/3})$ \\ 
		& Labeled/Arbitrary radii & $\floor{5n/3}$ & $9n/5$ \\ 
		\hline
		Sliding & Unlabeled & $(1+\frac{1}{15})n -O(\sqrt{n})$ & $\frac{3n}{2}+O(\sqrt{n \log n})$ \\ 
		& Labeled & $\floor{5n/3}$ & $2n$ \\ 
		& Arbitrary radii & $2n - o(n)$ & $2n$ \\ 
		\hline
		Translating & Unlabeled & $\floor{5n/3}-1$ & $2n-1$ \\ 
		& Labeled/Arbitrary radii & $2n$ & $2n$ \\ 
		\hline
	\end{tabular}
\end{table}

Note that other models were also considered in the literature, each with its own restrictions.
In~\cite{DBLP:journals/corr/cs-DM-0204002} for example, a move is considered to be a placement of a disc in a grid, such that it has to touch at least two other discs.
In~\cite{graphs-and-grids}, a move consists of shifting one chip from its current vertex to a different vertex on a graph, such that any intermediate vertex on the path is not occupied by other chips.

Since translating a disc along a straight-line trajectory is also a move along a continuous curve, and also admits a valid lifting move, any lower bound for the lifting model is also a lower bound for the sliding model, and any lower bound for the sliding model is also a lower bound for the translating model.
Similarly, any upper bound for the translating model is an upper bound for the sliding model, and any upper bound for the sliding model is also an upper bound for the lifting model.
Moreover, any lower bound for the labeled version is a lower bound for the arbitrary radii version (by changing the radii of the congruent discs by some generic infinitesimal factors),
and any upper bound for the arbitrary radii version is an upper bound for the labeled version.

Before we discuss each of the three models, it is worth mentioning a simple algorithm, which is usually referred to as the \emph{universal algorithm}, that performs $2n$ moves in every model and version.
The main idea of the algorithm is to first move all the discs into some free area in the plane (that does not contain any target disc), and then move each disc to its target (or one of its targets if there are several thereof).
This is trivial to implement in the lifting model, and fairly easy in the sliding model, by moving the discs according to the lexicographic order of their centers in a suitable coordinate frame; see~\cite{sliding}.
Moreover, the universal algorithm also holds for convex objects (in the sliding model, rotations are necessary for the second step).
In the translating model the universal algorithm also applies, but the details are more involved, and are given in \Subsection~\ref{subsec:translating}.

If we assume that no start disc coincides with its target placement, then each disc has to move at least once.
On the other hand, the universal algorithm provides us with an upper bound of $2n$ moves.
Thus, the bounds considered for the reconfiguration problems of all the kinds mentioned above are between $n$ and $2n$ moves.

\subsection{The Lifting Model}
The lifting model was first presented by Bereg and Dumitrescu~\cite{lifting}.
Recall that in this model, each move consists of lifting one disc and placing it back on the plane, such that it does not intersect any other disc in the current configuration.
Bereg and Dumitrescu establish the following theorems for bounding the number of moves.

\subsubsection{Unlabeled Version}
\begin{theorem}
	\textup{\cite{lifting}}
	Given a pair of start and target configurations $S$ and $T$, each with $n$ congruent unlabeled discs, one can move the discs from $S$ to $T$ using $n + O(n^{2/3})$ moves in the lifting model.
	The entire itinerary can be computed in $O(n \log n)$ time.
	On the other hand, for each $n$, there exist pairs of configurations of $n$ discs each, which require $n +\Omega(n^{1/2})$ moves for this task.
\end{theorem}
\noindent\emph{Upper bound:}
Bereg and Dumitrescu first claim that any set of $k$ discs can be separated by a line $\ell$ such that $\ell$ intersects $O(\sqrt{k})$ discs and each halfplane bounded by $\ell$ contains at most $\frac{2}{3}k$ centers of the discs.
Hence, the plane can be partitioned into $O(n^{1/3})$ convex polygonal regions by recursively applying the above dissection, such that each convex region fully contains between $\ceil{n^{2/3}}$ to $3\ceil{n^{2/3}}$ start and/or target discs.
The total number of discs intersecting the boundaries of all regions is $O(n^{2/3})$.
In order to come up with an efficient itinerary, the regions are sorted by the difference between the number of target discs and the number of start discs fully contained in them.
This way, we can shift the start discs from one region to the target discs in a preceding region according to the descending order of this difference.
If we wish to shift the discs according to this order, we first need to make sure that there are no discs on the boundaries of the regions, so that target discs which are fully contained in a region are not blocked by start discs which are on its boundary.
Moreover, we need to move away sufficient number of start discs from the first regions, so there can be enough free targets for the first shifting step.
Then, it is easy to move start discs to free target discs of a different region (not on the boundaries) and eventually place the discs that were initially moved to the free leftover target discs.

\noindent\emph{Lower bound:}
See Figure~\ref{fig:lifting_unlabeled_lower}.
Let $m\approx\sqrt{n}$ and densely place $m\times m$ target discs in a grid, and densely place $(m-1)\times (m-1)$ start discs in a similar grid, so that each such start disc overlaps with 4 target discs.
Place $2m-2$ start discs around the grid, so that each such start disc overlaps with two target discs; see Figure~\ref{fig:lifting_unlabeled_lower}.
Place the remaining $n-m^2$ target discs and the remaining $n-m^2+1$ start discs along a horizontal line, so that each such target disc overlaps with two of these start discs.
It is enough to show that in order to occupy $m-1$ target discs of distinct rows (or of distinct columns) of the grid, at least $\frac{m-1}{2}$ start discs needs to move away (not to a target disc).
The proof is quite technical; see~\cite{lifting} for more details.
We conclude that every valid itinerary for this construction has to perform at least $n+\Omega(n^{1/2})$ moves.

\begin{figure}
	\centering
	\includegraphics[trim=0 0 0 0, clip,width=0.8\textwidth]{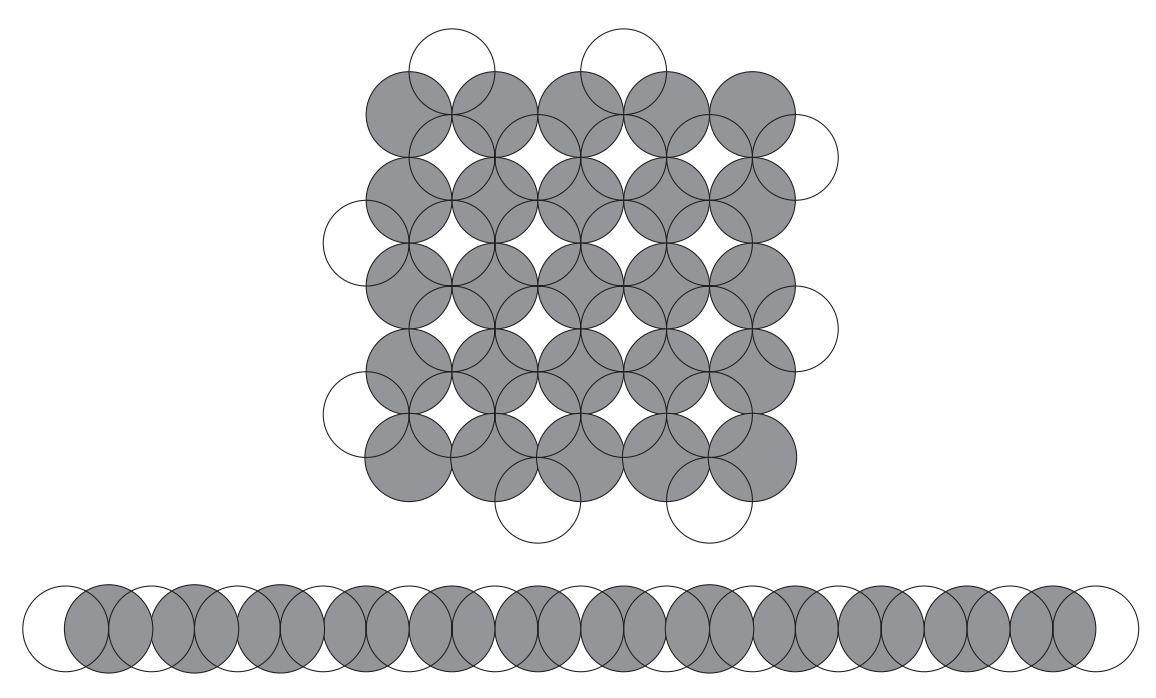}
	\caption{\sf 
		Lower bound construction for the unlabeled version of the lifting model.
		The start configuration (as empty discs) and target configuration (as shaded discs) in a grid like construction.
		Figure taken from~\cite{lifting}.
	} \label{fig:lifting_unlabeled_lower}
\end{figure}

\subsubsection{Labeled Version}

\begin{theorem}
\textup{\cite{lifting}}
Given a pair of start and target configurations $S$ and $T$, each with $n$ discs with arbitrary radii, $9n/5$ moves always suffice for transforming the start configuration into the target configuration under the lifting model.
On the other hand, for each $n$, there exist pairs of configurations which require $\floor{5n/3}$ moves for this task.
\end{theorem}
\noindent\emph{Upper bound:}
Bereg and Dumitrescu first prove that for any directed graph $G=(V,E)$,
$$\beta(G) \geq \max\left(\sum\limits_{v \in V}\frac{1}{d^+_v + 1}, \sum\limits_{v\in V}\frac{1}{d^-_v+1} \right),$$
where $\beta(G)$ is the maximum-size subset of vertices $V' \subseteq V$, such that the subgraph induced by $V'$ is acyclic, and $d^+_v$ (resp., $d^-_v$) is the out-degree (resp., in-degree) of $v$.
Moreover, they show that the intersection graph of the open discs of $S \cup T$ is planar.
Using the above claims, they prove that $\beta(D) \geq \frac{n}{5}$ for the directed \emph{blocking graph} $D=(S,F)$ on the set $S$ of $n$ start discs, where
$$F=\set{(s_i,s_j):i\neq j \text{ and } s_i \cap t_j \neq \emptyset, s_i,s_j \in S, t_j \in T}.$$
In simpler words, $(s_i,s_j)\in F$ if disc $s_i$ blocks disc $s_j$ from moving to its target.
Thus, there exists a topological ordering of $S' \subseteq S$ of size at least $\frac{n}{5}$, such that each of the discs in $S'$ can be moved to its target (according to the ordering) without intersecting any other disc of $S'$.
That said, we can initially move the discs in $S''=S\setminus S'$ to a free area, then move the discs in $S'$ to their targets according to the topological ordering, and eventually move the discs in $S''$, resulting in a valid itinerary comprising $9n/5$ moves.

\noindent\emph{Lower bound:} consider the construction in Figure~\ref{fig:lifting_labeled_lower}.
By repeating the construction a sufficient number of times, we get the lower bound of $\floor{5n/3}$ moves.

\begin{figure}
	\centering
	\includegraphics[trim=0 0 0 0, clip,width=0.2\textwidth]{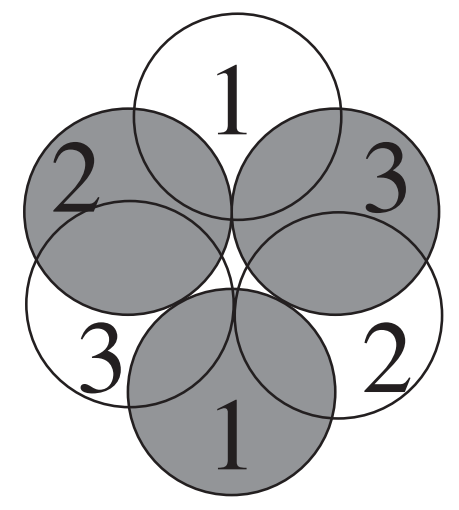}
	\caption{ \sf
		Lower bound construction for the labeled version of the lifting model.
		The start and target configurations are displayed as empty and shaded discs, respectively. Every valid itinerary for the labeled (or arbitrary radii) version of the lifting model has to perform at least $5$ moves.
		Figure taken from~\cite{lifting}.
	} \label{fig:lifting_labeled_lower}
\end{figure}

\medskip
\noindent{\bf Remark.}
The bounds for the labeled version also hold for pseudodiscs.

\subsection{The Sliding Model}
\label{subsec:sliding}

The sliding model was introduced in~\cite{sliding} by Bereg, Dumitrescu and Pach.
In this model, a disc is moved so its center follows a continuous curve in the plane, while the interior of the disc does not intersect any other unmoved disc throughout the motion.
Bereg et al.\ proved the following bounds.

\subsubsection{Unlabeled Version}

\begin{theorem}
\textup{\cite{sliding}}
Given a pair of start and target configurations $S$ and $T$, each consisting of $n$ congruent discs, $\frac{3n}{2} + O(\sqrt{n \log n})$ moves always suffice for transforming the start configuration into the target configuration under the sliding model.
The entire motion can be computed in $O(n^{3/2}(\log n)^{1/2})$ time.
On the other hand, there exist pairs of configurations that require $(1+\frac{1}{15})n - O(\sqrt{n})$ moves for this task, for arbitrarily large $n$.
\end{theorem}
\noindent\emph{Upper bound:}
Bereg et al.\ first prove that there exists a strip of width 6 such that $\ell$, the midline of the strip, bisects the set of the centers of $S$, and the strip fully contains at most $O(\sqrt{n \log n})$ discs.
Assume without loss of generality that $\ell$ is the $y$-axis and (by symmetry) that 
the centers of target discs to the right of $\ell$, denoted by $T_R$, is of size at least $n/2$.
Devising an itinerary is fairly simple: first move horizontally far to the right, all the start discs that are fully contained in the strip or to its right.
Note that the discs centered in $T_R$ are now free.
Move the remaining start discs to the left of $\ell$ (at most $n/2$) to the leftmost target discs centered in $T_R$.
Finally, bring back the initially moved discs to the rest of $T_R$ or to the targets left of $\ell$ (which are now free).
Since the discs of the first step were moved twice, and the rest of the discs just once, the total number of moves is $\frac{3n}{2}+O(\sqrt{n \log n})$.

\begin{figure}
	\centering
	\includegraphics[trim=0 0 0 0, clip,width=0.5\textwidth]{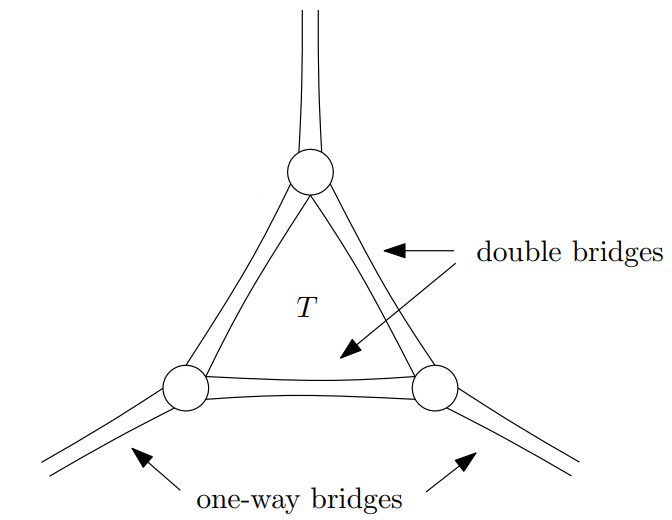}
	\caption{\sf
		Lower bound construction for the unlabeled version of the sliding model.
		The start configuration is positioned in a cage-like construction, where the target configuration $T$ is positioned in its center.
		Figure taken from~\cite{sliding}.
	} \label{fig:sliding_unlabeled_lower}
\end{figure}

\noindent\emph{Lower bound:}
The main idea is to create a cage-like construction, such that the target configuration is a set of  densely packed discs contained in a square of side length $\approx 2 \sqrt{n}$.
The start configuration is a set of discs that surround the target configuration and so it takes several moves until any start disc can move to some target.
A simple example for that will be $\Theta(\sqrt{n})$ concentric rings, each containing $\Theta(\sqrt{n})$ targets.
Before one can move a start disc to some target, at least one disc from each ring has to move away from the center of the ring, yielding a lower bound of $n + \Omega(\sqrt{n})$.
Bereg et al.\ present a more sophisticated cage; see Figure~\ref{fig:sliding_unlabeled_lower}.
In order to construct it, they present two main components: \emph{bridges}, which are arbitrarily long finite packings consisting of five rows of discs, and \emph{junctions}, which are other finite packings of a constant number of discs, and which are designed to connect between three bridges.
When these components are juxtaposed as in Figure~\ref{fig:sliding_unlabeled_lower}, the only discs that can move are the ones at the end of a bridge (four per such a bridge).
Hence, in order to enable a movement of one start disc to some target placement, one needs to move at least the number of discs of one row of a bridge, which is $\frac{n}{15}-O(\sqrt{n})$ long.
We arrive at a total of at least $(1+\frac{1}{15})n -O(\sqrt{n})$ moves, as claimed.

\subsubsection{Arbitrary Radii Version}
\label{subsec:sliding_arbitrary}

\begin{theorem}
	\textup{\cite{sliding}}
	Given a pair of start and target configurations, each consisting of $n$ discs of arbitrary radii, $2n$ moves always suffice for transforming the start configuration into the target configuration under the sliding model.
	The entire motion can be computed in $O(n \log n)$ time.
	On the other hand, there exist pairs of configurations that require $2n - o(n)$ moves for this task, for every sufficiently large $n$.
\end{theorem}

\noindent\emph{Upper bound:} 
Follows from the universal algorithm.

\noindent\emph{Lower bound:} see Figure~\ref{fig:sliding_labeled_lower}.
At each level of the recursive construction, each disc is surrounded by $2m+1$ smaller discs.
Note that in order to move a disc, one has to recursively move a fixed number of discs from the set of smaller discs that surrounds it (unless it is a disc of the smallest size).
Hence, if there are $k$ levels in the recursion, about $n + n/2 + n/4 + \dots + n/2^k$ moves are necessary.

\begin{figure}
	\centering
	\includegraphics[trim=0 0 0 0, clip,width=0.75\textwidth]{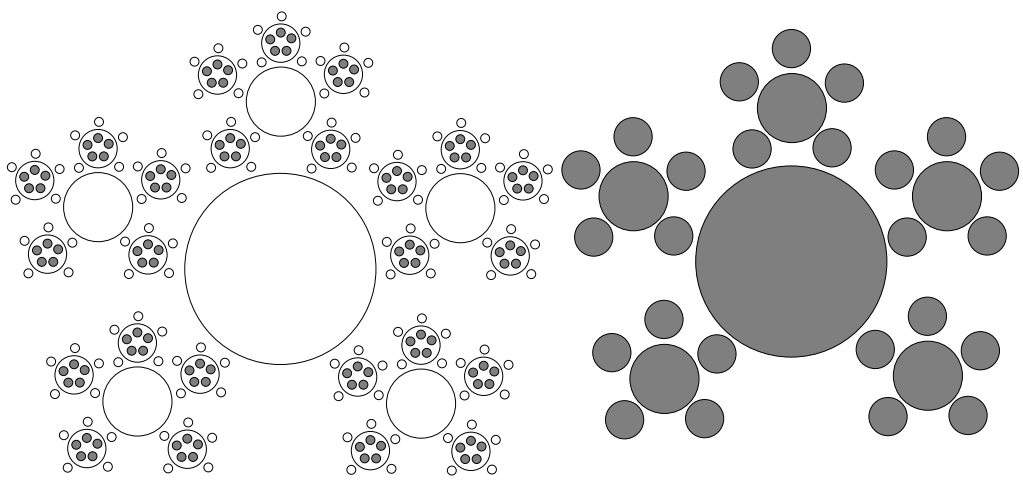}
	\caption{\sf
		Lower bound construction for the arbitrary radii version of the sliding model.
		The start configuration (as empty discs) and target configuration (as shaded discs) positioned in a recursive construction, where $m=2$ and $k=3$.
		Figure taken from~\cite{sliding}.
	} \label{fig:sliding_labeled_lower}
\end{figure}

\subsubsection{Labeled Version}
Clearly, the same construction for the lower bound of the arbitrary radii version in \Subsection~\ref{subsec:sliding_arbitrary} cannot hold for congruent discs.
As noted before, the lifting model can provide a lower bound of $\floor{5n/3}$ for the sliding model (see Figure~\ref{fig:lifting_labeled_lower}).

\subsubsection{Hardness results}
Dumitrescu and Jiang~\cite{in-the-plane} define the following decision problem.
Given a reconfiguration problem in the sliding model, and a positive integer $k$, determine if there is a valid itinerary of at most $k$ moves.
They denote the labeled (resp., unlabeled) version of the problem as \rm{L-SLIDE-RP} (resp., \rm{U-SLIDE-RP}).
They prove that both problems, \rm{L-SLIDE-RP} and \rm{U-SLIDE-RP}, are NP-hard, using reductions from the Rectilinear Steiner Tree problem, which is known to be strongly NP-complete.
The general idea is to show that a valid itinerary of $k$ moves is equivalent to a Steiner tree of length proportional to $k$.
The reduction can be done by centering a target disc on every point of the Steiner Tree problem and placing equal amount of start discs somewhere to the left of the first target disc.
By a polynomial number of obstacles (discs that coincide with their targets), which fills most of the working space, the construction forms a grid like environment that ensures that moving one disc is equivalent to one unit length of the Steiner tree.
This idea is an adaptation of a similar reduction from~\cite{graphs-and-grids} regarding chips on an infinite grid.

\subsection{The Translating Model}
\label{subsec:translating}
Although presented last, the translating model was the first to be considered in the reconfiguration problems realm, by Abellanas et al.~\cite{DBLP:journals/comgeo/AbellanasBHORT06}.
Recall that similarly to our work, every start disc can move along a straight line only (but discs can move more than once).
Other than the bounds on the number of moves of a valid itinerary without any space restrictions, Abellanas et al.\ also establish bounds under some confining assumptions on the problem.
Three different confining spaces are considered, where the discs and their motions are restricted to a given $a \times b$ bounding box:
\begin{description}
	\item[Unlabeled / Narrow:] $a \geq n$ and $b \geq 1$.
	\item[Arbitrary radii / Wide:] $a,b \geq D$ where $D$ is the sum of the diameters of the start discs.
	\item[Unlabeled / TooTight:]
	The given bounding box might be too small to admit any valid itinerary.
	Thus, the motions are restricted to a bounding box of size $\ceil{a} \times (b + \ceil{n/a})$, while the discs, in their start and target positions, are restricted to the given $a \times b$ box.
\end{description}
The bounds for the above assumptions are summarized in Table~\ref{tab:confining_bounds}.
\begin{table}[h]
	\centering
	\caption{\sf
	Bounds of the number of moves under different confining assumptions for the translating model.}
	\label{tab:confining_bounds}
	\begin{tabular}{||c | c | c  | c||} 
		\hline
		Version & Confining assumption & Lower bound & Upper bound \\ [0.5ex] 
		\hline\hline
		Unlabeled & Narrow & $\floor{8n/5}$ & $3n$ \\
		Arbitrary radii & Wide & $2n$ & $4n$ \\
		Unlabeled & TooTight & $\floor{8n/5}$ & $6n$ \\
		\hline
	\end{tabular}
\end{table}

\subsubsection{Labeled Version}
$2n$ moves are sufficient for any instance of the labeled version using the universal algorithm.
However, its implementation in the translating model is not straightforward, and so we present the key ideas that were proved in~\cite{DBLP:journals/comgeo/AbellanasBHORT06}.
Even for the version of arbitrary radii discs, Abellanas et al.\ prove that a set of discs can be translated, one by one according to their reversed lexicographic order, such that their centers are positioned along a (sufficiently far to the right and up) vertical segment, in any desired order along the segment.
Using this observation twice, once for moving the start discs to the vertical segment, and the second for ``moving'' the target discs to the vertical segment (moving the target discs is equivalent to reversing the process and moving the discs from their intermediate positions along the vertical segment to their target positions), yields the desired itinerary.
On the other hand, there exist instances that require $2n$ moves, by repeating, for example, the construction in Figure~\ref{fig:labeled_infeasible}, given in the next \ssection.

\subsubsection{Unlabeled Version}
Similarly to the labeled version, the upper bound of $2n-1$ moves in the unlabeled version is proven by using the universal algorithm with a minor change.
This time we can move the leftmost start disc directly to the leftmost target disc, so the total number of moves is reduced by 1.

\begin{figure}
	\centering
	\includegraphics[trim=0 0 0 0, clip,width=0.9\textwidth]{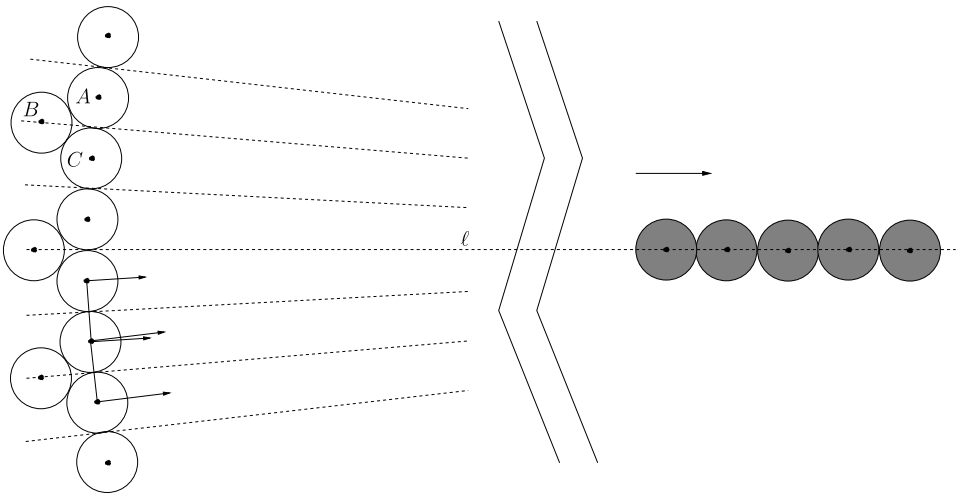}
	\caption{\sf
		Lower bound construction for the unlabeled version of the translating model.
		The start discs are placed on two concentric rings on the left, while the target discs are placed along a horizontal segment to the right.
		Figure taken from~\cite{in-the-plane}.
	} \label{fig:translating_unlabeled_lower}
\end{figure}

As for the lower bound, Abellanas et al.\ have shown that $\floor{8n/5}$ moves are necessary for some given construction.
Using a similar construction, Dumitrescu and Jiang~\cite{in-the-plane} managed to show a better lower bound of $\floor{5n/3}-1$ moves; see Figure~\ref{fig:translating_unlabeled_lower}.
In this construction, the start discs are centered on two concentric rings and the target discs are placed in a horizontal chain far to the right.
The rings are large enough so that for each two tangent discs placed on the inner ring, the open strip of width 2, where the perpendicular bisector of the centers of the discs runs in the middle of the strip, contains all the centers of the target discs.
This property ensures that for each pair of adjacent discs of the first layer (the inner ring), at least one has to perform an extra move before reaching a target.
Using the above observation, it is shown that any valid itinerary has to execute at least twice the number of discs in the first layer (minus 1) non-target moves, i.e., moves that their destination is not a target disc.

\subsubsection{Hardness results}
As already mentioned in the introduction, some problems in this model are hard to solve.
Recall the problems defined in \Subsection~\ref{subsec:sliding}.
The problems \rm{U-TRANS-RP} and \rm{L-TRANS-RP} are defined similarly, and are also proven in~\cite{in-the-plane} to be NP-hard.
The \rm{3-SET-COVER} problem can be reduced to both of the problems, using adaptations of the ideas found in~\cite{graphs-and-grids}, regarding chips on a ``broom''-shaped graph.

On the bright side, deciding whether there exists any valid itinerary for the labeled version can be done in polynomial time.
The proof can be found in~\cite{DBLP:journals/comgeo/AbellanasBHORT06} and in \Section~\ref{section:labeled}.

%% file: chapters/labeled.tex
\ifthesis
\section[Labeled Version: Analysis of the Translation Plane]%
{Labeled Version: \\ Analysis of the Translation Plane}
\vspace{-10pt}
\else
\section{Labeled Version: Analysis of the Translation Plane}
\fi
\label{section:labeled}

In this \ssection we consider the labeled version \lst{} of the problem.
First, in \Subsection~\ref{subsec:labeled_analysis}, we formally define the problem.
We then present an algorithm that determines whether there is a valid itinerary for a fixed initial translation $\vec{v}$, and computes such an itinerary if it exists.
We then compute, in \Subsection~\ref{subsec:labeled_Q}, the region $Q$ of all initial translations for which there exists a valid itinerary.
Finally, we present the data structure $D_Q$, that is used to query $Q$ efficiently, in preparation for the optimization algorithms in \Section~\ref{section:labeled_sa}.

\subsection{Preliminary Analysis}
\label{subsec:labeled_analysis}
\vspace{-5pt}

We are given two valid configurations $S$ and $T$ of $n$ points each, and a one-to-one matching $M$ between the positions of $S$ and those of $T$, which is the set of pairs $\set{(s,M(s))\mid s \in S}$, where $s$ and $M(s)$ share the same label, for each $s \in S$.
Our goal is to find a translation $\vec{v} \in \R^2$ such that there is a valid collision-free itinerary of $n$ unit discs from $S$ to $T + \vec{v}$ with respect to the matching $M$.
That is, the goal is to define an ordering on the elements of $M$, denoted by $(s_1, M(s_1)), (s_2, M(s_2)), \ldots, (s_n,M(s_n))$,
so that, for each $i=1,\ldots,n$ in this order, we can translate the disc placed at $s_i$ to the position $M(s_i) + \vec{v}$,
so that it does not collide with any still unmoved discs, placed at $s_{i+1}, \ldots, s_n$, nor with any of the already translated discs, placed at $M(s_1)+\vec{v}, \ldots, M(s_{i-1}) + \vec{v}$.

We call a translation $\vec{v}$ a \emph{valid translation} if it yields at least one valid itinerary.
In the labeled version, we show how to compute the set of all valid translations in $O(n^6)$ time.
We then present, in \Section~\ref{section:labeled_sa}, three algorithms, each of which finds a valid translation $\vec{v}$ that minimizes a different measure of proximity between $S$ and $T+\vec{v}$, as reviewed in the introduction.

We first address the subproblem in which $\vec{v}$ is fixed and our goal is to order $M$ so as to obtain a valid itinerary, if at all possible.
Let $A=(s,M(s))$ be a pair in the matching.
For convenience, we denote $s$ and $M(s)$ by $A^S$ and $A^T$, respectively.
Define the \emph{hippodrome} of two unit discs $D, D'$ to be the convex hull of their union.
Observe that the hippodrome is exactly the area that a unit disc will cover while moving from $D$ to $D'$ along a straight trajectory.
Define $\Hv(A)$ to be the hippodrome of $D(A^S)$ and $D(A^T + \vec{v})$.
Denote by $\kv$ the overall number of intersecting pairs of hippodromes $\set{\Hv(A), \Hv(B)}$, for all $A\neq B \in M$.
See Figure~\ref{fig:hippodromes} for an illustration.

\begin{figure}
  \centering
    \includegraphics[trim=150 10 140 65, clip,width=0.8\textwidth]{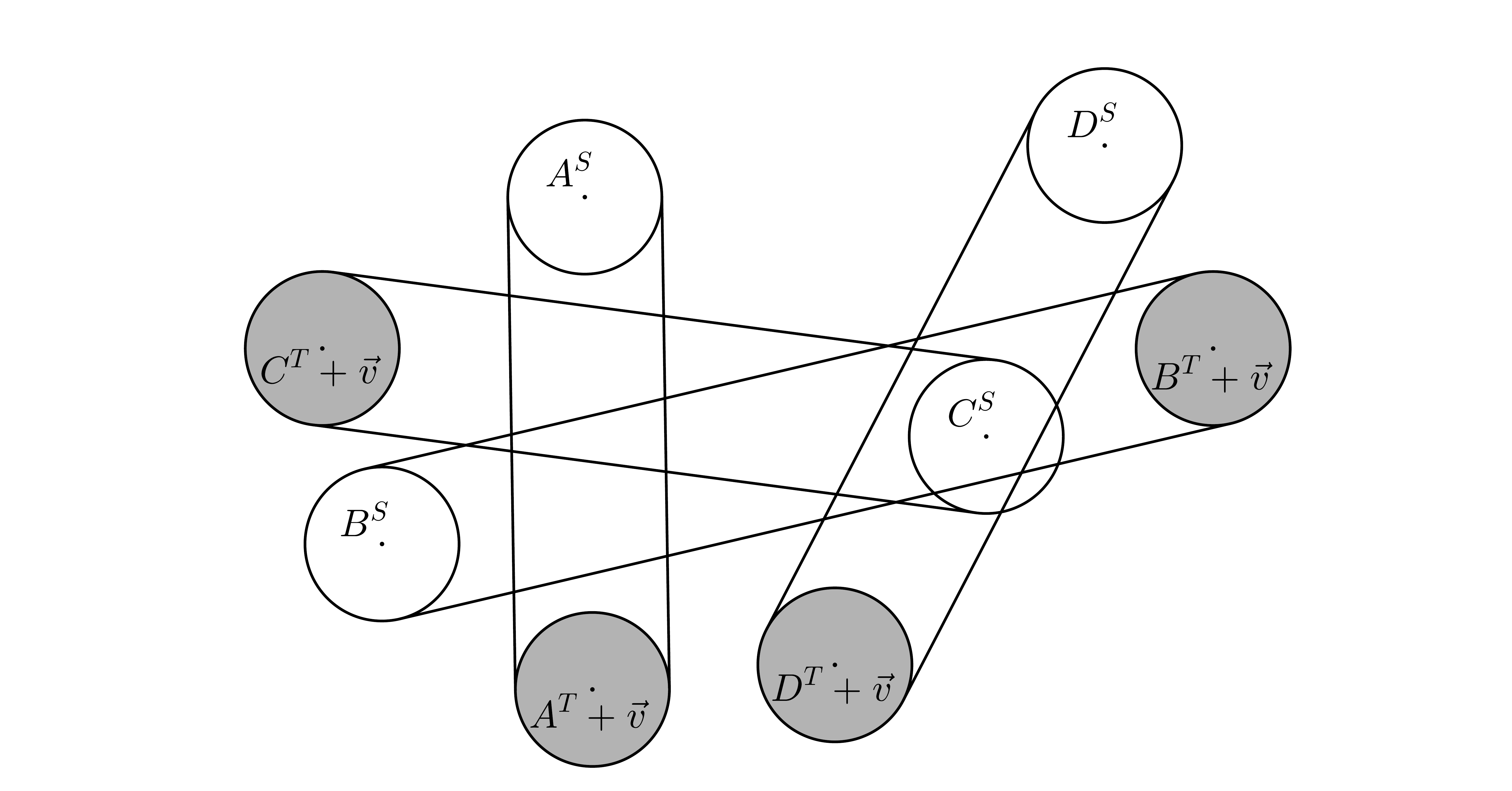}
    \caption{
       \sf The hippodromes $\Hv$ for four pairs $A,B,C,D \in M$ and some fixed $\vec{v}$. Notice that even though ${\Hv(A) \cap \Hv(B) \neq \emptyset}$, there is no restriction that the motion of $A$ must precede or succeed the motion of $B$. Such restrictions do exist for many other pairs, such as $B$ and $C$ ($C$ has to perform a motion before $B$).
    } \label{fig:hippodromes}
\end{figure}

\begin{theorem}[Abellanas et al.~\cite{DBLP:journals/comgeo/AbellanasBHORT06}]
\label{theo:abellanas}
Let $S$ and $T$ be two valid configurations of $n$ points each, and let $\vec{v}$ be a fixed translation.
Let $M:S \rightarrow T$ be a bijection between the two configurations.
Then one can compute, in $O(n \log n + \kv)$ time, a valid itinerary for $S$ and $T+\vec{v}$ with respect to $M$, if one exists.
\end{theorem}

We review the proof of the theorem, adapting it to our notations, and exploit later the ingredients of the analysis for the general problem (where we allow $T$ to be translated).
The constraints that the positions of the discs impose on the problem are as follows.
We say that a pair $A=(A^S, A^T)$ (in $M$) has to perform a motion (from $D(A^S)$ to $D(A^T + \vec{v})$) before another pair $B$, for a given translation $\vec{v}$, if in any ordering $\Pi$ of $M$ that yields a valid itinerary, the index of $A$ in $\Pi$ is smaller than the index of $B$.
In other words, for any two pairs $A,B \in M$, $A$ has to perform a motion before $B$ if either the disc $D(A^S)$ blocks the movement of $D(B^S)$ to the position $B^T+\vec{v}$, or the disc $D(B^T + \vec{v})$ blocks the movement of $D(A^S)$ to the position $A^T+\vec{v}$.
Formally, we have:

\begin{lemma}
\label{lemma:constraint_conditions}
Given pairs $A,B \in M$ and a fixed translation $\vec{v}$, $A$ has to perform a motion before $B$ (with respect to $\vec{v}$) if and only if at least one of the following conditions holds:
\begin{enumerate}
    \item $ D(A^S) \cap \Hv(B) \neq \emptyset. $
    \item $ D(B^T + \vec{v}) \cap \Hv(A) \neq \emptyset. $
\end{enumerate}
\end{lemma}

We next create a digraph whose vertices are the pairs of $M$, and whose edges are all the ordered pairs $(A,B) \in M^2$, for $A \neq B$, that satisfy (1) or (2).
Borrowing a similar notion from assembly planning~\cite{DBLP:journals/algorithmica/HalperinLW00}, we call the graph, for a fixed translation $\vec{v}$, the \emph{translation blocking graph} (TBG),
and denote it as $\Gv$.
Denote the number of edges in $\Gv$ as $\mv$, and observe that $\mv \leq \kv$.
Indeed, for every edge $(A,B) \in \Gv$ the hippodromes $\Hv(A)$ and $\Hv(B)$ intersect, as is easily verified, but not every pair of intersecting hippodromes necessarily induce an edge; see the pairs $A,B$ in Figure~\ref{fig:hippodromes}.
As proved in~\cite{DBLP:journals/comgeo/AbellanasBHORT06}, and as is easy to verify, the subproblem for a fixed $\vec{v}$ is feasible if and only if $\Gv$ is acyclic.

The circular arcs of a hippodrome can be split into two arcs, each of which is $x$-monotone.
This allows us to construct $\Gv$ in $O(n \log n + \kv)$ time, by the sophisticated sweep-line algorithm of Balaban~\cite{DBLP:conf/compgeom/Balaban95}, which applies to any collection of well-behaved $x$-monotone arcs in the plane. (A standard sweeping algorithm would take $O(n \log n + \kv \log n)$ time.)
Checking whether $\Gv$ is acyclic, and, if so, performing topological sorting on $\Gv$, takes $O(n+\mv)$ time.
By definition, any topological order of the vertices of $\Gv$, that is of $M$, defines a valid itinerary.
If $\Gv$ has cycles, no valid itinerary exists for $\vec{v}$.

We now consider the translation plane $\R^2$, each of whose points corresponds to a translation vector $\vec{v}$.
We say that a point $\vec{v}$ in the translation plane is valid, if the corresponding translation vector is valid, i.e., admits a valid itinerary from $S$ to $T + \vec{v}$.
We say that a set of points (in the translation plane) is valid, if each of its points is valid.
Our goal is to construct the region $Q$ of all the valid points (translations), and to partition $Q$ into maximal connected cells, so that all translations in the same cell have the same TBG.
Thus, for each cell, either all its points are valid (with the same set of common valid itineraries) or all its points are invalid.

 \begin{figure}
     \centering
     \includegraphics[trim=0 50 0 0, clip,width=0.4\textwidth]{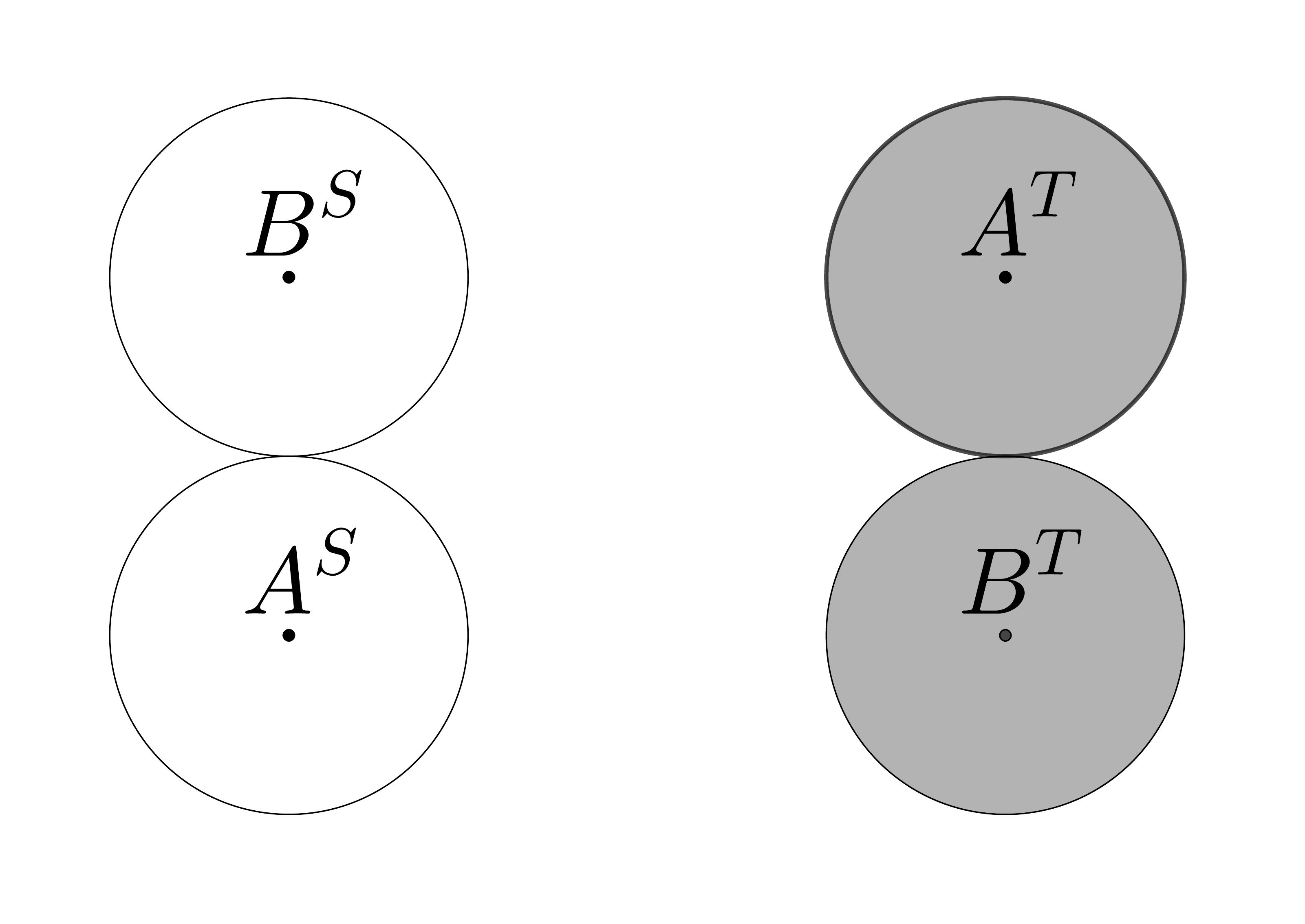}
     \caption{\sf An infeasible instance for the labeled version:
     	no valid itinerary exists between $S$ and $T+\vec{v}$, for any translation $\vec{v}$.}
     \label{fig:labeled_infeasible}
 \end{figure}

\medskip
\noindent{\bf Remark.}
For some instances, $Q$ is empty, as in the scenario depicted in Figure~\ref{fig:labeled_infeasible}.
In that case, our algorithms will report that no valid translation exists.
We also remark that tangency is not a necessary characteristic of infeasible instances, as we will shortly show.

We first fix two pairs $A,B \in M$, and consider the region $\V_{AB}$, which is the locus of those $\vec{v}$ for which the (directed) edge $AB$ is present in $\Gv$.
We can write $\V_{AB} = \V^{(1)}_{AB} \cup \V^{(2)}_{AB}$, where $\V^{(1)}_{AB}$ (resp., $\V^{(2)}_{AB}$) is the locus of all $\vec{v}$ for which condition (1) (resp., (2)) in Lemma~\ref{lemma:constraint_conditions} holds.
We thus have
\begin{align*}
    \V_{AB}^{(1)} &= \set{\vec{v} \in \R^2 \mid D(A^S) \cap \Hv(B) \neq \emptyset}, \\
    \V_{AB}^{(2)} &= \set{\vec{v} \in \R^2 \mid D(B^T + \vec{v}) \cap \Hv(A) \neq \emptyset}.
\end{align*}
We call $ \V_{AB}^{(1)}$ (resp., $ \V_{AB}^{(2)}$) the \emph{vippodrome} of $(A,B)$ of the first (resp., second) type.

\begin{figure}
    \centering
    \includegraphics[trim=0 200 0 130,clip,width=0.7\textwidth]{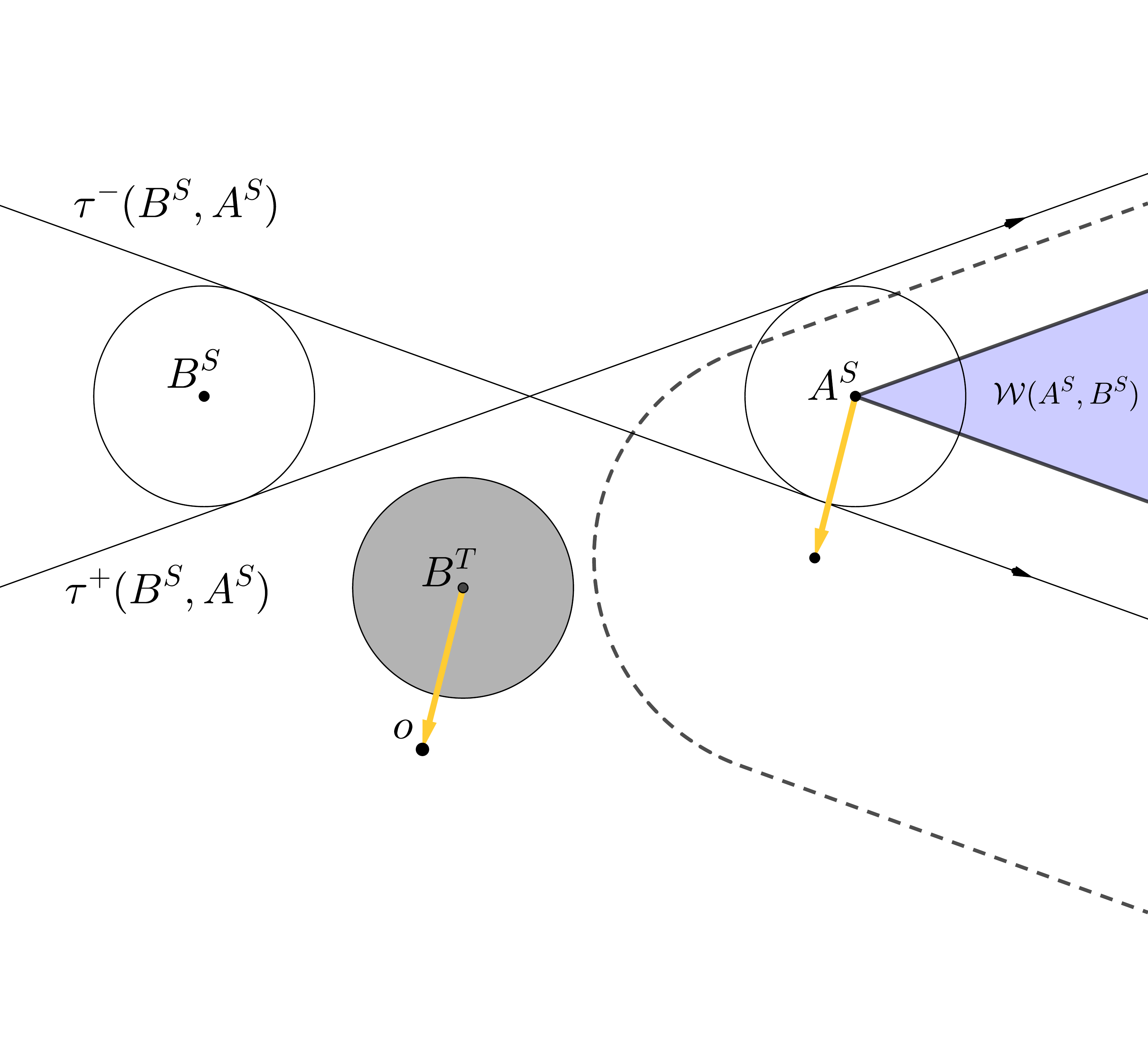}
    \caption{\sf The vippodrome $\V^{(1)}_{AB}$, which is the region to the right of the dashed curve in the translation plane.
    The wedge $\W(A^S,B^S)$ is colored in blue.
    The vippodrome is obtained by first expanding $\W(A^S,B^S)$ by $D_2(o)$, and then by shifting by the vector $-B^T$ (in orange).}
    \label{fig:vippo_construction}
\end{figure}

To construct $\V^{(1)}_{AB}$, we proceed as follows; see Figure~\ref{fig:vippo_construction}.
For given pairs $A,B \in M$, consider the two inner tangent lines, $\tau^-(B^S,A^S)$
and $\tau^+(B^S,A^S)$, to $D(B^S)$ and $D(A^S)$, and assume that they are both directed from $B^S$ to $A^S$, so that $B^S$ lies to the right of $\tau^-(B^S,A^S)$ and to the left of $\tau^+(B^S,A^S)$, and $A^S$ lies to the left of $\tau^-(B^S,A^S)$ and to the right of $\tau^+(B^S,A^S)$.
Let $\W(A^S,B^S)$ denote the wedge whose apex is at $A^S$ and whose rays are parallel to (and directed in the same direction as) $\tau^-(B^S, A^S)$ and $\tau^+(B^S,A^S)$.
Denote the origin as $o$.
We then have the following representation.

\begin{lemma}
\label{lemma:vippodrome_construction}
\begin{align} \label{vips-and-wedges}
\V^{(1)}_{AB} & = \W(A^S,B^S) \oplus D_2(o) - B^T
= (\W(A^S,B^S) - B^T) \oplus D_2(o) \\
\V^{(2)}_{AB} & = -\W(B^T,A^T) \oplus D_2(o) + A^S
= -(\W(B^T,A^T) - A^S) \oplus D_2(o). \nonumber
\end{align}
\end{lemma}
\begin{proof}
$\V^{(1)}_{AB}$ is the locus of all translations $\vec{v}$ at which $D(A^S)$ intersects $H_{\vec{v}}(B)$. 
Equivalently, $\V^{(1)}_{AB}$ is the locus of all translations $\vec{v}$ at which $D_2(A^S)$ intersects the segment $e_{\vec{v}} = {(B^S,B^T+\vec{v})}$.
The boundary of $\V^{(1)}_{AB}$ thus consists of all translations $\vec{v}$ for which either $e_{\vec{v}}$ is tangent to $D_2(A^S)$ or $B^T+\vec{v}$ touches $D_2(A^S)$.
That is, $\bd \V^{(1)}_{AB}$ consists of all translations $\vec{v}$ for which $B^T + \vec{v}$ lies on the boundary of $\W(A^S,B^S)\oplus D_2(o)$, from which the claim easily follows.
The claim for $\V^{(2)}_{AB}$ follows by a symmetric argument, switching between $S$ and $T$ and reversing the direction of the translation.
\end{proof}

Note that the boundary of a vippodrome $\V^{(1)}_{AB}$ is the smooth concatenation of two rays and a circular arc, where the rays are parallel to the rays of $\W(A^S,B^S)$, and where the arc is an arc of the disc $D_2(A^S-B^T)$, of central angle $\pi-\theta$, where $\theta$ is the angle of $\W(A^S,B^S)$.
The same holds for $\V^{(2)}_{AB}$, with the same disc $D_2(A^S-B^T)$. Hence the boundaries of $\V^{(1)}_{AB}$ and of $\V^{(2)}_{AB}$ (more precisely, the circular portions of these boundaries) might overlap. See
Figure~\ref{fig:overlap_arcs} for an illustration.

 \begin{figure}
	\centering
	\includegraphics[trim=70 150 80 100, clip, width=0.3\textwidth]{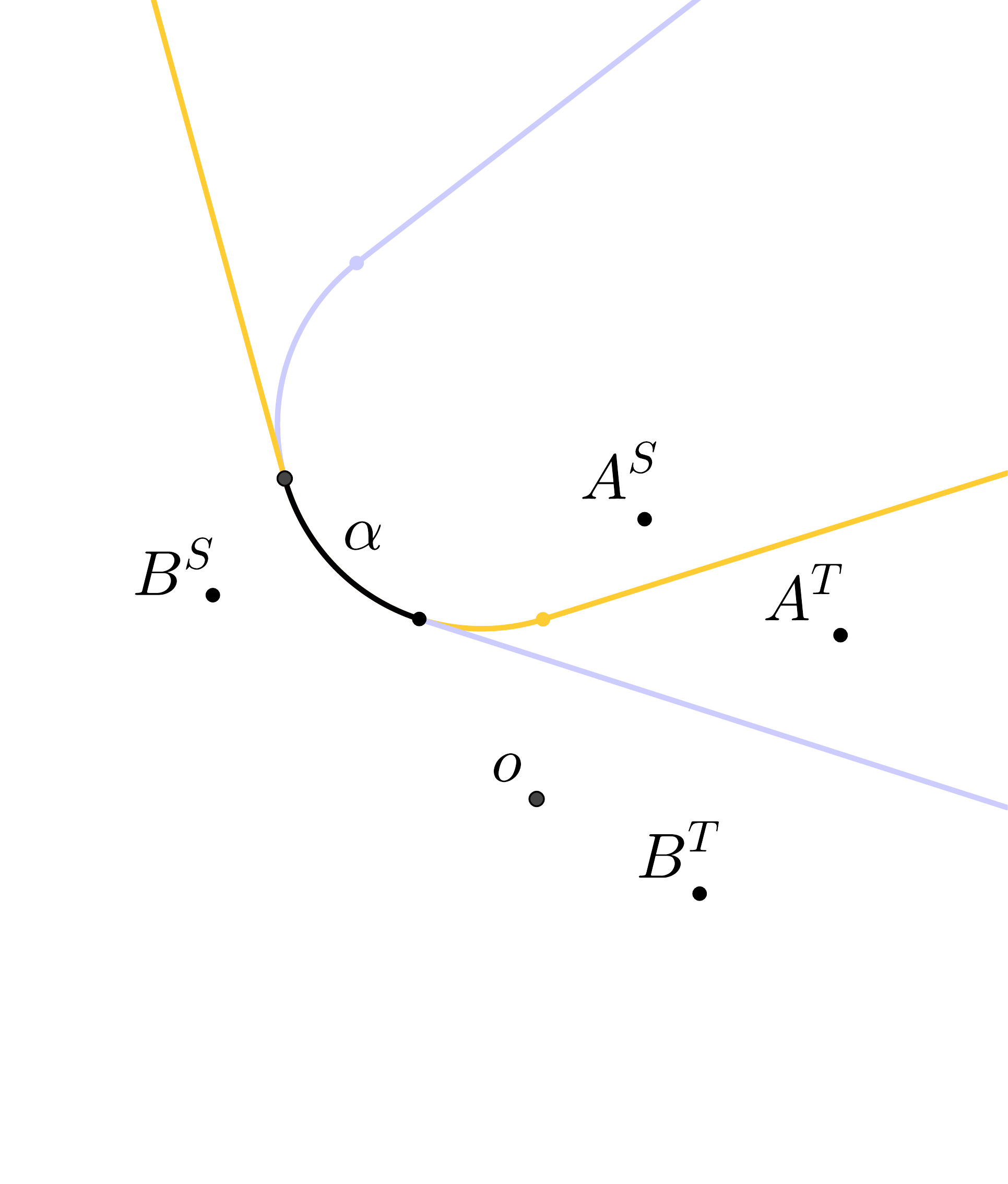}
	\caption{\sf The vippodromes $\V^{(1)}_{AB}$ and $\V^{(2)}_{AB}$, colored in blue and orange, respectively. The arc $\alpha$ is the overlap portion of the two vippodrome boundaries.}
	\label{fig:overlap_arcs}
\end{figure}

Using vippodromes, we can give, as promised, a scenario of a start and target configurations that admit no valid translation even though the discs of each configuration do not touch each other.
Such a scenario is depicted in the self-explanatory Figure~\ref{fig:labeled_infeasible_nontangent}.

\begin{figure}
	\begin{minipage}[c]{0.5\textwidth}
		\centering
		\includegraphics[trim=0 190 0 220, clip, width=0.7\textwidth]{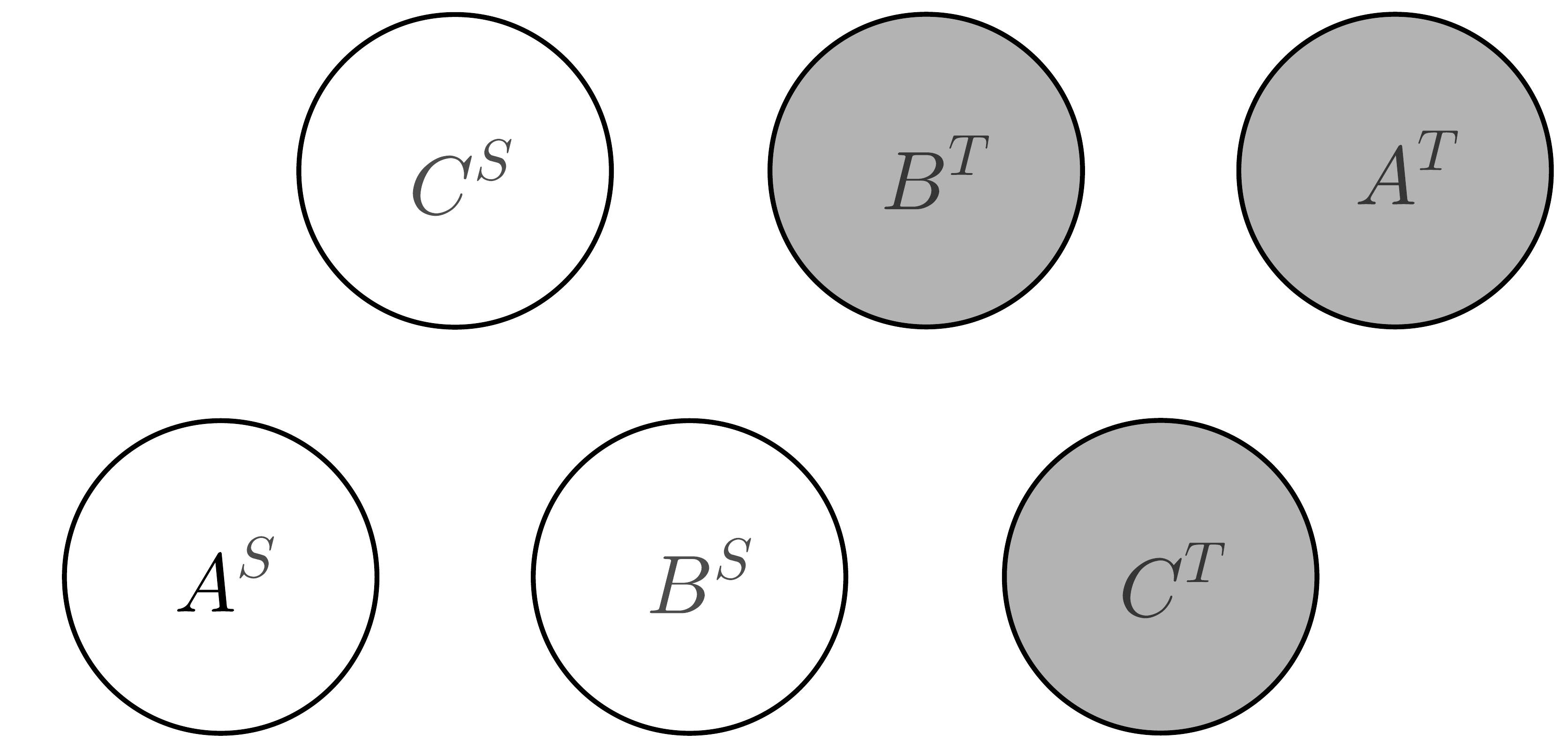}
	\end{minipage}
	\begin{minipage}[c]{0.5\textwidth}
		\centering
		\includegraphics[trim=0 0 0 0, clip, width=.7\textwidth]{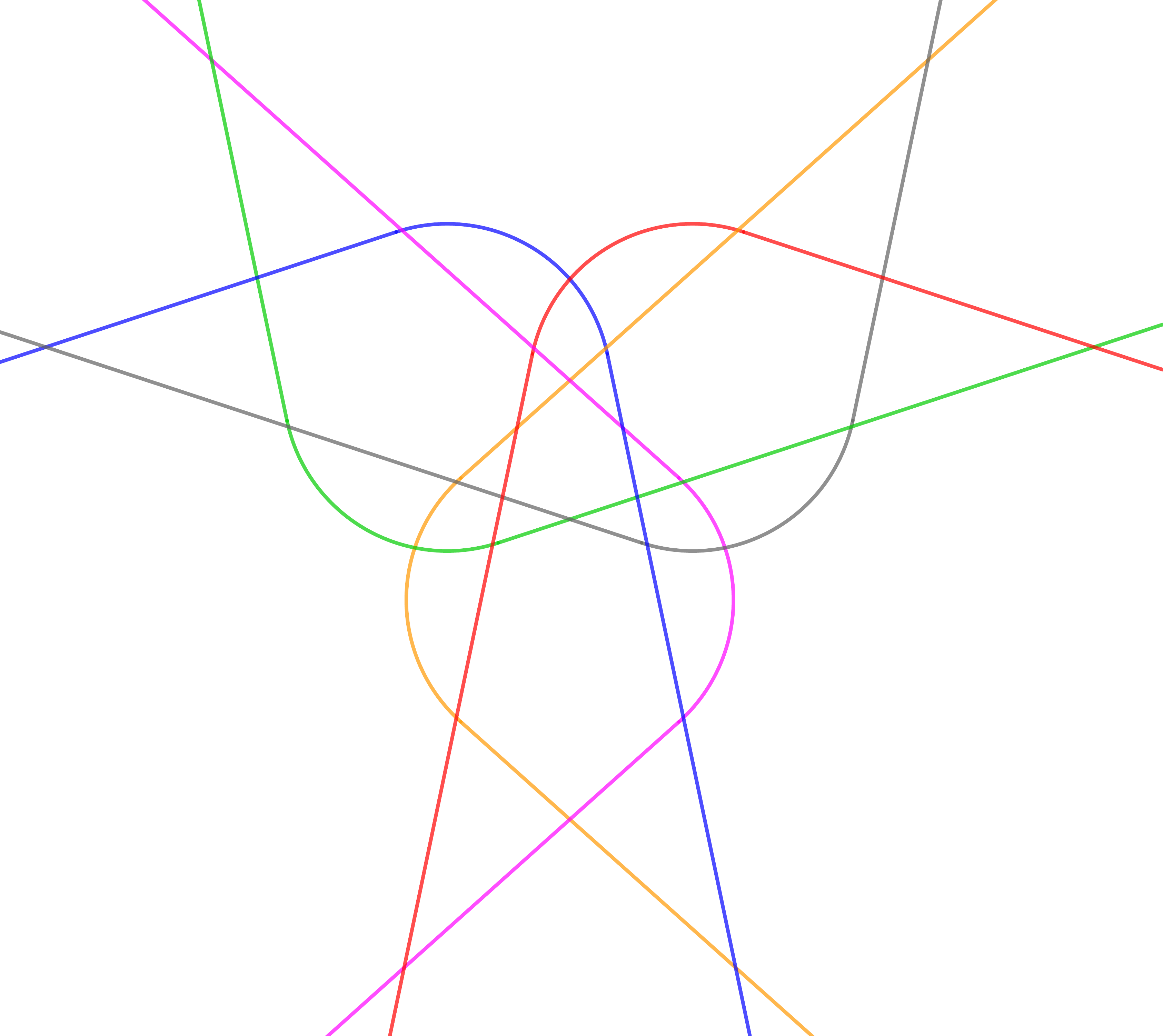}
	\end{minipage}
	\caption{\sf A scenario where no valid translation exists. 
		Left: The start and target configurations.
		Right: All the 12 vippodromes are drawn (in the translation plane).
		Observe that the vippodromes $\V^{(2)}_{AB}$ and $\V^{(1)}_{BA}$ coincide (their boundaries are drawn in orange), and so does every other similar pair of vippodromes for the other pairs of positions.
		Each point in the translation plane is covered by at least two vippodromes of contradicting constraints (e.g., for points inside the orange vippodrome, $A$ has to perform a motion before $B$ and vice versa), and so every translation is invalid (its TBG contains a 2-cycle).}
	\label{fig:labeled_infeasible_nontangent}
\end{figure}

It is clear that, unless they partially overlap, any pair of vippodrome boundaries intersect in a constant number of points.
The following lemma gives a sharp estimate on the number of intersections.

\begin{lemma}
	\label{lemma:intersect_four}
	For any $i,j \in \set{1,2}$ and any pairs $A \neq B, C\neq D \in M$, the vippodromes boundaries $\bd \V^{(i)}_{AB}, \bd \V^{(j)}_{CD}$ intersect at most four times.
\end{lemma}
\begin{proof}
Consider for simplicity two vippodromes of the form $\V^{(1)}_{AB}$ and $\V^{(1)}_{CD}$, and the corresponding wedges $U = \W(A^S,B^S) - B^T$ and $V = \W(C^S,D^S) - D^T$.
Observe that $\bd U$ and $\bd V$ intersect in at most four points, because each boundary consists of two rays.
Let $Z = U\cap V$; it is a possibly unbounded convex polygon with at most four edges.
Let $V' = V \setminus Z$.
It is easy to see that at most one vertex of $Z$ can lie in the interior of $V$.
If there is such a vertex then $V'$ is connected but not convex.
If there is no such vertex, then $V'$ consists of at most two connected components (it can also be empty), and each component is convex.
See Figure~\ref{fig:intersection_options} for an illustration.

\newcommand{\intersectionssubfigurewidth}{0.7}
\begin{figure}[h!]
	\centering
	\begin{subfigure}{.3\linewidth}
		\centering
		\includegraphics[trim=0 100 10 0,clip,width=\intersectionssubfigurewidth\textwidth]{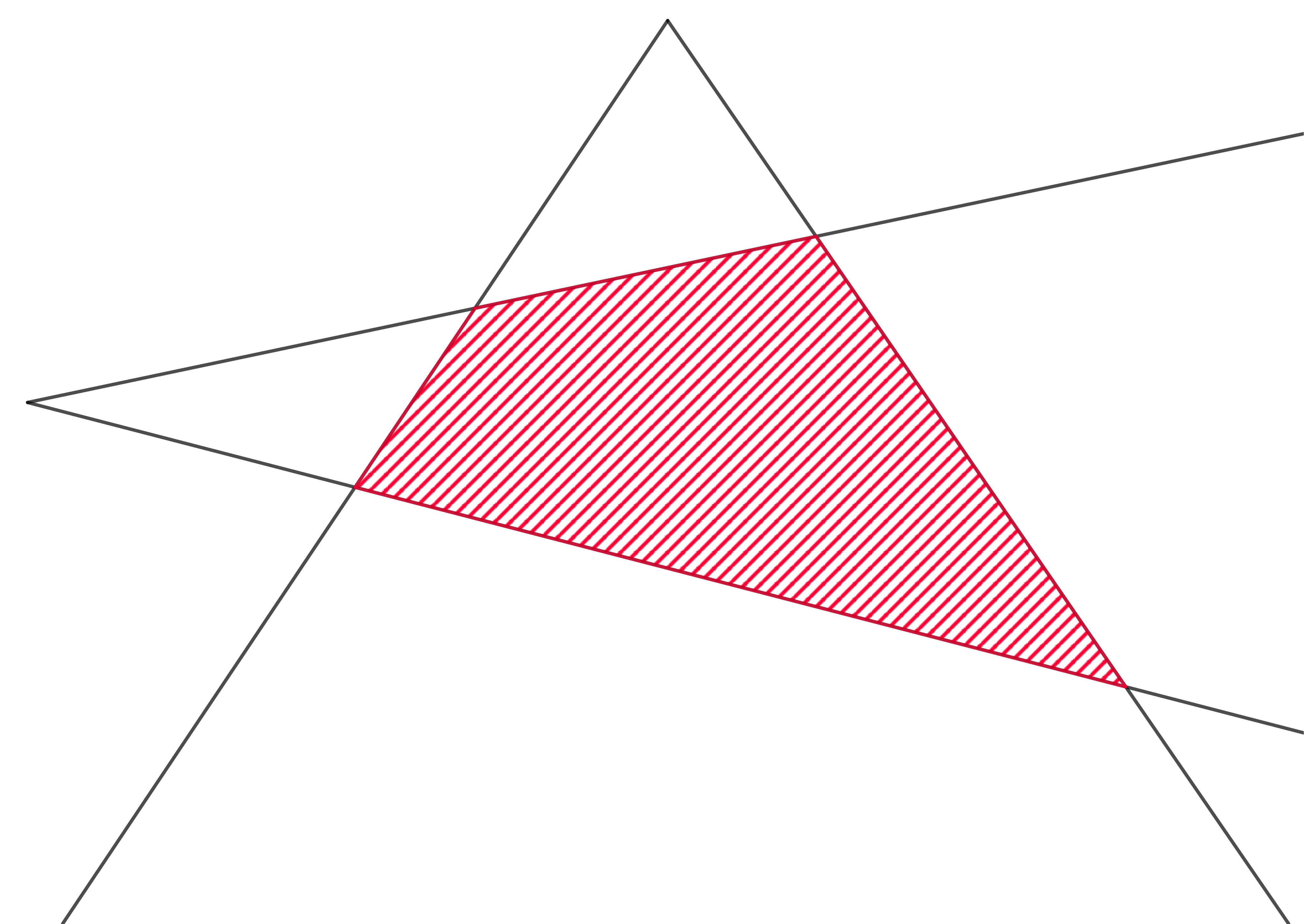}
		\caption{}
	\end{subfigure}%
	\begin{subfigure}{.3\linewidth}
		\centering
		\includegraphics[trim=0 100 10 0,clip,width=\intersectionssubfigurewidth\textwidth]{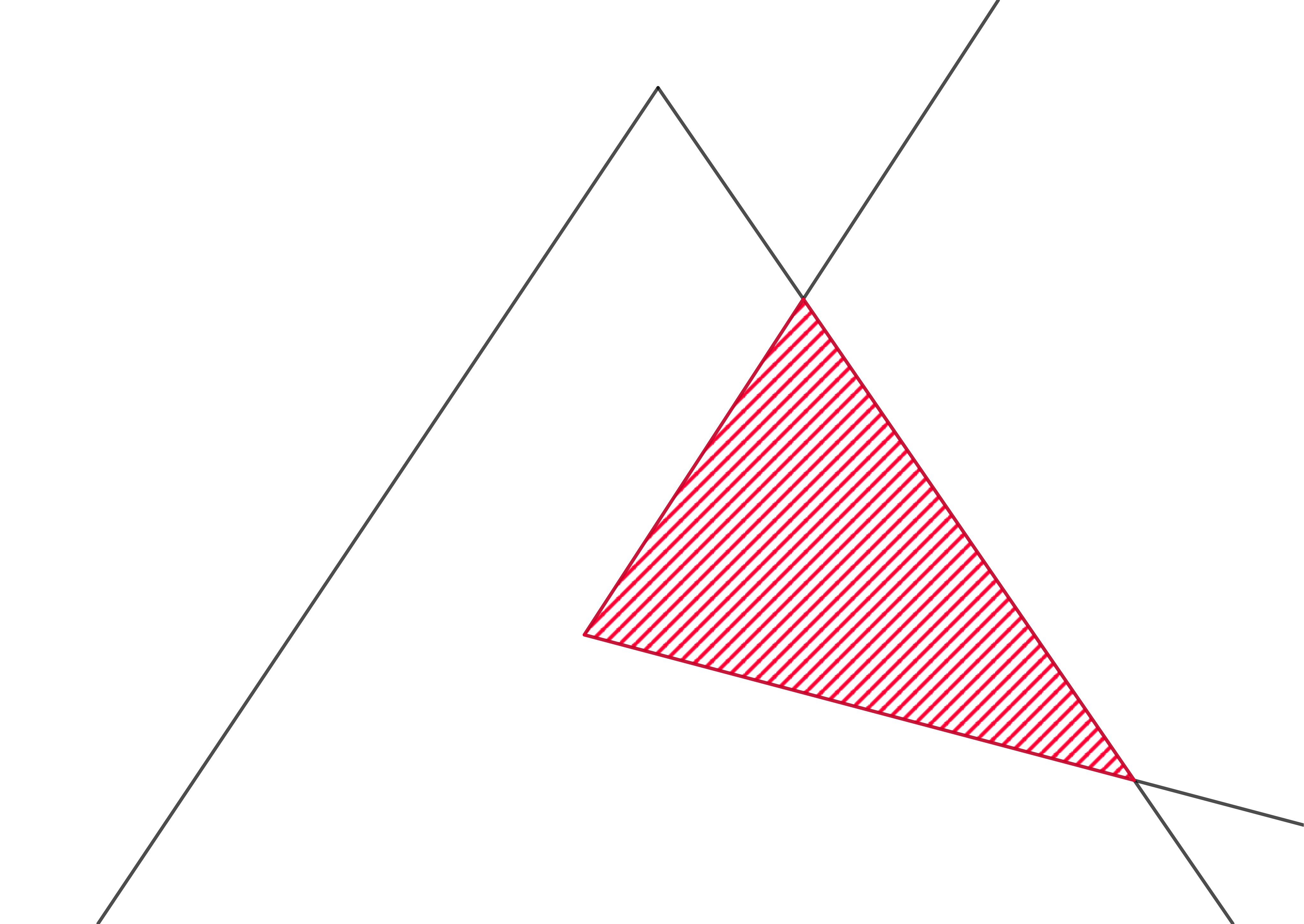}
		\caption{}
	\end{subfigure}\\[1ex]
	\begin{subfigure}{.3\linewidth}
		\centering
		\includegraphics[trim=0 100 10 0,clip,width=\intersectionssubfigurewidth\textwidth]{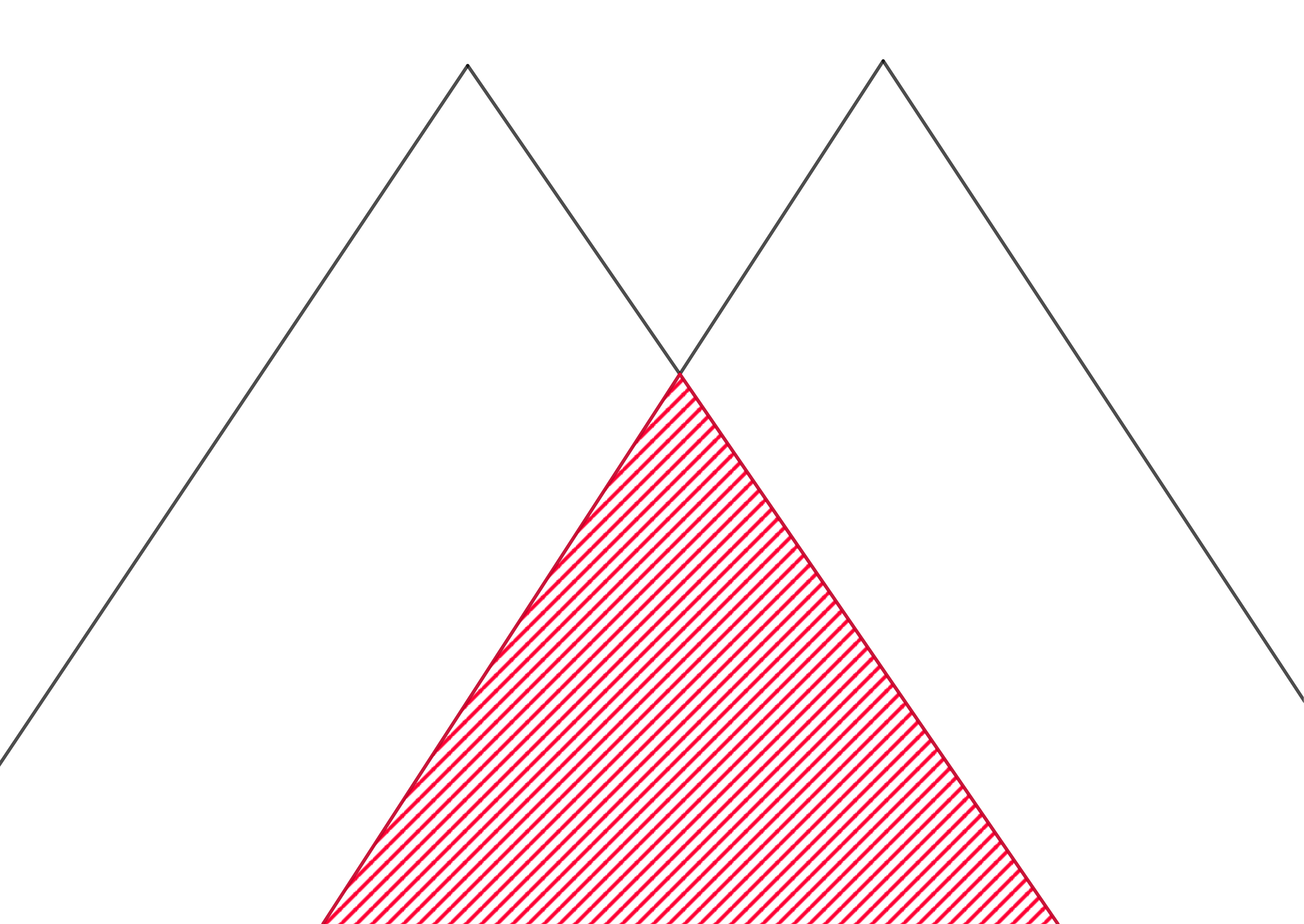}
		\caption{}
	\end{subfigure}
	\begin{subfigure}{.3\linewidth}
		\centering
		\includegraphics[trim=0 100 10 0,clip,width=\intersectionssubfigurewidth\textwidth]{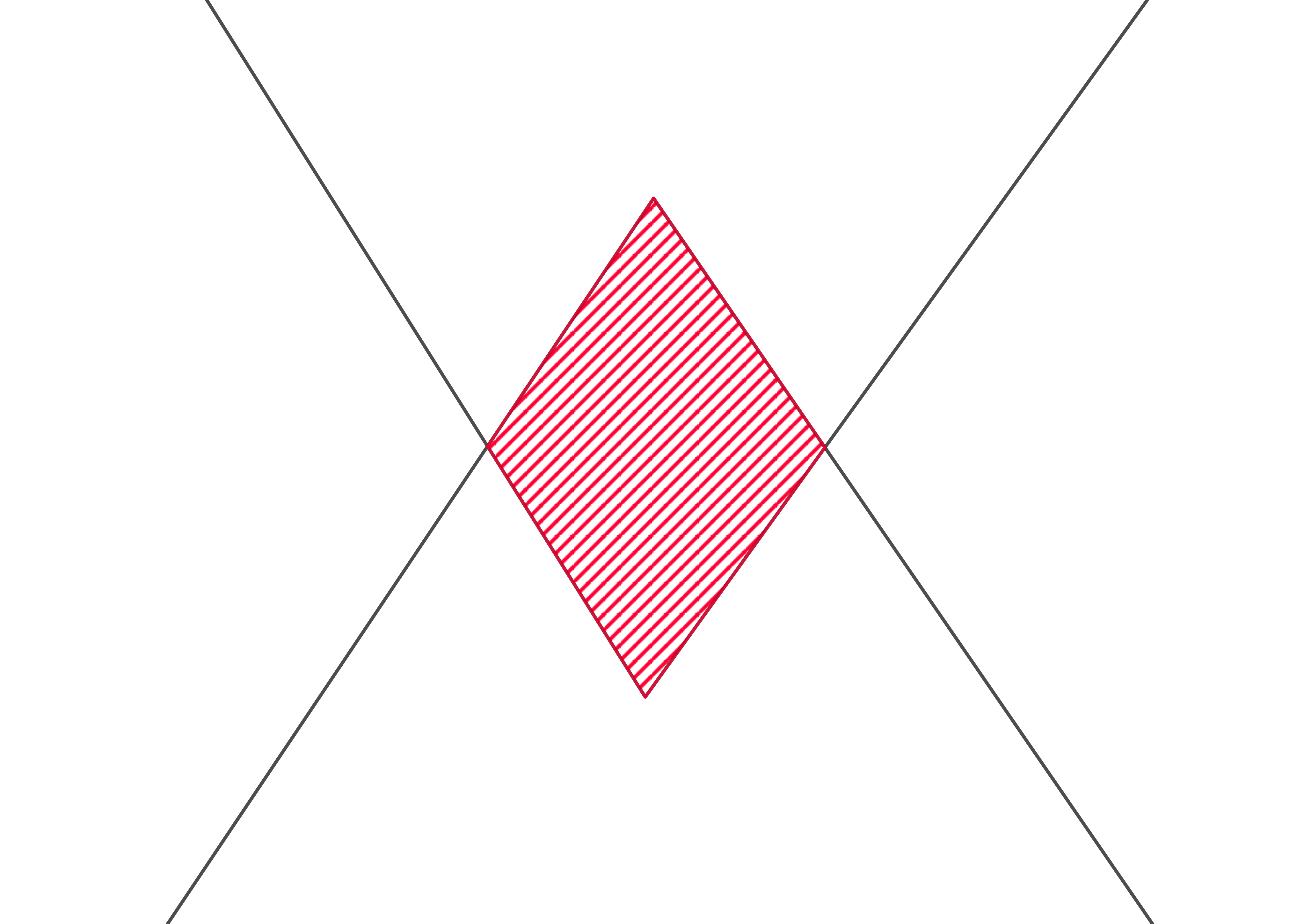}
		\caption{}
	\end{subfigure}
	\caption{\sf 
		Some of the possible intersections of $U$ and $V$, with the intersection $Z$ highlighted.
		In all cases, $V'$ consists either of one non-convex connected component, see (b) and (d), or of at most two convex connected components.}
	\label{fig:intersection_options}
\end{figure}

Let $q$ be an intersection point of $\bd \V^{(1)}_{AB}$ and $\bd \V^{(1)}_{CD}$; see Figure~\ref{fig:intersect_four}.
By (\ref{vips-and-wedges}), $q$ lies at distance $2$ from $\bd U$ and $\bd V$. Let $p$ be the point on $\bd V$ nearest to $q$.
It is easily checked that $p$ lies on $\bd V'$, so it lies either on $\bd V'_1$ or on $\bd V'_2$.
Since $V'_1$ and $U$ are interior-disjoint and convex, it follows from \cite{DBLP:journals/dcg/KedemLPS86} that the boundaries of their expansions by $D_2(0,0)$ intersect at most twice (unless they partially overlap).
The same holds for $V'_2$ and $U$, implying that $\bd \V^{(1)}_{AB}$ and $\bd \V^{(1)}_{CD}$ intersect in at most four points, as claimed.
\end{proof}

\begin{figure}
	\centering
	\includegraphics[trim=0 100 10 0,clip,width=0.5\textwidth]{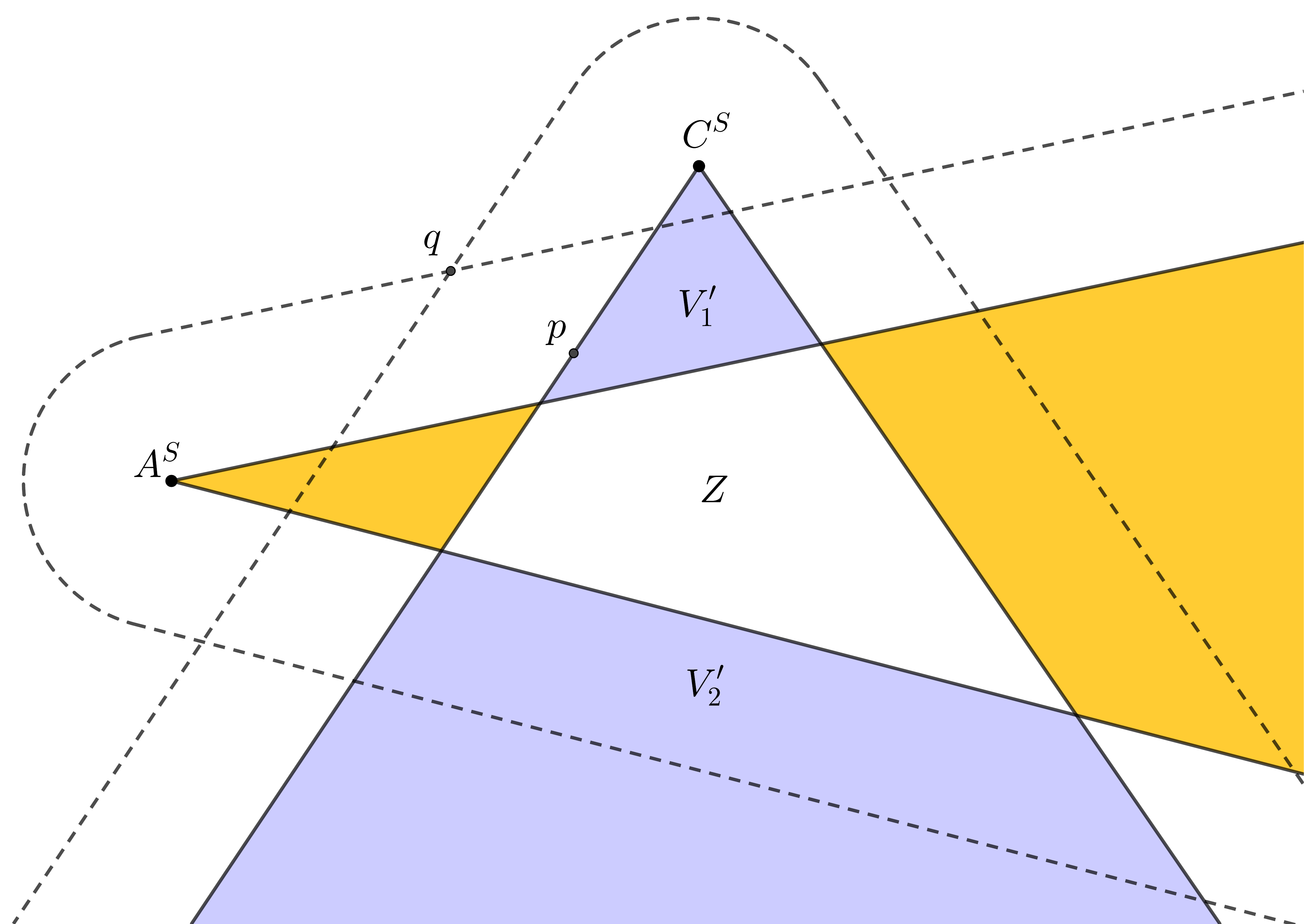}
	\caption{\sf Two vippodrome boundaries $\bd \V^{(1)}_{AB}$ and $\bd \V^{(1)}_{CD}$ marked by dashed curves.
	$V'$ and $U'= U\setminus Z$ are colored in blue and orange, respectively.
	Note that $p$ cannot lie on $\bd Z$.}
	\label{fig:intersect_four}
\end{figure}

\subsection{Constructing the Set of Valid Translations}
\label{subsec:labeled_Q}

Let $\V^\bd = \set{\bd \V^{(i)}_{AB} \mid i \in \set{1,2}, A \neq B \in M}$ and observe that $|\V^\bd| = 2n(n-1)$.
Define the vippodrome arrangement $\VA$, induced by $M$, to be the arrangement formed by the curves of $\V^\bd$; it is an arrangement of $O(n^2)$ rays and circular arcs.
Assuming general position, the only overlaps between features of the arrangement are between circular arcs of the two vippodromes of the same ordered pair $A,B$ (see Figure~\ref{fig:overlap_arcs} again).
To avoid these overlaps, we partition each of these arcs into two subarcs at the point where the overlap begins or ends (so each circle $\bd D_2(A^S-B^T)$ contributes at most three arcs to the arrangement).
It is worth mentioning that it is fairly easy to handle instances that are not in general position, where degeneracies may appear, such as collinear rays, or overlapping circular arcs, of two unrelated vippodromes.
In the rest of the \paper, we assume general position of pairs of unrelated vippodromes, to simplify the presentation.
Nevertheless, degeneracies can also be handled by suitable (and standard) extensions of our techniques.
We do, however, allow tangency between discs, which will require some special consideration.

We now observe that any pair of features (rays and circular arcs) of the (modified) arrangement intersect at most twice.
Hence the number of vertices in $\VA$ is at most $O(n^4)$, and so the overall complexity of the arrangement is also $O(n^4)$.
Consider a face $f$ of $\VA$.
We showed that for every ordered pair of pairs $(A,B) \in M^2$, $A \neq B$, the edge $AB$ is either in every graph $\Gv$, for $\vec{v} \in f$, or in none of these graphs.
Hence all the graphs $\Gv$, for $\vec{v} \in f$, are identical, and we denote this common graph as $G_f$.

We can construct $\VA$ either in $O(n^4 \log n)$ time using a plane-sweep procedure, or in $O(n^2 \lambda_4 (n^2))$ time,%
\footnote{Here $\lambda_s(m)$ denotes the maximum length of a Davenport-Schinzel sequence of order $s$ on $m$ symbols; see \cite{DBLP:books/daglib/0080837} for details.}
using the incremental procedure described in \cite[Theorem 6.21, p.~172]{DBLP:books/daglib/0080837}.

After $\VA$ has been constructed, we traverse its faces and construct the graphs $G_f$, over all faces $f$, noting that when we cross from a face $f$ to an adjacent face $f'$, the graph changes by the insertion or deletion of a single edge.%
\footnote{As already mentioned, although we assume general position, circular arcs of vippodromes may still overlap (see Figure~\ref{fig:overlap_arcs}).
Notice, however, that the overlapping arcs bound vippodromes that induce the same constraint on the itinerary and hence crossing the overlapping arcs still incurs insertion or deletion of a single edge to the graph.} We then test each graph $G_f$ for acyclicity.
The union of (the closure of) all the faces $f$ for which $G_f$ is acyclic is the desired region $Q$ of valid translations.
For each face $f$ that participates in $Q$, we run a linear-time procedure for topological sorting of $G_f$, and the order that we obtain\footnote{In general, $G_f$ can have exponentially many topological orders, each of which yields a valid itinerary.} defines a valid itinerary for all translations $\vec{v}\in f$.
The running time, for a fixed face $f$, is $O(n+m_f)$, where $m_f$ is the number of edges of $G_f$.
In the worst case we have $m_f = O(n^2)$, so the overall cost of the algorithm is $O(n^6)$.

We have thus obtained the following main result of this \ssection.

\begin{theorem}
\label{theorem:constructing_q}
Given a labeled instance ${\rm LST}(S,T,M)$ of the reconfiguration problem, with a valid start configuration $S$, a valid target configuration $T$, of $n$ points each, and a matching $M$ between $S$ and $T$, we can compute the region of all valid translations for $T$ in $O(n^6)$ time.
\end{theorem}

\medskip
\noindent{\bf Remark.} Clearly, Theorem~\ref{theorem:constructing_q} implies that we can find a valid translation, if one exists, within the same time bound $O(n^6)$.

\medskip
\noindent{\bf Remark.}
The bottleneck that determines the efficiency of the algorithm in Theorem~\ref{theorem:constructing_q} is the cost of testing for acyclicity of each of the graphs $G_f$ from scratch.
It would be interesting to see whether dynamic algorithms for maintaining acyclicity in a directed graph, under insertions and deletions of edges (namely, fully dynamic cycle detection algorithms), could be applicable when we traverse the faces of $\VA$.
Such algorithms can be found in \cite{DBLP:journals/jea/PearceK06}, but they do not seem to improve the asymptotic running time of our algorithm.
The only efficient algorithms that we are aware of, those with low total update time, only support insertions of edges but not deletions; see~\cite{bernstein2018incremental,cycleDetection}.
If a dynamic algorithm of this kind, that can handle both insertions and deletions on a prespecified sequence of operations, were available, with sublinear update time for each insertion and deletion, it would clearly improve the total running time of our procedure of finding $Q$, as well as the space-aware optimizations that are presented in the next \ssection.

In the algorithms of \Section~\ref{section:labeled_sa}, we need to minimize some function of $\vec{v}$ along the boundaries of the vippodromes (edges of $\VA$).
Since we require the translations to be valid, we are only interested in the valid intervals along the boundaries of the vippodromes, namely the edges of $Q$.
Moreover, for each vippodrome boundary, we will need to search, for any query point and direction along the boundary, for the nearest valid interval to the point in this direction.
Instead of iterating over each edge of $\VA$ or of $Q$ (which might take overall $O(n^4)$ time), we construct, as a by-product of the construction of $Q$, an auxiliary data structure $\D_Q(\gamma)$, whose performance is given in the following lemma.

\begin{lemma}
	\label{lemma:dq}
	We can construct, for each vippodrome boundary $\gamma$, a ``ray-shooting'' data structure $\D_Q(\gamma)$, so that each query to $\D_Q(\gamma)$ is a triplet $(a,d,b)$, where $a\in\gamma$ is a starting point, $d$ is a direction, which can be clockwise or counterclockwise along $\gamma$, and $b\in\gamma$ is a limit point beyond which the search stops, and which can be $\infty$ if the interval along $\gamma$ is unbounded.
	The answer to the query is the interval of $\gamma \cap Q$ nearest to $a$ along $\gamma$ in the direction $d$, which starts before we reach $b$, or else is an indication that no such interval exists.
	If $a \in Q$, it is reported as the beginning of the desired interval.
	$\D_Q(\gamma)$ can be constructed in $O(n^2)$ time, and each query takes $O(\log n)$ time.
\end{lemma}
\begin{proof}
	We construct $\D_Q(\gamma)$ as a balanced binary search tree of the endpoints of the arcs of $\gamma \cap Q$.
	It can be trivially built in $O(n^2)$ time while constructing $Q$, requires $O(n^2)$ storage (as there are $O(n^2)$ other vippodrome boundaries that can intersect $\gamma$) and answers a query in $O(\log n)$ time.
\end{proof}

Let $\D_Q$ denote the collection of the structures $\D_Q(\gamma)$, over all vippodrome boundaries $\gamma$.
$\D_Q$ can be constructed in $O(n^4)$ time while constructing $Q$, and its overall complexity is $O(n^4)$.

%% file: chapters/labeled_sa.tex
\section{Labeled Version: Space-Aware Optimization}
\label{section:labeled_sa}

In \Subsections~\ref{subsec:labeled_short_vector},~\ref{subsec:labeled_aabr}, and~\ref{subsec:labeled_sed}, we study the following three respective variants of the optimization criteria for the \salst{} problem:
\begin{enumerate}[(i)]
	\item \salstv{$S$,$T$,$M$}, for minimizing the length of the translation vector $\vec{v}$.
	\item \salsta{$S$,$T$,$M$}, for minimizing the area of the axis-aligned bounding rectangle of $D(S) \cup D(T+\vec{v})$, denoted as $\aabr{D(S) \cup D(T + \vec{v})}$.
	\item \salstd{$S$,$T$,$M$}, for minimizing the area of the smallest enclosing disc of $D(S) \cup D(T+\vec{v})$, denoted as \SED{$(D(S) \cup D(T+\vec{v}))$}.
\end{enumerate}

Our algorithms begin by computing $Q$ and $D_Q$, in $O(n^6)$ time, using the algorithm of the previous \ssection.
They then use $D_Q$ to solve the respective optimization problems, in time that is substantially smaller---it is $O(n^2 \log n)$ for problems (i) and (ii), and $O(n^5\log n)$ for problem (iii).
The results and bounds presented in this \ssection cater only to the optimization phases.

\subsection{Minimizing the Translation Vector}
\label{subsec:labeled_short_vector}
We define \salstv{$S$,$T$,$M$} as the following problem:
Given an \lst{} instance, of two configurations $S$ and $T$, and a matching $M$ between $S$ and $T$, find a valid translation $\vec{v} \in \R^2$, with respect to $M$, such that $|\vec{v}|$ is minimized.%
\footnote{Recall that we assume that at translation $\vec{v} = \vec{0}$ the centers of mass of $S$ and $T +\vec{v} = T$, or the centers of their smallest enclosing discs, coincide.
This makes the shortest valid translation a plausible criterion for space-aware optimization.}

\begin{theorem}
Once $Q$ and $\D_Q$ have been computed, \salstv{$S$,$T$,$M$} can be solved in $O(n^2 \log n)$ time.
\end{theorem}
\begin{proof}
Let $\psi(\vec{v}) = |\vec{v}|$ and observe that it has a single global minimum at the origin $o$ and no other local minima.
It thus follows that $\psi$ attains its minimum over any connected region $K$ at a point on $\bd K$, unless $o \in K$.
We claim that for every valid face $f$ of $\VA$, $\bd f$ is also valid.
This is a straightforward consequence of the fact that we allow the translating disc to touch other discs (without penetrating into them).
Thus, a valid translation of minimum length $\vec{v}$ is either $o$ or a point (in the valid portion) of some vippodrome boundary.
This suggests the following procedure.
We check whether $o$ is a valid translation in $O(n^2)$ time, according to Theorem~\ref{theo:abellanas}.
If so, we report it as the desired valid translation of minimum length.
If not, iterate over all the $O(n^2)$ vippodrome boundaries.
For each such boundary $\gamma$, we find the point $\vec{v}\in \gamma$ that minimizes $\psi(\vec{v})$ (i.e., that is nearest to the origin), as follows.

It is easily checked that $\restr{\psi}{\gamma}$ has only $O(1)$ points of local extrema along $\gamma$.
(There is an exceptional situation when $o$ is the center of the circle containing the circular arc of $\gamma$.
In this case we use a single arbitrary point on this boundary as a local extremum.)
We split $\gamma$ at these points into $O(1)$ subarcs, and $\psi$ is monotone along each subarc.

Let $\bar{\gamma}$ be one of these subarcs, let $a$ denote its endpoint at which $\restr{\psi}{\bar{\gamma}}$ is minimal, let $b$ be its other endpoint (possibly $\infty$), and let $d$ denote the direction from $a$ to $b$ (along $\gamma$).
We search along $\bar{\gamma}$, using Lemma~\ref{lemma:dq}, for the nearest point to $a$ that is in $Q$.
Repeating this to each subarc of each vippodrome boundary, we get, in $O(n^2\log n)$ time, $O(n^2)$ candidate translations, and return the one with the smallest value of $\psi$, if it exists.
Hence, after having computed $Q$ and $\D_Q$, \salstv{$S$,$T$,$M$} can be solved in $O(n^2\log n)$ time, as asserted.
\end{proof}

\ifthesis
\subsection[Minimizing the Area of the Axis-Aligned Bounding Rectangle]%
{\parbox[t]{17em}{Minimizing the Area of the \\ Axis-Aligned Bounding Rectangle}}
\else
\subsection{Minimizing the Area of the Axis-Aligned Bounding Rectangle}
\fi
\label{subsec:labeled_aabr} We define \salsta{$S$,$T$,$M$} as the following problem:
Given an \lst{} instance, of two configurations $S$ and $T$, and a matching $M$ between $S$ and $T$, find a valid translation $\vec{v} \in \R^2$, with respect to $M$, such that the area of $\aabr{D(S) \cup D(T+\vec{v})}$ is minimized.%

\begin{theorem}
Once $Q$ and $\D_Q$ have been computed, \salsta{$S$,$T$,$M$} can be solved in $O(n^2\log n)$ time.
\end{theorem}
\begin{proof}
Denote $\aabr{D(S)}$ by $R_1$ and $\aabr{D(T)}$ by $R_2$.
Note that \\
$\aabr{D(S)\cup D(T+ \vec{v})}$ is the axis-aligned bounding rectangle of $R_1 \cup (R_2 + \vec{v})$. Write $R_1 = [a_S, b_S] \times [c_S, d_S]$ and $R_2 = [a_T, b_T] \times [c_T, d_T]$.
Then, putting $\vec{v} = (x,y)$, we have that $\aabr{D(S)\cup D(T+ \vec{v})}$ is the axis-aligned bounding rectangle of
\[
\Bigl( [a_S, b_S] \times [c_S, d_S] \Bigr) \cup \Bigl( [x+a_T, x+b_T] \times [y+c_T, y+d_T] \Bigr) ,
\]
so it is the rectangle $[a^*(x), b^*(x)] \times [c^*(y), d^*(y)]$, where
\begin{align*}
a^*(x) & = \min \{ a_S,x+a_T \}, & b^*(x) & = \max \{ b_S,x+b_T \}, \\
c^*(y) & = \min \{ c_S,y+c_T \}, & d^*(y) & = \max \{ d_S,y+d_T \} .
\end{align*}
Let $\varphi(\vec{v}) = \varphi(x,y)$ denote the area of $\aabr{D(S)\cup D(T+ \vec{v})}$.
That is,
\[
\varphi(x,y) = \left( b^*(x) - a^*(x) \right) \left( d^*(y) - c^*(y) \right) .
\]
The function $b^*(x) - a^*(x)$ is piecewise linear in $x$, with the two breakpoints $a_S-a_T$, $b_S-b_T$.
Similarly, the function $d^*(y) - c^*(y)$ is piecewise linear in $y$, with the two breakpoints $c_S-c_T$, $d_S-d_T$ (the breakpoints of either function may appear in any order).
Each function is constant over the interval between the breakpoints, has slope $-1$ to the left of the interval, and
slope $+1$ to the right of the interval.

\begin{figure}
	\centering
	\includegraphics[trim=170 200 220 110, clip, width=\textwidth]{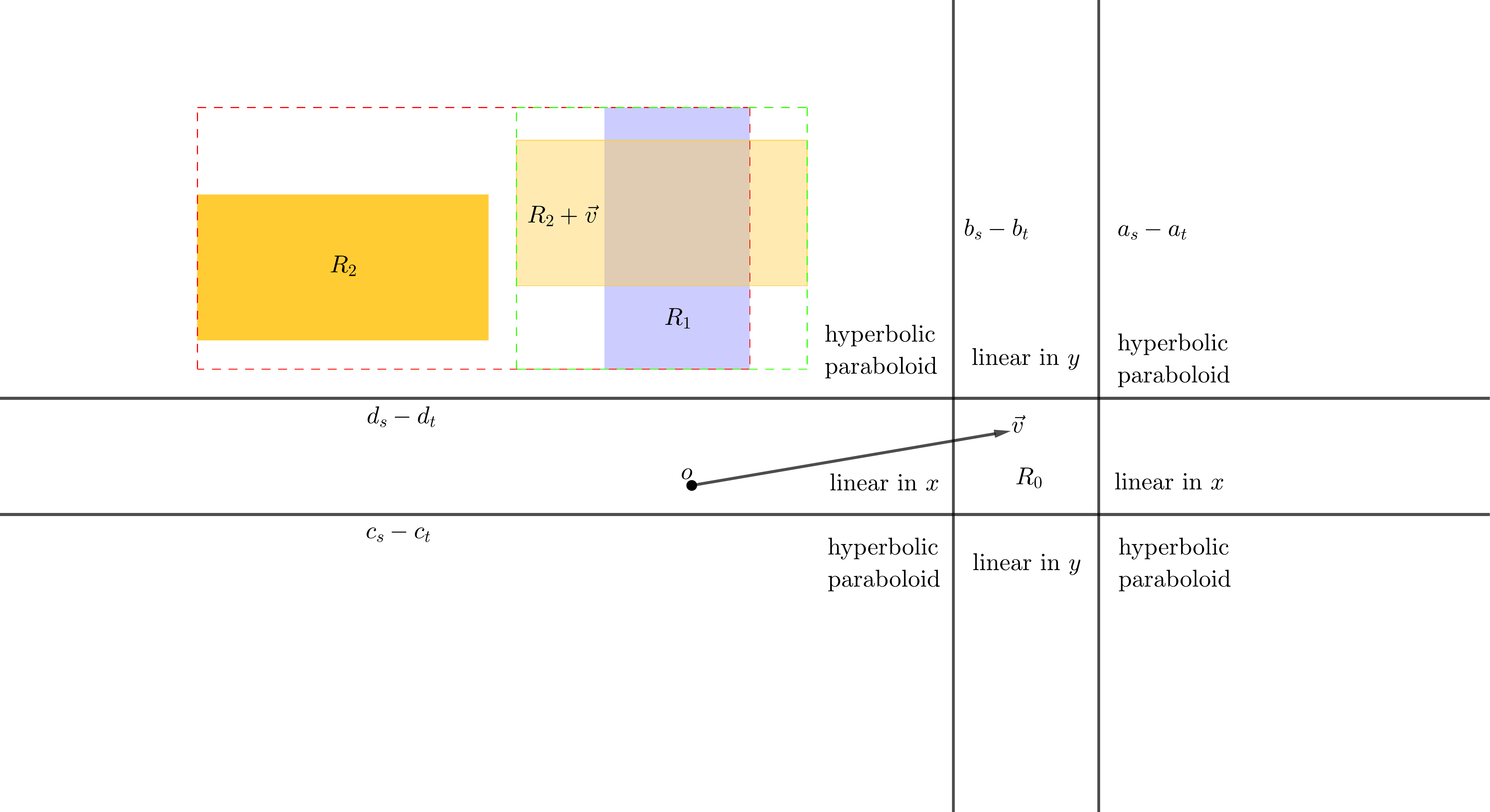}
	\caption{\sf An example of $\aabr{D(S)}$ (in blue), $\aabr{D(T)}$ (in orange) and $\aabr{D(T+\vec{v})}$ (in light orange) for some $\vec{v} \in R_0$.
	The lines of $L$ partition the translation plane into nine regions, in each of which $\varphi$ behaves differently, as indicated.
	The boundaries of $\aabr{D(S)\cup D(T)}$ and $\aabr{D(S)\cup D(T+\vec{v})}$ are depicted by red and green dashed segments, respectively.
	Observe that $\aabr{D(S)\cup D(T+\vec{v})}$ is of minimal area.}
	\label{fig:AABR}
\end{figure}

Consider the vertical lines through the breakpoints of $b^*(x) - a^*(x)$ and the horizontal lines through the breakpoints of $d^*(y) - c^*(y)$, and denote the set of these four lines by $L$.
$L$ partitions the plane into nine rectangular (bounded and unbounded) regions; see Figure~\ref{fig:AABR}.
The function $\varphi(x,y)$ is constant over the center region $R_0$, is a linear function in $x$ over the regions to the left and to the right of $R_0$, is a linear function in $y$ over the regions above and below $R_0$, and is a hyperbolic paraboloid, of the form $\pm (x-\alpha)(y-\beta)$, over each of the other four regions.
Observe that $\varphi$ is continuous.
Since none of the expressions for $\varphi$ has any local minimum in the interior of its region (except for $R_0$ where the expression is constant), it thus easily follows that, for any connected region $K$, $\varphi$ attains its minimum over $K$ at a point on $\bd K$, unless $R_0$ is fully contained in $K$, in which case the minimum is attained at all the translations in $R_0$.

Hence, we proceed exactly as in the previous problem.
We check whether an arbitrary point in $R_0$ is valid.
If so, report it as the desired valid translation that minimizes $\varphi(\vec{v})$.
If not, it means that $R_0$ is not contained in any cell of $Q$, and so it suffices to check for the minimum on the vippodrome boundaries.
We iterate over the $O(n^2)$ vippodrome boundaries, and for each such boundary $\gamma$, we minimize $\varphi$ over $\gamma$, by applying a variant of the procedure presented for the case of the shortest valid translation, in $O(\log n)$ time.
(Here too, ignoring validity, $\gamma$ contains only $O(1)$ local extrema of $\varphi$ restricted to $\gamma$, and we split $\gamma$ at these points, as in the preceding case.
Note that $\gamma$ may cross some lines of $L$.
In such a case, we split $\gamma$ further into $O(1)$ subarcs, each fully contained in one of the rectangular regions of the partition, and minimize $\varphi$ separately over each subarc.)
We output the translation $\vec{v}$ that minimizes $\varphi(\vec{v})$, over all the translations that we have obtained.
Hence, after having computed $Q$ and $\D_Q$, \salsta{$S$,$T$,$M$} can be solved in $O(n^2\log n)$ time, as asserted.
\end{proof}

\ifthesis
\subsection[Minimizing the Area of the Smallest Enclosing Disc]%
{\parbox[t]{17em}{Minimizing the Area of the \\ Smallest Enclosing Disc}}
\else
\subsection{Minimizing the Area of the Smallest Enclosing Disc}
\fi
\label{subsec:labeled_sed}

The problem here, denoted as \salstd{$S$,$T$,$M$}, is to find a valid translation $\vec{v}$ (with respect to $M$) that minimizes the radius of the smallest enclosing disc%
\footnote{This is indeed an equivalent formulation to the one given earlier: The smallest enclosing disc of $D(S)\cup D(T+\vec{v})$ has the same center as the disc that we find, and its radius is larger by $1$.
In contrast to other sections in the \paper, here we consider the \SED{} to be a closed disc.}
of $S\cup (T+\vec{v})$.
We denote the smallest enclosing disc of a set $P$ as $\SED(P)$, and its radius as $r(P)$.
A similar problem was studied in~\cite{DBLP:journals/comgeo/BanikBD14}, in which a characterization of the locus of the center of the smallest enclosing disc, and its radius, are given for a static set of points and only one mobile point, moving along a straight line.
Here, we study a more intricate problem, as our mobile points are more numerous and are not moving along a line, but are moving rigidly according to the valid translations of the region $Q$.
In other words, we want to optimize $\SED(S\cup (T+\vec{v}))$ only over translations $\vec{v} \in Q$.
(Note that without this constraint the problem is trivial: simply translate $T$ by $\vec{v}$ for which the centers of $\SED(S)$ and $\SED(T+\vec{v})$ coincide, and output the larger of the two discs.)

\begin{theorem}
	Once $Q$ and $\D_Q$ have been computed, \salstd{$S$,$T$,$M$} can be solved in $O(n^5 \log n)$ time.
\end{theorem}
\begin{proof}
If we fix $\vec{v}$, the smallest enclosing disc $D = \SED(S \cup (T+\vec{v}))$ passes through either two points, in which case $D$ is the diametral disc formed by the two points, or three points (or more, which we handle as degenerate instances of the first two options).
We conclude that $\bd D$ passes through either:
\begin{description}
	\item[(i)] two or three points that belong to the same set ($S$ or $T+\vec{v}$); or
	\item[(ii)] two points that belong to one set and a third point that belongs to the other set; or
	\item[(iii)] two points, each belonging to a different set.
\end{description}
For each case, we collect valid translations which are candidates to realize the smallest enclosing disc under the specific requirements of that case.
We then output $\vec{v}$ as the valid translation realizing the smallest $\SED(S \cup (T+\vec{v}))$ among the candidates.

Case~(i) is the simplest.
If $D$ is determined by points of the same set, say $S$, then $\SED(S)=\SED(S \cup (T+\vec{v}))$, unless other points from $T+\vec{v}$ are on the boundary as well, which we handle in Case (ii).
Indeed, this is trivial by the uniqueness of the $\SED$.
By the symmetry of the setup, assume without loss of generality that $r(S) \geq r(T)$.
It is then sufficient to find one valid translation $\vec{v}$ such that $T+\vec{v} \subset \SED(S)$.
Let $\xi$ be the center of $\SED(S)$.
In order for $\SED(S)$ to contain $T+\vec{v}$ for some translation $\vec{v}$,
$\xi$ has to lie in the intersection of all the discs of radius $r(S)$ centered at the points of $T+\vec{v}$.
Observe that this region of translations, denoted as $V(\xi)$, is the intersection $\bigcap_{t\in T} (\SED(S) - t)$.
We thus construct $V(\xi)$ and overlay it with the valid portion $Q$ of the arrangement $\VA$ of the vippodrome boundaries.
It is then easy to find, in time proportional to the complexity of the overlay, a valid translation $\vec{v}$ such that $T+\vec{v} \subset \SED(S)$, namely a translation in $Q \cap V(\xi)$, if one exists.
Since the complexity of the overlay is still $O(n^4)$ (its new vertices are intersections of edges of $V(\xi)$ with edges of $\VA$, and there are only $O(n^3)$ such intersections), this takes $O(n^4)$ time.
Observe that in this case only, if a valid translation was found, we do not need to consider other candidates from different cases, as $\SED(S)$ is clearly the smallest possible disc we are looking for.

Consider next Case (ii) where the boundary of the smallest enclosing disc passes through two points of $S$ and one point of $T+\vec{v}$.
(Handling the case where it passes through two points of $T+\vec{v}$ and one point of $S$ is fully symmetric.)
Any such pair of points of $S$ must define an edge of the farthest-neighbor Voronoi diagram $\FVD(S)$ of $S$, and this diagram has only $O(n)$ edges.
Fix such an edge, defined by two points $s_1, s_2\in S$, and denote it as $e = e_{s_1,s_2}$; see Figure~\ref{fig:SED_on_two}.
Assume, without loss of generality, that the perpendicular bisector $b(s_1,s_2)$ of $s_1s_2$ is the $x$-axis, with the origin placed at
the midpoint of $s_1s_2$.
Then we can write $e$ as an interval $[\xi_1,\xi_2]$.
Without loss of generality, assume that $0\le\xi_1 < \xi_2$.
(If $0$ is an interior point of $e$, split it into two subintervals at $0$ and handle each of them separately.
Note also that $e$ may be unbounded, in which case one of its endpoints is $\pm \infty$.)
For each $\xi\in e$, let $D(\xi)$ denote the disc centered at $\xi$ so that its bounding circle passes through $s_1$ and $s_2$.
Clearly, for a translation $\vec{v}$, $D(\xi)$ is the smallest enclosing disc of $S\cup (T+\vec{v})$, with respect to placements in $e$, if and only if $T+\vec{v} \subset D(\xi)$ and $D(\xi)$ is the smallest such
disc (meaning that $\xi$ is closest to the origin); by definition, $S \subset D(\xi)$ for every such $D(\xi)$.
Let $V(\xi)$, for $\xi \in e$, denote the set of all translations $\vec{v}$ for which $T+\vec{v} \subset D(\xi)$.
As in Case (i), we have $V(\xi) = \bigcap_{t\in T} (D(\xi) - t)$.
Define
$$
V = \{(\vec{v},\xi) \mid \xi\in e,\;\vec{v}\in V(\xi) \} ,
$$
which is a subset of $\reals^3$.
We claim that $V$ is connected.
For this, note that when $\xi' < \xi''$,  $D(\xi'')$ is a larger disc than $D(\xi')$, and therefore $V(\xi')$ can be translated to a subset of $V(\xi'')$.
More precisely, we have, as is easily verified,
$$
V(\xi') + (\xi''-\xi',0) \subseteq V(\xi'') .
$$
In particular, if $V(\xi')$ is nonempty then so is $V(\xi'')$.
Moreover, if $\vec{v}\in V(\xi')$ then $\vec{v} + (\xi''-\xi',0) \in V(\xi'')$, implying that we can continuously move from any point in $V$ to any point in a higher cross section of $V$ along a path that is contained in $V$, which clearly implies the connectivity of $V$.

\begin{figure}
	\centering
	\includegraphics[trim=20 20 20 20, clip, width=0.5\textwidth]{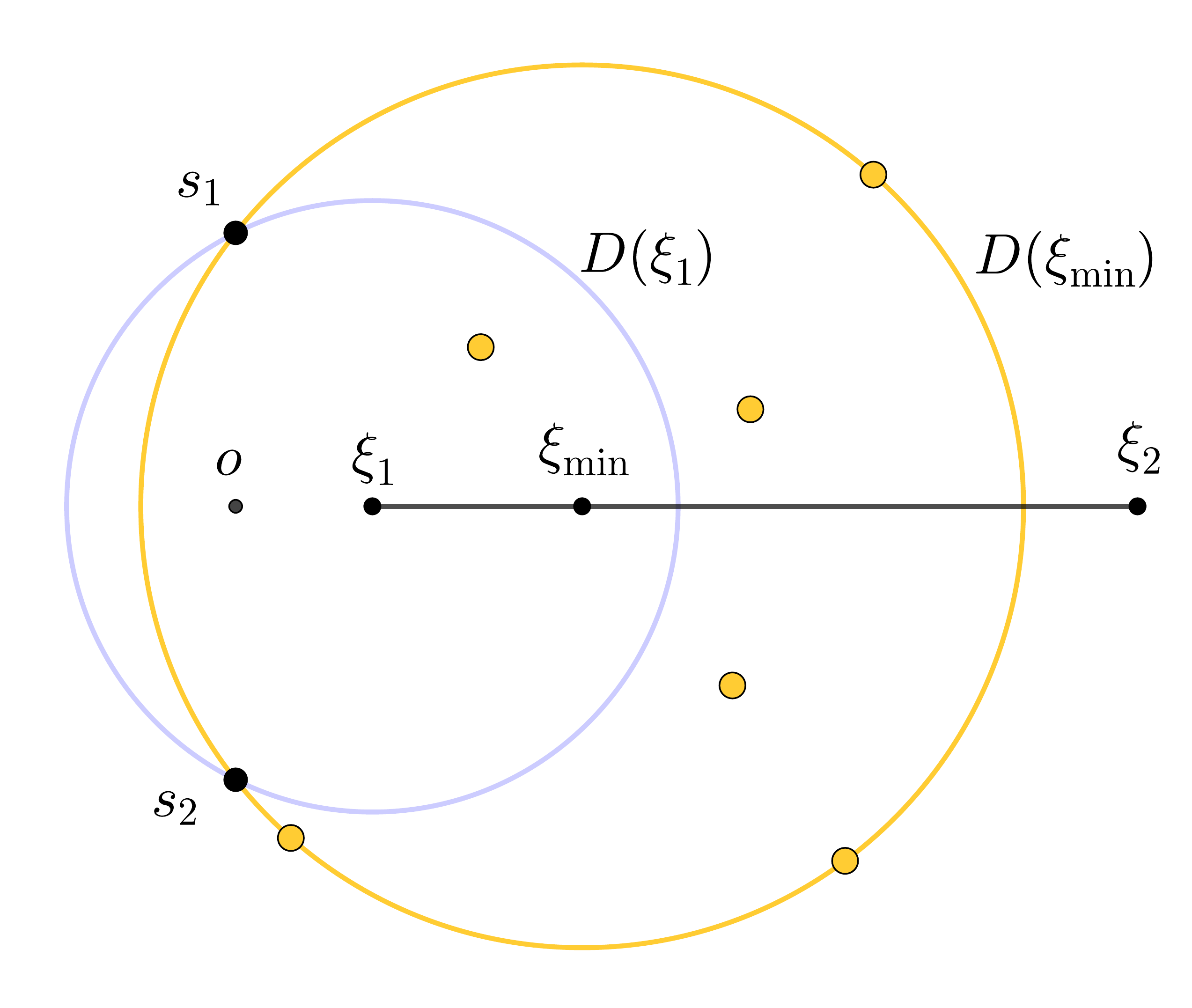}
	\caption{\sf The origin $o$ is the midpoint of $s_1 s_2$.
		The points of $T+\vec{v}$ are drawn in orange.
		The other points of $S$ are not drawn, but they lie in all the discs $D(\xi)$, for $\xi\in [\xi_1,\xi_2]$.
		$D(\xi_1)$ (in blue) is the smallest enclosing disc of $S$ among those centered on the Voronoi edge $e_{s_1,s_2} = [\xi_1, \xi_2]$.
		$D(\xi_{\min})$ (in orange) is the smallest enclosing disc of $S \cup (T + \vec{v})$ whose center lies on $e_{s_1,s_2}$.}
	\label{fig:SED_on_two}
\end{figure}

The coordinate $\xi_{\min}$ of the global minimum $(\vec{v},\xi_{\min})$ of $V$ can be attained at one of two specific points on $e$ (with a suitable $\vec{v}$), which is either $\xi_1$ if $r(D(\xi_1)) > r(T)$, or else the point $\xi_{\min}$ for which $D(\xi_{\min})$ is of radius $r(T)$.
In the latter case, the corresponding disc $D$ is $D(\xi_{\min})=  \SED(T + \vec{v})$, for some translation $\vec{v}$, and so is handled in Case~(i).
In the former case, the corresponding disc $D$ is $D(\xi_1)$ and we show below that either it is discovered by Case~(i) of the algorithm, or it cannot be the smallest among the candidates.
If $D(\xi_1) = \SED(S)$, it is discovered by Case~(i).
Otherwise, we argue that the translation $\vec{v}$ where the minimum is attained must be on the boundary of a vippodrome, a case that will be handled further below.
Assume to the contrary that there exists some valid translation $\vec{v}$ such that $\vec{v}$ is not on any vippodrome boundary and $D(\xi_1)$ encloses $S\cup(T+\vec{v})$.
Since $D(\xi_1)$ is not $\SED(S)$, the uniqueness of $\SED(S)$ implies that $r(S) < r(D(\xi_1))$.
Moreover, $\xi_1$ cannot be the midpoint $o$ of $s_1s_2$, for otherwise $D(\xi_1)$ would be $\SED(S)$.
For the same reason, the three points of $S$ on $\bd D(\xi_1)$ must form an obtuse triangle, implying that there exists a point $\xi'_1$ in the neighborhood of $\xi_1$, not necessarily on $e$, such that $D(\xi'_1)$ is smaller than $D(\xi_1)$.
We now claim that $D(\xi_1)$ cannot be the smallest among the candidates and so it can be ignored.
Indeed, since $r(T)$ is smaller than $r(D(\xi_1))$, we can choose the point $\xi_1'$ in the preceding argument so that $r(T) < r(D(\xi_1'))$.
Since $\vec{v}$ is not on a vippodrome boundary, there exists another valid translation $\vec{v}'$ sufficiently close to $\vec{v}$ so that $T+\vec{v}' \subset D(\xi_1')$.
Therefore, $D(\xi_1')$ is a smaller candidate than $D(\xi_1)$.

Hence, the minimum is attained at a (valid) translation $\vec{v}$ that lies on the boundary of some vippodrome.
We thus take each of the $O(n^2)$ such boundaries $\gamma$, and mark on it the maximal subarcs of valid translations (which are delimited at intersection points of $\gamma$ with other vippodrome boundaries).
For each $t\in T$ and each translation $\vec{v}$, the range of $\xi\in e$ for which $t+\vec{v} \in D(\xi)$ is a subinterval of $e$.
If the subinterval is nonempty, one of its endpoints is either $\xi_1$ or $\xi_2$.
We write this interval, if nonempty, as $f_t(\vec{v}) \le \xi \le g_t(\vec{v})$, write the empty intervals as $\xi\le -\infty$ or $\xi \ge +\infty$, and conclude that the range of $\xi$ for which $T+\vec{v} \subset D(\xi)$ is given by
$$
\max_{t\in T} f_t(\vec{v}) \le \xi \le \min_{t\in T} g_t(\vec{v}).
$$
We thus compute this sandwich region $\Xi$ between the upper envelope of the functions $f_t(\vec{v})$ and the lower envelope of the functions $g_t(\vec{v})$, as $\vec{v}$ ranges over $\gamma$, and seek a pair $(\vec{v},\xi) \in \Xi$ such that $\vec{v}$ is a valid translation and $\xi$ has the minimum value under these conditions.
Clearly $(\vec{v},\xi)$ lies on the upper envelope of the functions $f_t(\vec{v})$.
Since the functions $f_t,g_t$ are partially defined piecewise algebraic functions of constant degree, and $\vec{v}$ varies along a one-dimensional curve (a vippodrome boundary), this takes $O(\lambda_s(n)\log n)$ time, for some constant parameter $s$ (see, e.g., \cite{DBLP:books/daglib/0080837}).

For the running time of this procedure, we need to repeat it for the $O(n)$ pairs of neighbors $(s_1,s_2)$ in $\FVD(S)$ (and similarly for $\FVD(T)$).
For each pair we run $O(n^2)$ one-dimensional minimization procedures, over the vippodrome boundaries, each costing $O(n^2)$ time (just to mark on it the valid portions).
Thus the total running time in this case is $O(n^5)$.

Consider finally Case (iii), in which the smallest enclosing disc is the diametral disc of two points $s\in S$ and $t+\vec{v}\in T+\vec{v}$.
There are $O(n^2)$ pairs $(s,t)$ to test.
Fix one such pair $(s_0,t_0)$, and, for each $\vec{v}$, denote the corresponding diametral disc as $D_{s_0,t_0}(\vec{v})$
(note that we overload the definition of $D$ until the end of the \ssubsection).

\begin{figure}
	\centering
	\includegraphics[trim=0 150 0 50, clip, width=0.6\textwidth]{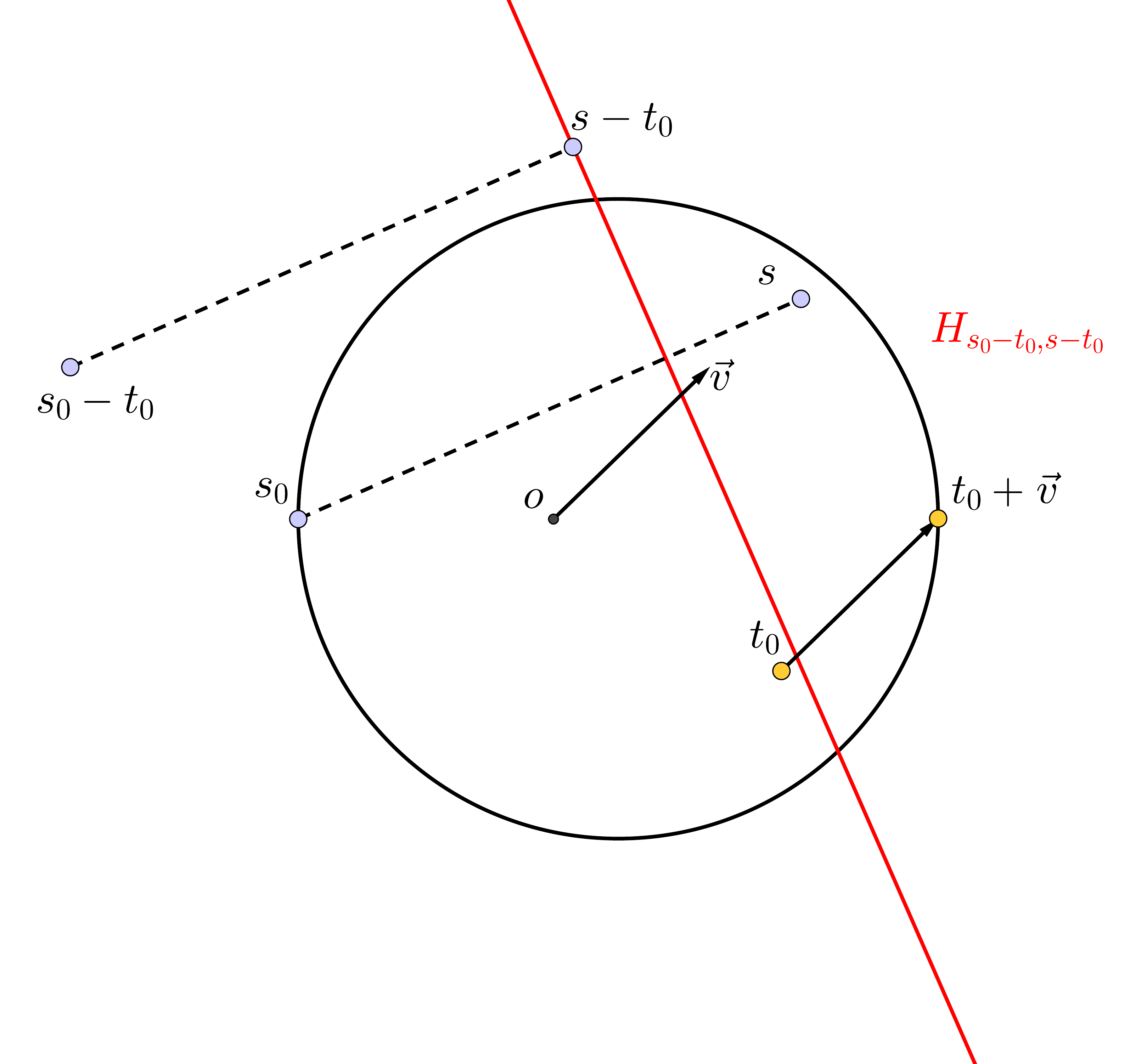}
	\caption{\sf The points $s_0, s$ (in blue), $t_0$ (in orange), and their respective translated points according to $-t_0$ and $\vec{v}$.
		The halfplane $H_{s_0-t_0,s-t_0}$, lying to the right of its bounding line, are both marked in red.
		Observe that $\vec{v} \in H_{s_0-t_0,s-t_0}$ and so $s$ is inside the diametral disc of $s_0$ and $t_0+\vec{v}$.}
	\label{fig:SED_halfplane}
\end{figure}

A point $s\in S\setminus\{s_0\}$ is in $D_{s_0,t_0}(\vec{v})$ if and only if $\angle (s_0,s,t_0+\vec{v}) \ge \pi/2$, which is equivalent to $t_0+\vec{v}$ lying in the halfplane $H_{s_0,s}$, defined so that it does not contain $s_0$, its bounding line passes through $s$, and is orthogonal to $s_0s$. This in turn is equivalent to $\vec{v}$ lying in the similarly defined halfplane $H_{s_0-t_0,s-t_0}$; see Figure~\ref{fig:SED_halfplane}.
This has to hold for every $s$, and, symmetrically, $\vec{v}$ also has to lie in each halfplane $H_{s_0-t_0,s_0-t}$, for each $t\in T\setminus\{t_0\}$, as is easily checked.
We thus form the intersection of these $2(n-1)$ halfplanes, in $O(n\log n)$ time, to obtain a convex polygon $K_{s_0,t_0}$ with $O(n)$ edges, which we abbreviate as $K$.
Note that $K$ is always an unbounded region, unless it is empty.
We then seek a valid translation $\vec{v}$ inside $K$ for which the function $\varphi(\vec{v}) = |\vec{v} + t_0 - s_0|$ (that is, the diameter of $D_{s_0,t_0}(\vec{v})$) is minimized; namely, we want to minimize $\restr{\varphi}{K \cap Q}$.
In other words, we seek the valid translation $\vec{v}$ inside $K$ that is closest to $s_0-t_0$.
Since, by construction, $s_0-t_0$ lies outside each of the halfplanes that form $K$, it follows that the desired translation $\vec{v}$ lies either on $\bd K$ or on the boundary of some valid face of $Q$, that is, on the valid portion of some vippodrome boundary.
Specifically, $\vec{v}$ is either a point of some connected component $e$ of the valid portion of some vippodrome boundary $\gamma$ within $K$ so that $\vec{v}$ minimizes $\varphi$ over $e$, or the unique point $\vec{v}_{\min}$ on $\bd K$ that is closest to $s_0-t_0$, if that point is valid.
(Note that in this latter case, $\vec{v}_{\min}$  is an inner point of some cell of $Q$.)

All these observations suggest the following procedure.
We first construct $K$, in $O(n\log n)$ time.
We then find $v_{\min}$, the unique point on $\bd K(s_0,t_0)$ closest to $s_0-t_0$, and check the validity of $\vec{v}_{\min}$ in $O(n^2)$ time, according to Theorem~\ref{theo:abellanas}.
If $\vec{v}_{\min}$ is valid, we output it as the valid translation that minimizes the radius of $D_{s_0,t_0}(\vec{v})$ (for the fixed pair $(s_0,t_0)$).

Otherwise, we look for the desired translation on a vippodrome boundary.
For each vippodrome boundary $\gamma$, we apply the following procedure, which finds the valid translation that minimizes $\restr{\varphi}{\gamma \cap K}$, and output the translation that minimizes $\restr{\varphi}{Q \cap K}$, by taking the translation with minimum value of $\varphi$, among all the candidate translations, collected over all vippodrome boundaries, if such candidate translations exist at all.

We prepare a data structure $\D_K(\gamma)$ along $\gamma$ for $K$, which is similar to the data structure $\D_Q(\gamma)$ described in Lemma~\ref{lemma:dq}.
The structure $\D_K(\gamma)$ is a balanced binary search tree over the points $\gamma\cap \partial K$, it can be built in $O(n\log n)$ time, and a query takes $O(\log n)$ time.

We split $\gamma$ at local extrema of $\restr{\varphi}{\gamma}$.
Notice that we have only up to three local extrema%
\footnote{Unless $\gamma$ is $\bd\V^{(1)}_{AB}$ or $\bd\V^{(2)}_{AB}$ such that $A^S = s_0$ and $B^T = t_0$.
	In this exceptional case, the entire circular subarc of $\gamma$ is a continuum of minima.
	This case is easily handled by a slight modification to the algorithm, which we omit here.}
(in the worst case we have one minimum on each ray of $\gamma$ and one maximum on its circular portion) and they split $\gamma$ into at most four subarcs, along each of which $\restr{\varphi}{\gamma}$ is monotone.

Let $\bar{\gamma}$ be one of these subarcs, let $a$ denote its endpoint at which $\restr{\varphi}{\bar{\gamma}}$ is minimal, let $b$ be its other endpoint (possibly $\infty$), and let $d$ denote the direction from $a$ to $b$ (along $\gamma$).
We search along $\bar{\gamma}$ for the nearest point to $a$ that is in $K\cap Q$.
Finally we compare the outputs over all the subarcs and return the point with the smallest value of $\varphi$, which is the optimal valid translation $\vec{v}$ along $\gamma$, if it exists.

Searching along one subarc $\bar{\gamma}$ is carried out as follows.
Using $\D_K(\gamma)$ we obtain the next interval of $\bar{\gamma}\cap K$ nearest to $a$, if it exists, and use it (with trivial adaptations) to query $\D_Q(\gamma)$.
If $\D_Q(\gamma)$ returns an interval, we use the first point of this interval as our answer (namely the value of $\varphi$ at this point) and by that conclude handling the subarc $\bar{\gamma}$.
If $\D_Q(\gamma)$ returns null, we use $\D_K(\gamma)$ to fetch the next interval along $\bar{\gamma}\cap K$ nearest to $a$, and so on until we find a valid solution or reach the other end of $\bar{\gamma}$.

Since there are at most $O(n)$ intervals in $\gamma\cap K$, the overall number of queries to $\D_K(\gamma)$ will be $O(n)$ at the cost of $O(\log n)$ time each.
For each of the at most $O(n)$ such queries we query $\D_Q(\gamma)$ once, in $O(\log n)$ time.
Thus processing $\gamma$ takes a total of $O(n\log n)$ time.

Repeating this over the $O(n^2)$ vippodrome boundaries, and repeating the entire process for the $O(n^2)$ pairs $(s_0,t_0)$, the overall cost of this case is $O(n^5\log n)$ time.

The total running time of all the cases is thus $O(n^5 \log n)$.
This completes the proof of the theorem.
\end{proof}

We note that the other two optimization criteria (shortest translation and minimum-area {\sf AABR}) have considerably simpler and significantly more efficient solutions (modulo the preprocessing stage of constructing $Q$ and $\D_Q$).
We leave it as a challenging open problem to improve the running time of the optimization phase for this case, at least to nearly quartic in $n$.

%% file: chapters/unlabeled.tex
\section{Unlabeled Version: Preliminary Analysis}
\label{section:unlabeled}

In this \ssection we study the reconfiguration problem for unlabeled discs.
The main result of the \ssection is summarized the following theorem.

\begin{theorem}
\label{theo:unlabeled_feasible}
Let $S$ and $T$ be two valid configurations, of $n$ points each.
For every direction $\delta$, except possibly for finitely many directions, there
exists a translation $\vec{v} \in \R^2$ in direction $\delta$ such that the unlabeled problem {\rm \ust}($S$,$T+\vec{v}$) is feasible.
\end{theorem}

\medskip
\noindent{\bf Remark.}
Note that the theorem implies that we can always move the discs from $S$ to $T$ in $2n$ moves:
first move the discs from $S$ to $T+\vec{v}$, using the valid itinerary provided by the theorem, and then move the discs from $T+\vec{v}$ to $T$ by translating each of them by $-\vec{v}$, in the order of their centers in direction $\vec{v}$ (which can easily be shown to be collision-free).
This almost reproduces the result of \cite{DBLP:journals/comgeo/AbellanasBHORT06}, already mentioned in the introduction, where the bound is $2n-1$, for the case where we are not allowed to shift the target locations.
In some sense, our result is stronger, in that in the second step, all the discs are translated by the same vector $-\vec{v}$.

\begin{proof}
Let $C$ be a valid configuration of $n$ points in the plane.
Let $c, c'$ be two points in $C$ and let $b(c,c')$ denote their perpendicular bisector.
Put $\B(C) = \set{b(c,c') \mid c,c' \in C, dist(c,c') = 2}$, which is the set of all perpendicular bisectors (common inner tangents) of any pair of touching discs of $D(C)$.
We say that a direction is \emph{generic} for $C$ if it is not parallel to any line in $\B(C)$.
(Note that, by Euler's formula for planar maps, there are only $O(n)$ non-generic directions.)
We fix a generic direction $\delta$ for both $S$ and $T$.
Observe that $\delta$ is also generic for $T+\vec{v}$, for any vector $\vec{v}$.
Without loss of generality, assume that $\delta$ is horizontal
and points to the right.
We define $\Pi_\delta(C)$ to be the reverse lexicographical order of the points in $C$, that is,
$\Pi_\delta(C) = (c_1, c_2, \dots, c_n)$, so that, for any $1 \leq i < j \leq n$, $c_i$ is to the right of (or at the same $x$-coordinate but above) $c_j$.
We now fix a matching $M_\delta$ according to the orders $\Pi_\delta(S)$ and $\Pi_\delta(T)$, by aligning both orders, i.e., $M_\delta(s_i) = t_i$, where $s_i$ (resp., $t_i$) is the $i$-th point in $\Pi_\delta(S)$ (resp., $\Pi_\delta(T)$), for $i=1,\dots,n$.
The matching $M_\delta$ transforms the problem to the labeled version \lst($S$,$T$,$M_\delta$).
We claim that this specific instance is always feasible, and, moreover, admits valid translations in direction $\delta$.
Order $M_\delta$ in the same order of $\Pi_\delta(S)$ and $\Pi_\delta(T)$, i.e., $(s_1,t_1), \dots, (s_n,t_n)$, and denote this order as $\Pi(M_\delta)$.
We claim that one can always choose $\vec{v}$, in direction $\delta$, such that $(s_1,t_1+\vec{v}), \dots, (s_n,t_n + \vec{v})$ is a valid itinerary.

We apply a simpler variant of the techniques developed in \Section~\ref{section:labeled}.
Since we have already assigned the fixed order $\Pi(M_\delta)$ to $M_\delta$, we do not need to take into consideration all the vippodromes, but only the ones that impose constraints that violate $\Pi(M_\delta)$.
Let $A_i$ be the pair $(s_i,t_i) \in M_\delta$ (so $D(s_i)$ is the disc that moves at step $i$).
Let $$\V_{\rm bad}(\delta)= \set{\V^{(j)}_{A_k A_l} \mid j \in \set{1,2}, A_k, A_l \in M_\delta, k>l},$$
and observe that $|\V_{\rm bad}(\delta)| = n(n-1)$.
In other words, $\V_{\rm bad}(\delta)$ is the subset of all the vippodromes $V$, such that, for each $\vec{v} \in V$, the constraint that $V$ represents violates the itinerary according to $\Pi(M_\delta)$ between $S$ and $T+\vec{v}$.
Thus, in order to find a valid translation $\vec{v}$ in direction $\delta$, it suffices to show that the ray $\rho$ from the origin in direction $\delta$ (the positive $x$-axis by assumption) is not fully contained in the union of the vippodromes in $\V_{\rm bad}(\delta)$.

\begin{figure}
	\centering
	\includegraphics[trim=0 170 100 45, clip, width=.7\textwidth]{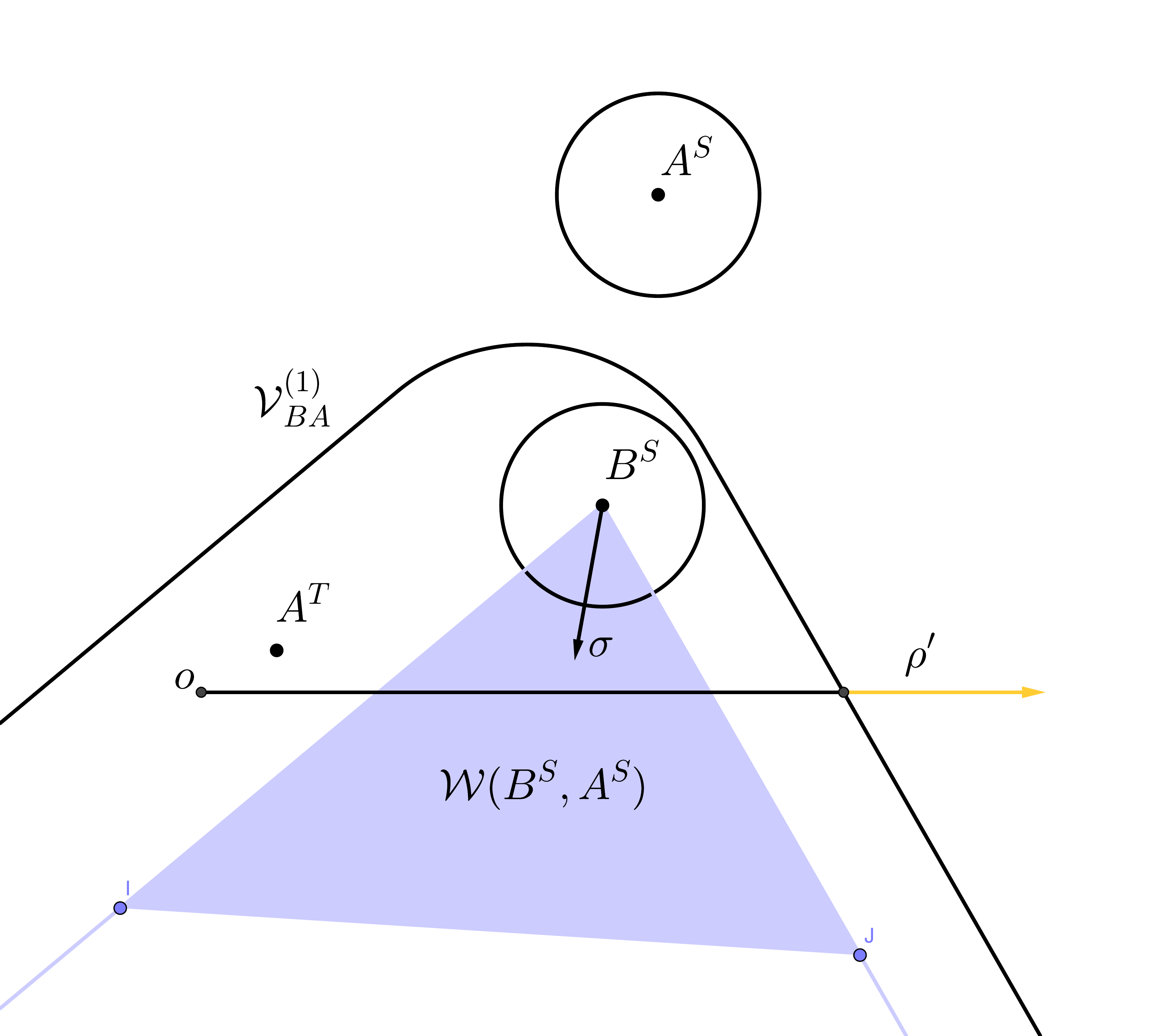}
	\caption{\sf 	
	According to $\Pi(M_\delta)$, $A$ preforms a motion before $B$. 
	Therefore, we need to rule out translations $\vec{v}$ for which $B$ has to perform a motion before $A$, whose locus is exactly the bad vippodrome $\V^{(1)}_{BA}$.
	Here $\delta$ is horizontal	and points to the right.
	The wedge $\W(B^S,A^S)$ (in blue), and thus also the bad vippodrome, cannot fully contain $\rho$, and thus there exists $\rho'$ (in orange) that is disjoint from the vippodrome.}
	\label{fig:bad_vippo}
\end{figure}

We claim that there exists a ray $\rho' \subseteq \rho$ such that $\rho' \cap V = \emptyset$ for every $V \in \V_{\rm bad}(\delta)$.
Indeed, let $V = \V^{(1)}_{BA}$, such that $A,B \in M_\delta$ and $A$ performs a motion before $B$ according to $\Pi(M_\delta)$; that is, $V \in \V_{\rm bad}(\delta)$.
By construction, if $A$ performs a motion before $B$ according to the itinerary, then $A^S$ is lexicographically larger than $B^S$, and so $B^S$ is to the left of (or has the same $x$-coordinate and is below) $A^S$; see Figure~\ref{fig:bad_vippo}.
Since the positive $x$-direction $\delta$ is generic, $D(A^S)$ and $D(B^S)$ cannot lie vertically above one another and have a common inner tangent.
Let $\sigma$ be the ray that emanates from $B^S$ in the direction from $A^S$ to $B^S$.
By our assumption, $\sigma$ points either directly downwards, or else to the left (contained in the open left vertical halfplane that contains $B^S$ on its right boundary).
Note that $\sigma$ is the mid-ray of the wedge $\W(B^S,A^S)$.
In the former case, the opening angle of $\W(B^S,A^S)$ is strictly smaller than $\pi$, and in the latter case, it is at most $\pi$.
In either case, $\W(B^S,A^S)$ is disjoint from the rightward-directed horizontal ray from $B^S$ (the ray in direction $\delta$).
This implies that $\W(B^S,A^S)$ cannot fully contain any rightward-directed ray.
Since $V = \W(B^S,A^S) \oplus D_2(o) - A^T$, the same claims hold for $V$ as well.
The argument for $\V^{(2)}_{BA}$ is similar.
In conclusion, $\rho$ must exit from every vippodrome of $\V_{\rm bad}(\delta)$, which establishes the claim.

Hence, there are infinitely many translations $\vec{v}$ in $\rho$ that do not belong to any vippodrome of $\V_{\rm bad}(\delta)$.
By construction, this implies that \ust($S$, $T+\vec{v}$) is feasible for every such $\vec{v}$ (with the valid itinerary induced by $\Pi(M_\delta)$).
Furthermore, the above holds for every generic direction $\delta$.
This completes the proof of the theorem.
\end{proof}

It is now fairly simple to devise an algorithm for finding a valid translation $\vec{v}$ and for constructing a valid itinerary from $S$ to $T +\vec{v}$.
First, choose a generic direction $\delta$, in $O(n \log n)$ time, and assume it to point in the positive $x$-direction.
Calculate $\Pi_\delta(S), \Pi_\delta(T)$ and $M_\delta$ in $O(n \log n)$ time.
Construct $\V_{\rm bad}(\delta)$ according to $M_\delta$, in $O(n^2)$ time.
Intersect all the vippodromes of $\V_{\rm bad}(\delta)$ with the positive $x$-axis, and find the rightmost intersection point $\vec{v}_{\max}$ with these vippodromes, which can be done in $O(n^2)$ time.
Any translation $\vec{v}$ on the $x$-axis to the right of $\vec{v}_{\max}$ has a valid itinerary from $S$ to $T+\vec{v}$, given by the order $\Pi(M_\delta)$.
The overall running time of this algorithm is therefore $O(n^2)$.

Note that limiting the range of valid translations to the ray to the right of $\vec{v}_{\max}$ may be too restrictive---any translation on the positive $x$-axis that lies outside the union of the bad vippodromes is valid, and we can simply return the one closest to the origin, say.
We will exploit this simple observation in the following section.

%% file: chapters/unlabeled_sa/unlabeled_sa.tex
\ifthesis
\section
[Unlabeled Version: Space-Aware Practical Solutions]%
{Unlabeled Version: \\ Space-Aware Practical Solutions}

In this \ssection we provide practical heuristic solutions for the different optimization variants of the unlabeled version, namely \saustv{$S$,$T$}, \sausta{$S$,$T$} and \saustd{$S$,$T$}, defined in full analogy to the corresponding versions for the labeled case.
As noted before, the unlabeled version seems to be harder than the labeled version, and so we only present heuristic approaches that guarantee optimality under more restrictive notions of validity.
In \Subsection~\ref{subsec:heuristics} we propose several optimization algorithms, each of which considers a different notion of validity, so that the running time improves as the optimality of the solution deteriorates.
We defer the details for the more involved case of minimizing the smallest enclosing disc, under the strictest notion of validity, to \Subsection~\ref{subsec:sed_1d}.
Our fastest algorithms are restricted to valid translations in any fixed direction $\delta$, such that the matching (between $S$ and $T$) and the itinerary are set according to $\delta$.
The algorithms run in $O(n^2 \log n)$ time for \saustv{$S$,$T$} and \sausta{$S$,$T$}, and in $O(n^2\alpha(n)\log n)$ time for \saustd{$S$,$T$}, where $\alpha(n)$ is the inverse Ackermann function.
In \Subsection~\ref{subsec:diameter}, we prove that these algorithms provide valid translations that are larger by at most a linear factor for the shortest translation or the radius of $\SED(S\cup(T+\vec{v}))$, for a suitable choice of $\delta$.
Assuming that the distance between any two points in each configuration is at least $2+\eps$, for some parameter $\eps>0$, we prove that the heuristic solution is larger by only a factor of $O(\frac{1}{\sqrt{\eps}})$.
Experimental results of an implementation of the heuristic algorithm, for the most restrictive notion of validity for \saustv{$S$,$T$}, are presented in \Subsection~\ref{subsec:implement}.
As shown, the algorithm (implemented in Python~3.7) solves different types of instances for the \saustv{$S$,$T$} problem, each with hundreds of discs, in seconds.
Moreover, the solutions are efficient and applicable for practical usage.

\else
\section
[Unlabeled Version: Space-Aware Practical Solutions]%
{\parbox[t]{15em}{Unlabeled Version: \\ Space-Aware Practical Solutions}}

In this \ssection we provide practical heuristic solutions for the different optimization variants of the unlabeled version, namely \saustv{$S$,$T$}, \sausta{$S$,$T$} and \saustd{$S$,$T$}, defined in full analogy to the corresponding versions for the labeled case.
In \Subsection~\ref{subsec:heuristics} we propose several optimization algorithms, each of which considers a different notion of validity, so that the running time improves as we make the problem more constrained.  
We defer the details for the more involved case of minimizing the smallest enclosing disc, under the strictest notion of validity, to \Subsection~\ref{subsec:sed_1d}.
Our fastest algorithms are restricted to valid translations in any fixed direction $\delta$, such that the matching (between $S$ and $T$) and the itinerary are set according to $\delta$.
In \Subsection~\ref{subsec:diameter}, we prove that these algorithms provide valid translations that are larger by at most a linear factor for the shortest translation or the radius of $\SED(S\cup(T+\vec{v}))$, for a suitable choice of $\delta$.
Assuming that the distance between any two points in each configuration is at least $2+\eps$, for some parameter $\eps>0$, we prove that the heuristic solution is larger by only a factor of $O(\frac{1}{\sqrt{\eps}})$.
Experimental results obtained with our implementation of the heuristic algorithm, for the most restrictive notion of validity for \saustv{$S$,$T$}, are presented in \Subsection~\ref{subsec:implement}.
\fi

\input{chapters/unlabeled_sa/heuristic}
\input{chapters/unlabeled_sa/SED_1d}
\input{chapters/unlabeled_sa/upper_bound}
\input{chapters/unlabeled_sa/implementation}

%% file: chapters/unlabeled_sa/heuristic.tex
\subsection{Heuristics for Short Valid Translations}
\label{subsec:heuristics}
The analysis in \Section~\ref{section:unlabeled}, while providing an abundance of valid translations, has the disadvantage that the valid translations that it yields are potentially too long (one needs to go sufficiently far away in the $\delta$-direction to get out of all the bad vippodromes).
This is undesirable with our space-aware objective in mind, where we seek short valid translations.
In this \ssection, we provide heuristics for finding shorter valid translations, thereby obtaining shorter heuristic solutions to \saustv{$S$,$T$}.
Similar techniques also yield heuristic solutions to \sausta{$S$,$T$} and to \saustd{$S$,$T$}.
Here \saustv{$S$,$T$}, \sausta{$S$,$T$} and \saustd{$S$,$T$} stand for the \saust{} problem under the three minimization criteria used in \Section~\ref{section:labeled_sa}.
The resulting algorithms are faster than those obtained for the labeled case (at the cost of not guaranteeing optimality over the entire space of valid translations).
As the unlabeled problem is much harder than the labeled problem (and we believe it to be NP-hard), we make no attempt at solving it exactly.
We exploit Theorem~\ref{theo:unlabeled_feasible}, which shows the existence of at least one valid translation for any generic direction $\delta$, and the algorithm, presented in \Section~\ref{section:unlabeled}, for finding such translations.
This machinery motivates the following algorithms.

In \Subsection~\ref{subsec:diameter} we will show how to choose a good direction $\delta$ for which we can give reasonable upper bounds on the length of the shortest valid translation, or of the valid translation $\vec{v}$ that minimizes the area of the axis-aligned bounding rectangle, or the smallest enclosing disc of $S \cup (T + \vec{v})$.
In practice, one might want to choose a sufficiently dense set of generic directions, in the hope of improving the quality of the following solutions.

For now, fix a generic direction $\delta$ for $S$ and $T$, and assume, for simplicity and with no loss of generality, that it is the positive $x$-direction.
Recall that it takes $O(n\log n)$ time to compute $\Pi_\delta(S)$, $\Pi_\delta(T)$, $M_\delta$, and $\Pi(M_\delta)$.
This transforms the problem to a labeled instance, according to $M_\delta$, with the additional constraint that we require the discs to move according to the order $\Pi(M_\delta)$.
Since the order is now fixed, it suffices to consider, as in \Section~\ref{section:unlabeled}, only the vippodromes in $\V_{\rm bad}(\delta)$.
We form the arrangement $\A_{\rm bad}$ of their boundaries, and collect all the faces of $\A_{\rm bad}$ that lie outside the union of the bad vippodromes.
This is the set of all the valid translations $\widetilde{Q}$, with respect to the matching $M_\delta$ and the order $\Pi(M_\delta)$.

Constructing $\widetilde{Q}$, that is, collecting the faces of $\A_{\rm bad}$ outside the union of the bad vippodromes, takes $O(n^4\log n)$, or $O(n^2\lambda_4(n^2))$ time, as shown in \Section~\ref{section:labeled}.%
\footnote{Recall that the optimization procedures in \Section~\ref{section:labeled_sa} require a preprocessing stage that computes $Q$ in $O(n^6)$ time.
This cost is only a consequence of the (expensive) need to test each face of $\VA$ for acyclicity, which is no longer needed under our new notion of validity.
This is why constructing $\widetilde{Q}$ is faster.
Note also that restricting the analysis to the bad vippodromes cleans the presentation, but is not the source of the above improvement, as the number of bad vippodromes is still half the number of all vippodromes.}
Finding an optimal translation in $\widetilde{Q}$ where we either minimize its length or the area of the axis-aligned bounding rectangle, can be done in $O(n^2\log n)$ additional time (see \Subsections \ref{subsec:labeled_short_vector} and \ref{subsec:labeled_aabr}).
Optimizing the size of the smallest enclosing disc, using the algorithm of \Subsection~\ref{subsec:labeled_sed}, is more expensive, and takes $O(n^5 \log n)$ time.%

We note that a major role of the direction $\delta$ is to define the orders $\Pi_\delta(S)$ and $\Pi_\delta(T)$ (and thus also the order $\Pi(M_\delta)$ of $M_\delta$), as the reverse lexicographical order, induced by $\delta$ and its orthogonal direction, of the disc centers.
The algorithms mentioned above will find an optimal valid translation, under this restricted notion of validity, within the entire set of valid translations $\widetilde{Q}$, not necessarily in direction $\delta$.

We can do better, at the cost of further constraining the set of valid translations, considering only the translations $\vec{v}$ in direction $\delta$, namely the intersection of $\widetilde{Q}$ and the ray emanating from the origin in direction $\delta$.
Recall the algorithm for finding a valid translation at the end of \Section~\ref{section:unlabeled}.
As noted at the end of that \ssection, it is likely that the translation $\vec{v}_{\max}$ computed there is not the shortest valid translation in direction $\delta$, according to $\Pi(M_\delta)$.
In order to find the shortest such valid translation, we again construct the bad vippodromes and intersect them with the positive $x$-axis.
Instead of finding the rightmost intersection point, we sort the resulting valid intervals along the $x$-axis (the ones that are free of all bad vippodromes), and output the leftmost valid point, which is clearly the shortest valid translation, under the present restricted setup.
This process is mildly slower than the original algorithm (see the end of \Section~\ref{section:unlabeled}), since we need to sort $O(n^2)$ intervals, and can be carried out in $O(n^2 \log n)$ time.

The above performance bounds also hold for finding the translation $\vec{v}$ that minimizes the area of the axis-aligned bounding rectangle of $D(S)\cup D(T+\vec{v})$: here, at each valid interval $I$ along the ray, we need to compute the minimum of a univariate function of constant complexity over $I$, which takes constant time.

The problem of minimizing the smallest enclosing disc of $D(S)\cup D(T+\vec{v})$ is more involved, and is discussed next.

%% file: chapters/unlabeled_sa/SED_1d.tex
\ifthesis
\subsection[Smallest Enclosing Disc for Translations Along a Line]%
{\parbox[t]{17em}{Smallest Enclosing Disc \\ for Translations Along a Line}}
\else
\subsection{Smallest Enclosing Disc for Translations Along a Line}
\fi
\label{subsec:sed_1d}

We next show how to adapt the algorithm for finding a valid translation that minimizes the radius of $D := \SED(S\cup(T+\vec{v}))$, which is described in \Subsection~\ref{subsec:labeled_sed} and to which we will refer as the \emph{original algorithm}, to the case where the translations are restricted to lie on a single line $\lambda$ in some generic direction $\delta$, and $\lambda$ passes through the origin.
Moreover, the matching $M_{\delta}$ and the order $\Pi(M_\delta)$ are fixed according to $\delta$.
We show that the algorithm can be implemented so that it runs in nearly quadratic time for this special case.

In what follows, we only highlight the necessary adjustments of the original algorithm to fit this special case.
The main differences come from searching for translations along $\lambda$ instead of searching for translations along the many vippodrome boundaries.

Recall that $\widetilde{Q}$ is the set of all valid translations according to the matching $M_{\delta}$ and the order $\Pi(M_\delta)$, i.e., the set of all translations that are not contained in any bad vippodrome of $\V_{\rm bad}(\delta)$.
We begin by computing the $O(n^2)$ intersection points of the bad vippodrome boundaries with $\lambda$, select those points that lie on $\bd \widetilde{Q}$, and sort them in their order along $\lambda$.
The cost of this step is $O(n^2\log n)$.
Let $F$ denote this sorted sequence.
Ignoring degenerate situations in which $\lambda$ is tangent to some vippodrome boundary, $F$ partitions $\lambda$ into an alternating sequence of valid and invalid intervals.
We store $F$ in a balanced binary search tree $\D_F$, similarly to that in Lemma~\ref{lemma:dq}.
Using $\D_F$, we can answer two types of queries:
(i) Given a translation $\vec{v}$, determine whether $\vec{v}$ is valid.
(ii) Given an interval $I\subset\lambda$, find the smallest or largest
point of $F\cap I$.
Each query of either type can be answered in $O(\log n)$ time.

Case (i) in this restricted variant of the algorithm is almost identical to the original algorithm ---
the sole difference is the need to calculate the intersection $I = V(\xi) \cap \lambda$, and then test, by searching in $\D_F$, whether $I \cap F \ne\emptyset$.
If so, the translation $\vec{v}\in I \cap F$ is valid, and we can halt the algorithm, returning $\vec{v}$ and $\SED(S)$.
This case can be trivially executed in $O(n \log n)$ time.

Consider next Case (ii).
Recall that for each translation $\vec{v}$, the functions $f_t,g_t$ are defined to be the endpoints of a subinterval of $e$, such that $t+\vec{v}$ is enclosed by the disc centered at any point of the subinterval and passes through $s_1$ and $s_2$.
Since the functions $f_t,g_t$ are simpler in this case, we can more precisely bound the running time of the algorithm.
(We note though that the difference in running time is negligible.)
We construct the sandwich region $\Xi$ where $\vec{v}$ varies along the line $\lambda$, defined by:
$$
\max_{t\in T} f_t(\vec{v}) \le \xi \le \min_{t\in T} g_t(\vec{v}).
$$
This is a univariate sandwich computation, which takes $O(\lambda_{s+1}(n)\log n)$ time,%
\footnote{The reason why the index in the above bound is $s+1$ is that each of the domains of the partial functions $f_t(\vec{v})$, $g_t(\vec{v})$ is an interval with at least one endpoint coinciding with an endpoint of $e$; see \cite{DBLP:books/daglib/0080837} for details.}
where $s$ is the maximum number of intersections between a pair of functions $f_t(\vec{v})$, $f_{t'}(\vec{v})$, or a pair $g_t(\vec{v})$, $g_{t'}(\vec{v})$, for $t, t'\in T$.
Any such intersection occurs at a translation $\vec{v}$ at which $s_1, s_2, t+\vec{v}, t'+\vec{v}$ become cocircular. 
If we assume, without loss of generality, that $\lambda$ is the $x$-axis, and use $v$ to denote the $x$-coordinate of $\vec{v}$, then the cocircularity property is expressed by the equation
$$
\begin{vmatrix}
1 & s_{1x} & s_{1y} & s_{1x}^2 + s_{1y}^2 \\
1 & s_{2x} & s_{2y} & s_{2x}^2 + s_{2y}^2 \\
1 & t_{1x}+v & t_{1y} & (t_{1x}+v)^2 + t_{1y}^2 \\
1 & t_{2x}+v & t_{2y} & (t_{2x}+v)^2 + t_{2y}^2
\end{vmatrix}
= 0 .
$$
Subtracting the fourth row from the third one, one can easily verify that the resulting equation is quadratic in $v$, implying that $s=2$.
Hence the sandwich region can be constructed in $O(\lambda_3(n)\log n) = O(n\alpha(n)\log n)$ time.

To find this pair, we break the portion of the upper envelope within $\Xi$ into maximal connected subarcs, so that each arc $\gamma$ lies on the graph of a single function $f_t(\vec{v})$.
We further break $\gamma$ into $O(1)$ subarcs, so that in each of them $f_t(\vec{v})$ is monotone. 

We search with each subarc $\gamma'$ in $\D_F$, compute the smallest (resp., largest) point of $\gamma'\cap F$, when $f_t$
is increasing (resp., decreasing) over $\gamma'$, and add it to the list of candidate solutions.
We return the translation $\vec{v}$ from the list that has smallest value of the corresponding $f_t(\vec{v})$.

The overall cost of this step, for a fixed pair $(s_1,s_2)$, continues to be $O(n\alpha(n)\log n)$.
Summing over all the pairs $(s_1,s_2)$, the overall cost of handling Case (ii) is $O(n^2\alpha(n)\log n)$.

We finally turn to Case (iii).
Recall that we want to compute the intersection of $K = K_{s_0,t_0}$ with $\lambda$, for any pair $s_0 \in S$ and $t_0 \in T$, and then find a valid translation (contained in the intersection) that is closest to $s_0-t_0$.
In order to do that efficiently, we refrain from calculating $K$ explicitly, and so we take a different approach as follows.

We seek the translations $\vec{v}$ so that the center of the diametral disc $D(\vec{v})$ determined by $s_0$ and $t_0 +\vec{v}$, which is $c_{s_0,t_0}(\vec{v}) = \frac12(s_0+t_0)+\frac12 \vec{v}$, satisfies:
(a) $c_{s_0,t_0}(\vec{v})$ lies in the cell $V(s_0)$ of $s_0$ in $\FVD(S)$ (thus guaranteeing that $S\subset D(\vec{v})$), and
(b) $c_{s_0,t_0}(\vec{v})$ lies in the cell of $t_0+\vec{v}$ in $\FVD(T+\vec{v}) =  \FVD(T)+\vec{v}$ (thus guaranteeing that $T+\vec{v} \subset D(\vec{v})$).
This latter cell is $V(t_0) + \vec{v}$, where $V(t_0)$ is the cell of $t_0$ in $\FVD(T)$.
In other words, $\vec{v}$ has to satisfy
\begin{gather*}
\frac12(s_0+t_0) + \frac12\vec{v}  \in V(s_0) \\
\frac12(s_0+t_0) + \frac12\vec{v}  \in V(t_0) + \vec{v},
\end{gather*}
or
\begin{gather*}
\vec{v}  \in 2V(s_0)    - (s_0+t_0) \\
\vec{v}  \in - 2V(t_0)  + (s_0+t_0)  ,
\end{gather*}
where $2V(s_0) = \{2z \mid z\in V(s_0)\}$ and similarly for $2V(t_0)$.
Note that the cells $V(s_0)$ and $V(t_0)$ are convex.

These observations lead to the following procedure:
We first compute $\FVD(S)$ and $\FVD(T)$, in $O(n\log n)$ time, and enlarge both diagrams by the factor $2$.
We then preprocess each of the (convex polygonal) cells of both (enlarged) diagrams for efficient line intersection queries, where all the query lines have a fixed direction (that of $\lambda$).
This preprocessing is trivial and takes linear time.
Answering a query is equally trivial and takes $O(\log n)$ time.

We now iterate over all pairs $s_0\in S$, $t_0\in T$.
For each such pair, we
(a) take the line $(s_0+t_0) + \lambda$ and intersect it with $2V(s_0)$, to obtain an interval $I_{s_0,t_0}$, and
(b) take the line $(s_0+t_0) - \lambda$ (since we assume that $\lambda$ passes through the origin, this is the same as the line $(s_0+t_0) + \lambda$) and intersect it with $2V(t_0)$, to obtain an interval $J_{s_0,t_0}$.
All this takes $O(\log n)$ time.
We map these intervals back into $\lambda$ and intersect the resulting intervals, that is, we obtain the interval
$$
K^{\lambda}_{s_0,t_0} := \Bigl( I_{s_0,t_0} - (s_0+t_0) \Bigr) \cap
\Bigl( (s_0+t_0) - J_{s_0,t_0} \Bigr) ,
$$
along $\lambda$, in additional constant time.

Note that our goal is to find in $K^{\lambda}_{s_0,t_0}$ a valid translation that minimizes the diameter of $D_{s_0,t_0}(\vec{v})$, which is $\|s_0-t_0-v\|$.
That is, we seek the translation that is closest to $s_0-t_0$ among all valid translations in $K^{\lambda}_{s_0,t_0}$.
To do so, we first compute the closest translation to $s_0-t_0$ in $K^{\lambda}_{s_0,t_0}$ (which may be an endpoint of $K^{\lambda}_{s_0,t_0}$), and check whether it is valid.
If so, we output it as a candidate solution.
Otherwise, we break $K^{\lambda}_{s_0,t_0}$ into two subintervals at the point of $\lambda$ closest to $s_0-t_0$ (if that point lies in $K^{\lambda}_{s_0,t_0}$), search $\D_F$ with each subinterval, and output the first or last point of $F$ in that interval, as appropriate.
We repeat this step for each pair $s_0\in S$ and $t_0\in T$, and return the candidate translation that minimizes the corresponding distance $\|s_0-t_0-v\|$.
The overall cost of this step is $O(n^2\log n)$.

In summary, we have:
\begin{theorem} \label{sed-1d}
	Finding the valid translation $\vec{v}\in\lambda\cap \widetilde{Q}$ that minimizes the radius of $\SED(S\cup(T+\vec{v}))$ can be done in $O(n^2\alpha(n)\log n)$ time.
\end{theorem}

%% file: chapters/unlabeled_sa/upper_bound.tex
\subsection{Bounding the Heuristic Solutions}
\label{subsec:diameter}
The analysis in the preceding \ssubsection provides heuristics for obtaining short valid translations, but gives no guarantees on how well we approximate the optimal valid translation.
In this \ssubsection we show how to choose a good direction $\delta$ for which we can give a reasonable (and sometimes asymptotically optimal) upper bound on the length of the shortest valid translation in direction $\delta$, or of a valid translation $\vec{v}$ in direction $\delta$ that minimizes \SED{($S\cup (T+\vec{v})$)}.%
\footnote{Similar techniques can be applied to the variant of optimizing the area of the axis-aligned bounding rectangle.
To keep the presentation somewhat shorter, we do not include this case in this \ssection.}
Recall that we denote by $r(P)$ the radius of the smallest enclosing disc of a set of points $P$.
Our main result is:

\begin{theorem}
\label{theo:bound_length}
Let $S$ and $T$ be two valid configurations, of $n$ points each, such that $S$ and $T$ share the centers of their smallest enclosing discs.
There exists a translation $\vec{v}$ such that ${\rm \ust}(S,T+\vec{v})$ is feasible and $|\vec{v}|=O((r(S)+r(T)) n)$.
The same asymptotic bound also applies to $r(S\cup (T+\vec{v}))$.
\end{theorem}
\begin{proof}
Consider the discs in $D(S)$ (the case of $T$ is essentially identical).
We define the following set $\Delta_S$ of directions.
For each pair $s_1,s_2 \in S$ of neighbors in the (nearest-neighbor) Delaunay triangulation of $S$, we include in $\Delta_S$ the directions of the two inner tangents of $D(s_1)$ and $D(s_2)$, obtaining at most $12n-24$ directions (each tangent is directed both ways).
We obtain a similar set $\Delta_T$ of at most $12n-24$ directions from the pairs of neighbors in the Delaunay triangulation of $T$.

\begin{lemma}
\label{lem:separ}
For any pair of unit discs in $D(S)$ whose centers are not Delaunay neighbors, the angle between their inner tangents in the wedges that contain neither of the two discs is at least $\pi/2$. The same holds for $D(T)$.
\end{lemma}
\begin{proof}
Let $p_1$ and $p_2$ be two points in $S$ so that $p_1$ and $p_2$ are not Delaunay neighbors.
Hence the diametral disc determined by $p_1$ and $p_2$ must contain another point $q\in S$, so the angle $\angle p_1qp_2$ is at least $\pi/2$, and so $p_1$ and $p_2$ are at least $2\sqrt{2}$ apart (recall that the discs are interior disjoint).
See Figure~\ref{fig:delaunay} for an illustration.
The argument for $T$ is identical, and the lemma follows.
\end{proof}

\begin{figure}
    \centering
    \includegraphics[trim=0 0 0 0, clip, width=0.7\textwidth]{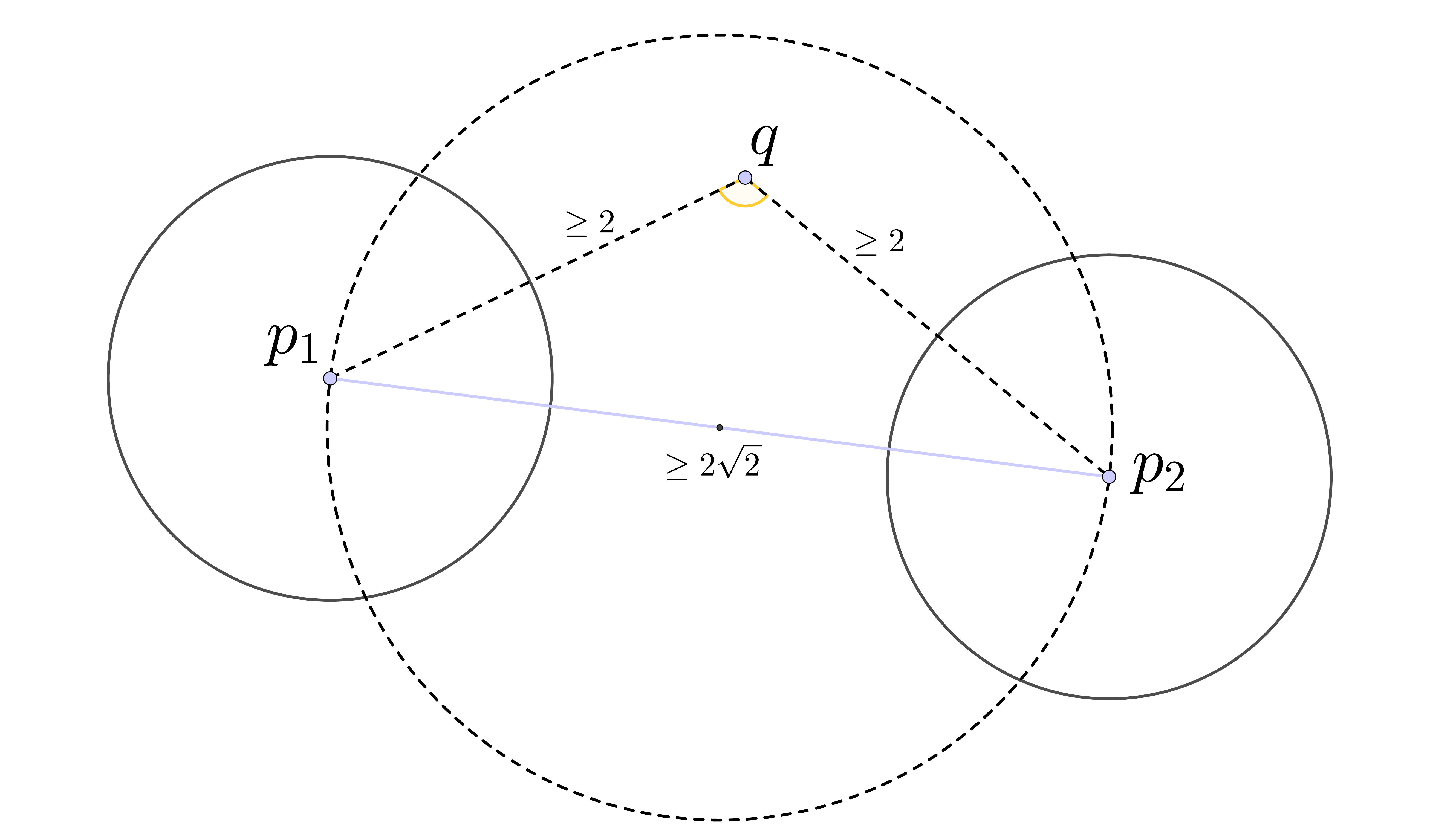}
    \caption{\sf The segment $p_1p_2$ (in blue) is not a Delaunay edge, so its diametral disc contains another point $q\in S$, and then $\angle p_1qp_2 \geq \frac{\pi}{2}$.}
    \label{fig:delaunay}
\end{figure}

Note that the lemma holds for any pair of discs whose centers are not neighbors in the Gabriel graph of $D(S)$ or of $D(T)$.

Let $\delta$ be a direction whose smallest angle from any direction in $\Delta_S \cup \Delta_T$ is as large as possible.
In particular, we can choose $\delta$ such that the angle it forms with the direction of any of the Delaunay inner tangents is $\Omega(1/n)$;
actually, the reasoning above implies a lower bound of at least $\pi/(24n)$.
Observe that $\delta$ is generic for $S$ and $T$, and so, by Theorem~\ref{theo:unlabeled_feasible}, there exists a valid translation in direction $\delta$.
By rotating the plane, we may assume, as before, that $\delta$ is the positive $x$-direction.
We compute the reverse lexicographical orders $\Pi_\delta(S)$, $\Pi_\delta(T)$, as defined in \Section~\ref{section:unlabeled}, and obtain the corresponding matching $M_\delta$ and its order $\Pi(M_\delta)$.

\begin{lemma}
\label{lemma:bound_on_ray}
The translation $\vec{v} = ((r(S)+r(T)+2)(1+8n),0)$ is a valid translation (in direction $\delta$) with respect to the matching $M_\delta$ and its order $\Pi(M_\delta)$.
\end{lemma}
\begin{proof}
By the analysis of \Section~\ref{section:unlabeled}, it suffices to show that $\vec{v}$ does not belong to any bad vippodrome.
For this, it suffices to show that any bad vippodrome intersects the positive $x$-axis to the left of $\vec{v}$.
Consider such a bad vippodrome, say $\V^{(1)}_{AB}$, where $A,B \in M_\delta$, so that $A$ appears after $B$ in $\Pi(M_\delta)$.
Symmetric arguments apply to bad vippodromes of the form $\V^{(2)}_{AB}$.
Assume, without loss of generality, that $B^S$ is above, or at the same height, as $A^S$ (note that, by construction, this is always the case when $A^S$ and $B^S$ have the same $x$-coordinate).
Since $A$ appears after $B$ in $\Pi(M_\delta)$, $A^S$ is lexicographically smaller than $B^S$, so the direction of $\overrightarrow{B^S A^S}$ points to the left (or vertically down); recall the proof of Theorem~\ref{theo:unlabeled_feasible}.
Let $r_1,r_2$ be the two rays of $\V^{(1)}_{AB}$, and assume that $r_1$ intersects the $x$-axis to the right of the intersection point of $r_2$ with the $x$-axis ($r_1$ is in direction of the tangent $\tau^+(B^S,A^S)$); see Figure~\ref{fig:upper_bound}.
Let $p$ be the point that $r_1$ emanates from, and let $q = (q_x,0)$ be the intersection point of $r_1$ and the $x$-axis.
We observe that the $y$-component of the direction of $r_1$ is negative.
Indeed, $r_1$ forms an angle of at most $\pi/2$ with the direction of $B^SA^S$, which is leftward-directed or points down; in the latter case the angle is strictly smaller than $\pi/2$, as easily follows from our choice of $\delta$.
It follows that $q_x$ increases as $p$ moves either to the right or up.
By the vippodrome construction, and the locations of $S$ and $T$, both the $x$- and the $y$-coordinates of $p$ cannot exceed $K := r(S)+r(T)+2$.
We may thus assume that $p=(K,K)$; for any other point, $q_x$ is smaller.

\begin{figure}
    \centering
    \includegraphics[trim=0 290 0 200, clip, width=0.7\textwidth]{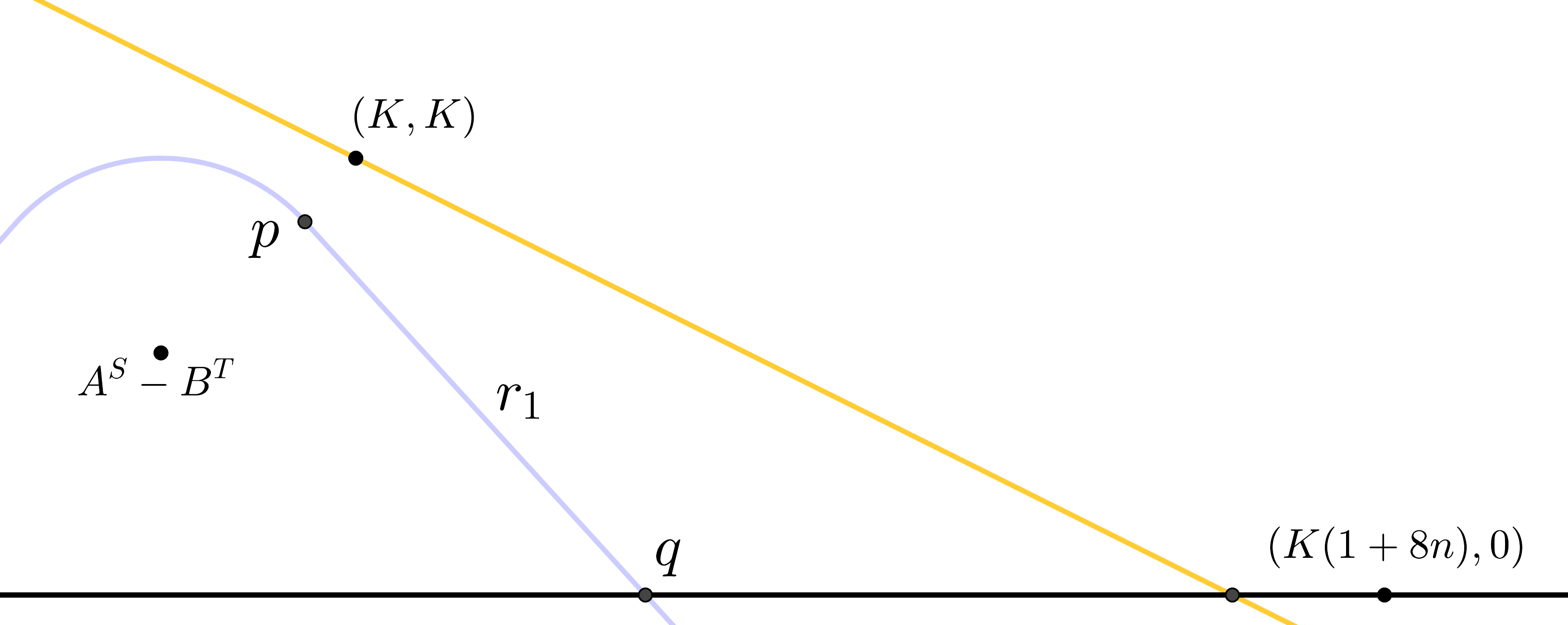}
    \caption{\sf The $x$-axis is drawn in black, a bad vippodrome $\V^{(1)}_{AB}$ in blue. The orange line, with slope $-\tan\frac{\pi}{24n}$, meets the $x$-axis to the right of any bad vippodrome of the first type, which is directed downwards (and to the left), like $\V^{(1)}_{AB}$.}
    \label{fig:upper_bound}
\end{figure}

If the $x$-component of the direction of $r_1$ is non-positive, $q_x \leq K$,
so we may assume that it is positive.
If $A^S$ and $B^S$ are not neighbors in the Delaunay diagram of $S$, the angle of $\W(A^S, B^S)$ is at most $\frac{\pi}{2}$, by Lemma~\ref{lem:separ}.
Thus, the slope of $r_1$ is at most $-1$.
If $A^S$ and $B^S$ are neighbors in the Delaunay diagram of $S$, then the directions of their common inner tangents are in $\Delta_S$, and so, by the choice of $\delta$, the slope of $r_1$ is at most $-\tan\frac{\pi}{24n}$.
We then again may assume that the slope of $r_1$ is $-\tan\frac{\pi}{24n}$, as for any smaller slope, $q_x$ is smaller.

The supporting line of $r_1$, according to our upper bounding assumptions, is
$$ y + x \tan\frac{\pi}{24n} - K \left( 1 + \tan\frac{\pi}{24n} \right) = 0,$$
and thus
$$ q_x = K \left( 1 + \frac{1}{\tan\frac{\pi}{24n}} \right) < K(1+8n),$$
where one can verify that the inequality holds for any $n\ge 1$.
This is the inequality asserted in the lemma.
\end{proof}
Repeating a symmetric argument for bad vippodromes of the second type,
Theorem~\ref{theo:bound_length} now follows readily, both for the length of $\vec{v}$ and for $r(S\cup(T+\vec{v}))$, as both of them are clearly $O((r(S)+r(T))n)$.
\end{proof}

The physical space needed for the reconfiguration is worse by the factor $O(n)$,
compared to the ideal bound $O(r(S)+r(T))$, which is (asymptotically) the minimum value of $r(S \cup (T + \vec{v}))$, over all translations $\vec{v}$.
Interestingly, we can attain this bound asymptotically if the discs of $D(S)$, as well as the discs of $D(T)$, are sufficiently separated.
That is, we have:

\begin{theorem}
\label{theo:bound_sep}
Let $S$ and $T$ be two valid configurations, of $n$ points each, such that $S$ and $T$ share the centers of their smallest enclosing discs.
Assume that there exists a fixed constant
$\eps>0$, so that the distance between any pair of points in $S$, or any pair of points in $T$, is at least $2+\eps$.
Then, for \emph{any} direction $\delta$, there exists a translation $\vec{v}$ in direction $\delta$, such that ${\rm \ust}(S,T+\vec{v})$
is feasible and $|\vec{v}|=O((r(S)+r(T))/\sqrt{\eps})$.
The same asymptotic bound also holds for $r(S\cup (T+\vec{v}))$.
\end{theorem}
\begin{proof}
Observe that the separation property guarantees that every direction is generic.
As is easily checked, the angle between the inner tangents of any pair of discs in $D(S)$,
within the wedge containing none of the discs, is at least $c \sqrt{\eps}$, for a suitable constant $c$.
Hence, the opening angle of any vippodrome is at most $\pi-c\sqrt{\eps}$. This implies that, for any direction $\delta$, and for any bad vippodrome, with respect to $\Pi(M_\delta)$, the angle that its ray $r_1$, in the notation of the proof of Lemma~\ref{lemma:bound_on_ray}, forms with the $\delta$-direction is at least $\frac{c}{2}\sqrt{\eps}$.
Following the same analysis as in the proof of Lemma~\ref{lemma:bound_on_ray}, we can show that there exists a valid translation $\vec{v}$ in direction $\delta$,
whose length is $O((r(S)+r(T))/\sqrt{\eps})$.
This also bounds $r(S\cup(T+\vec{v}))$.
\end{proof}

Note that Theorem~\ref{theo:bound_sep} is stronger than Theorem~\ref{theo:bound_length} also in that it holds for \emph{every} direction $\delta$, whereas Theorem~\ref{theo:bound_length} only holds for restricted values of $\delta$.

%% file: chapters/unlabeled_sa/implementation.tex
\ifthesis
\subsection[Implementation and Experimentation with the Heuristic Algorithm]%
{\parbox[t]{17em}{Implementation and Experimentation \\ with the Heuristic Algorithm}}
\else
\subsection[Implementation and Experimentation with the Heuristic Algorithm]{Implementation and Experimentation \\ with the Heuristic Algorithm}
\fi
\label{subsec:implement}


\newcommand{\imagescalse}{0.09}
\begin{table}[h]
	\caption{\sf Different input types of \ust{}.
		For each input type, the configurations are first presented separated, for better visualization, then in their initial positions (sharing the centers of their smallest enclosing discs) and with $T$ translated according to an approximate shortest valid translation (in red), produced by our heuristic algorithm.
	}
	\scriptsize
	\renewcommand{\arraystretch}{1.5}
	\newlength{\cellWidth}\setlength\cellWidth{25ex}
	\label{tab:experiments}
	\centering
	\begin{tabular}{|l|c|c|c|l|l|l|}
		\hline
		\textbf{Conf.} & \textbf{$D(S)$/$D(T)$} & \textbf{Initial} & \textbf{Translated} & \textbf{$n$} & \textbf{$r(S)+r(T)$} & \textbf{$|\vec{v}|$}   \\
		\hline
		\multirow{4}{*}{\textbf{Circle}} &
		\multirow{4}{*}{\includegraphics[scale=\imagescalse]{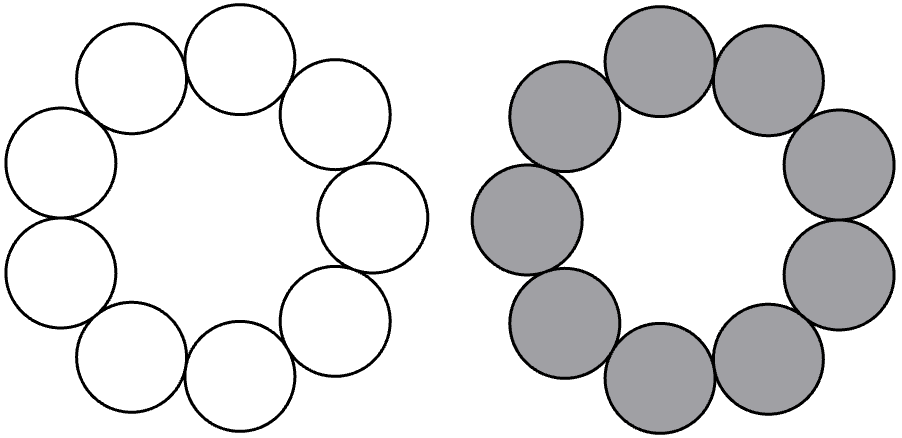}} &
		\multirow{4}{*}{\includegraphics[scale=\imagescalse]{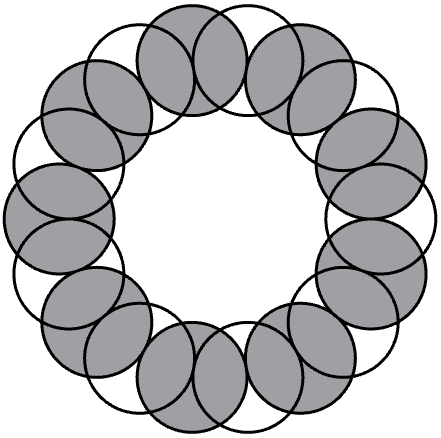}} &
		\multirow{4}{*}{\includegraphics[scale=\imagescalse]{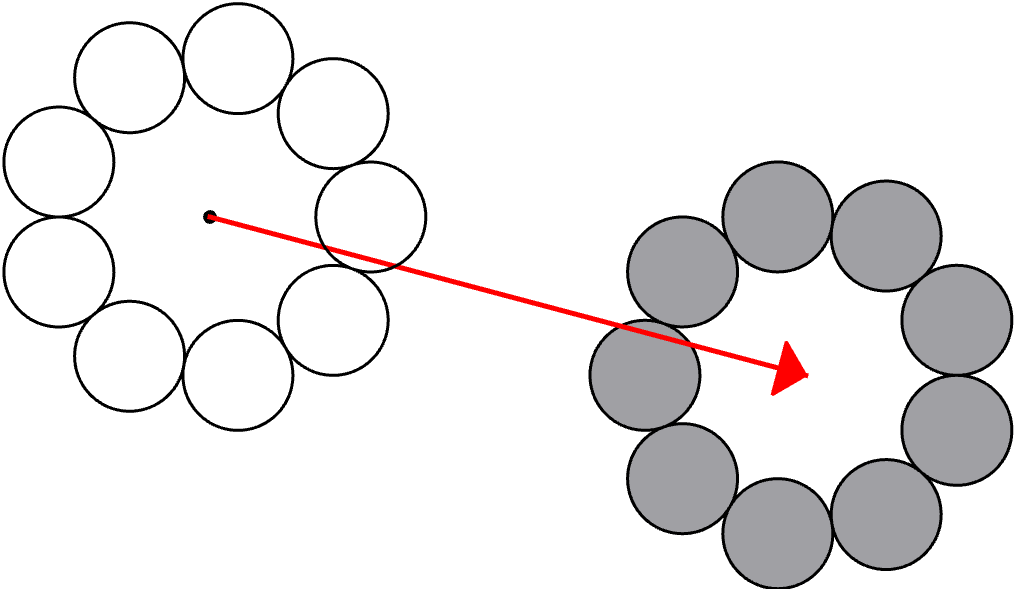}} &
		100   & 65.67  & 190.19   \\
		& & & & 200   & 129.32 & 376.24   \\
		& & & & 500   & 320.31 & 913.79   \\
		& & & & 1,000 & 638.60 & 1,757.26 \\
		\hline
		\multirow{4}{*}{\textbf{Packing}\footnotemark{}} &
		\multirow{4}{*}{\includegraphics[scale=\imagescalse]{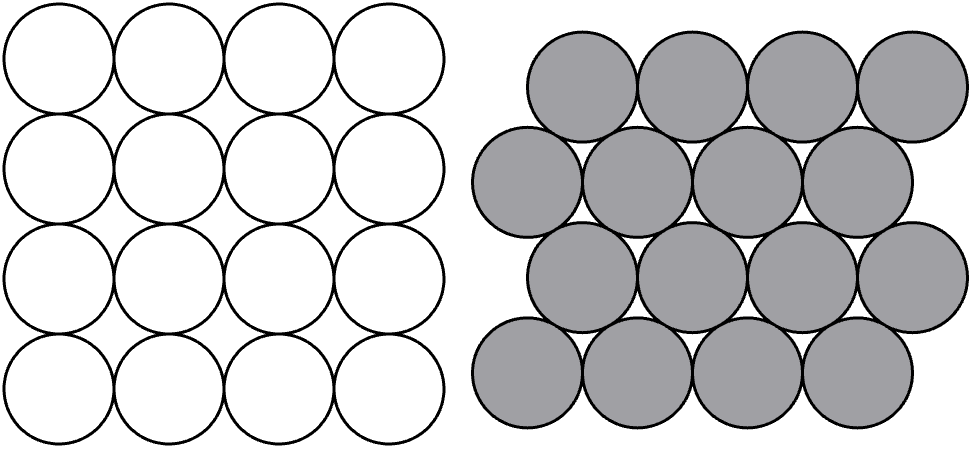}} &
		\multirow{4}{*}{\includegraphics[scale=\imagescalse]{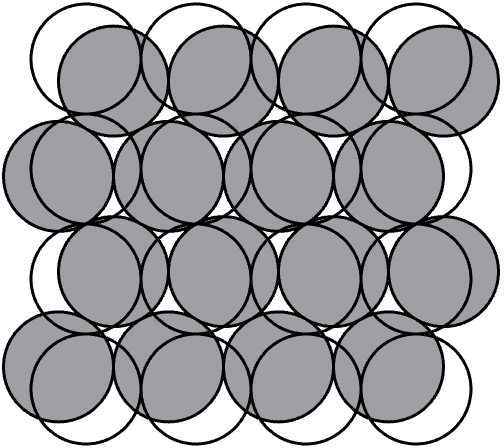}} &
		\multirow{4}{*}{\includegraphics[scale=\imagescalse]{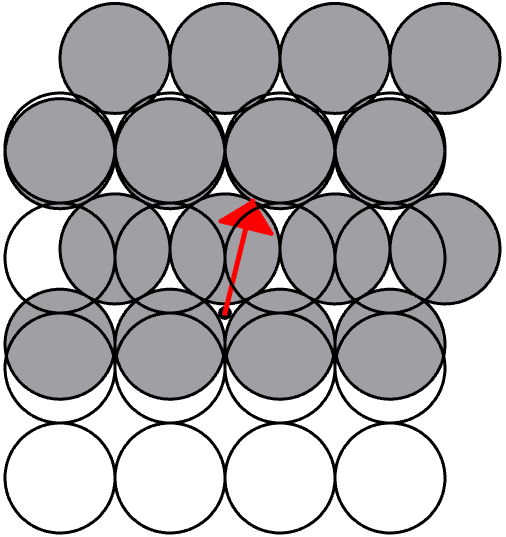}} &
		100   & 27.01 &  5.53  \\
		& & & & 210   & 38.88 &  2.18  \\
		& & & & 506   & 60.76 &  3.25  \\
		& & & & 1,024 & 87.22 & 19.75 \\
		\hline
		\multirow{4}{*}{\textbf{Cross}} &
		\multirow{4}{*}{\raisebox{-1\height}{\includegraphics[scale=\imagescalse]{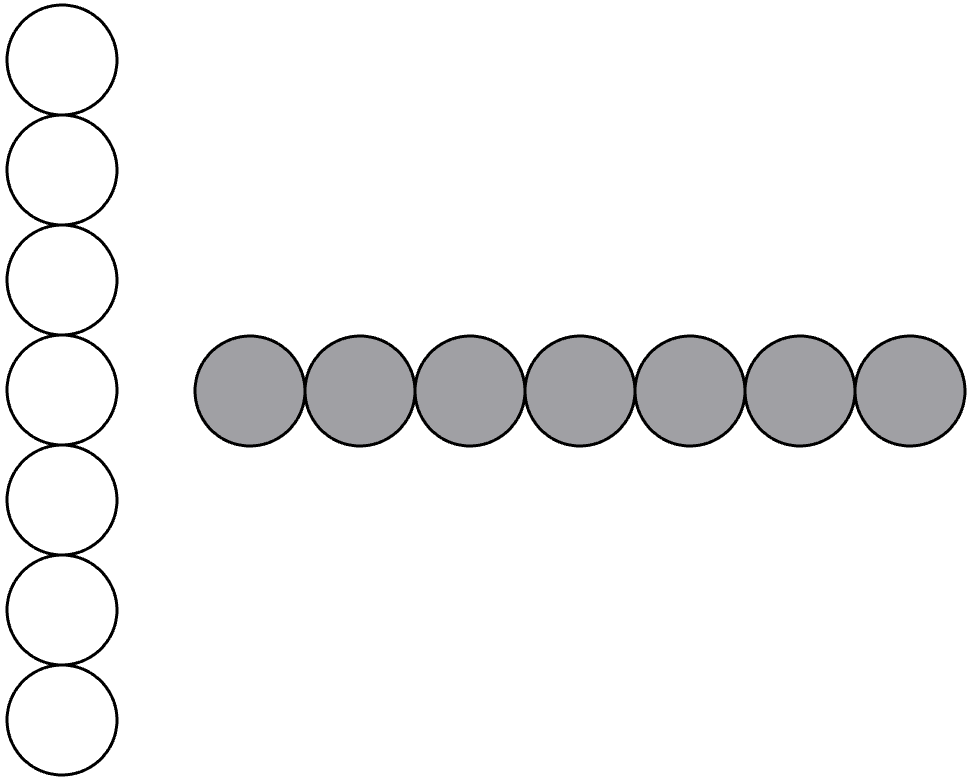}}} &
		\multirow{4}{*}{\raisebox{-1\height}{\includegraphics[scale=\imagescalse]{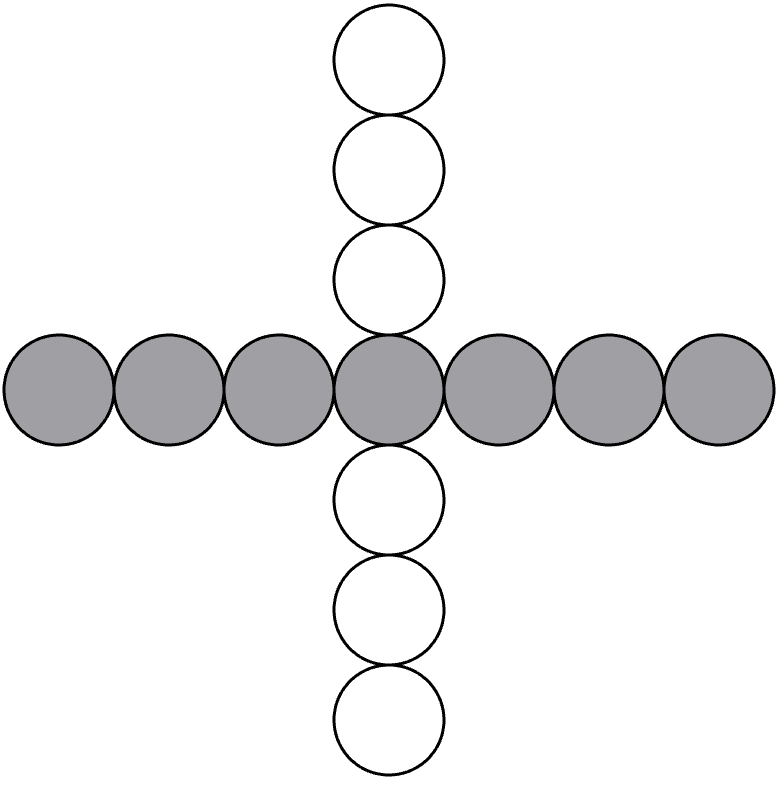}}} &
		\multirow{4}{*}{\raisebox{-1\height}{\includegraphics[scale=\imagescalse]{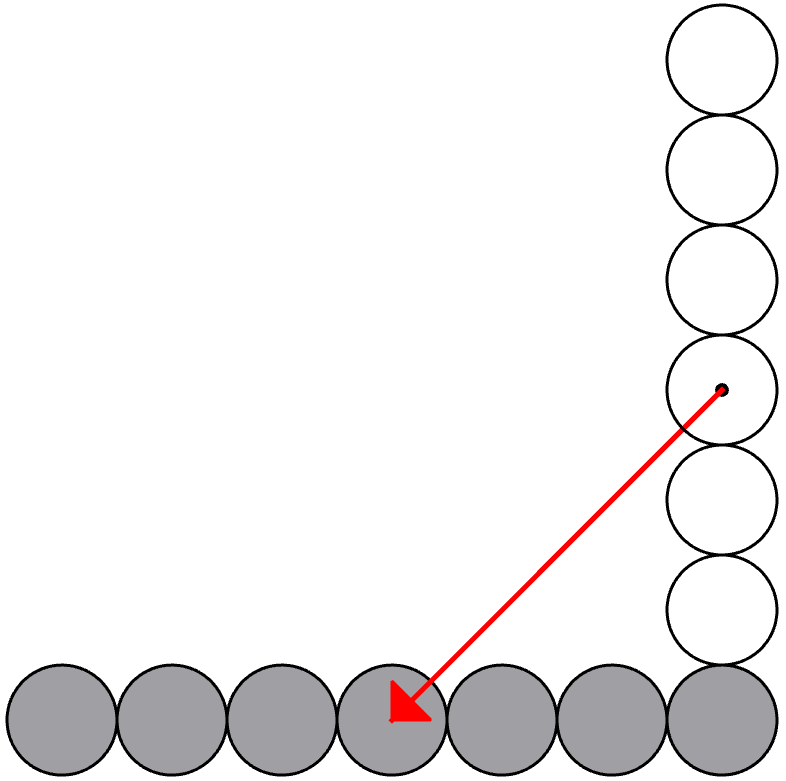}}} &
		100   & 200   &   140.07 \\[1ex]
		& & & & 200   & 400   &   281.43 \\[1ex]
		& & & & 500   & 1,000 &   706.15 \\[1ex]
		& & & & 1,000 & 2,000 & 1,413.47 \\[1ex]
		\hline
		\multirow{4}{*}{\textbf{Random}\footnotemark{}} &
		\multirow{4}{*}{\includegraphics[scale=\imagescalse]{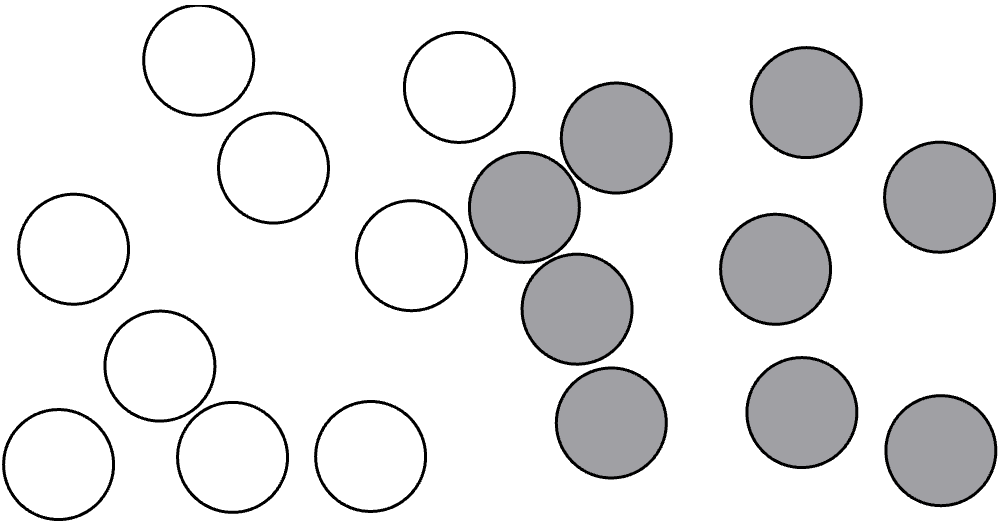}} &
		\multirow{4}{*}{\includegraphics[scale=\imagescalse]{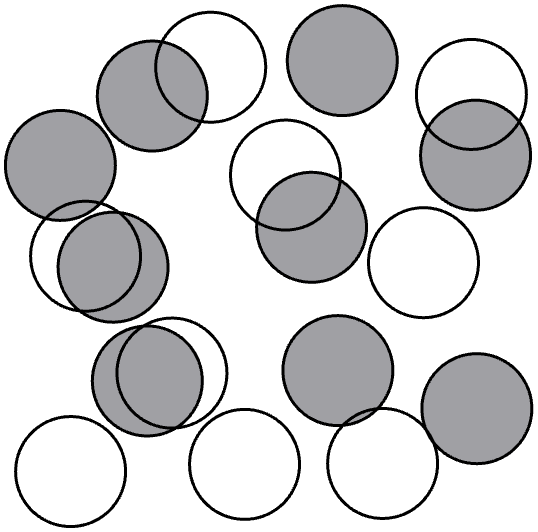}} &
		\multirow{4}{*}{\includegraphics[scale=\imagescalse]{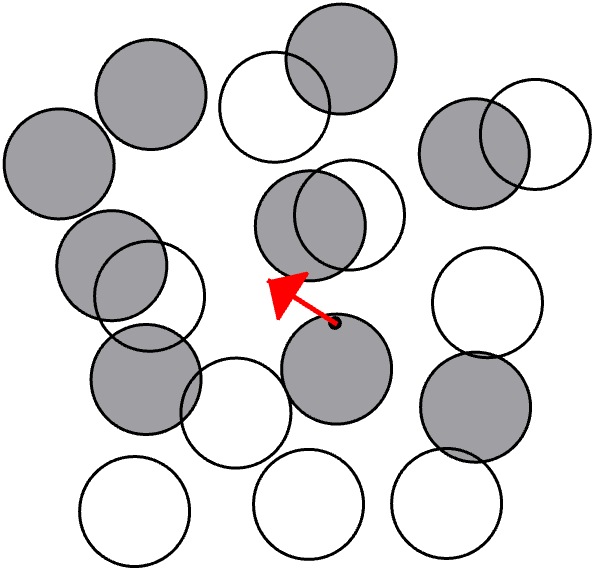}} &
		100   &  36.78 &  16.96 \\
		& & & & 200   &  51.70 &  34.01 \\
		& & & & 500   &  82.26 &  78.49 \\
		& & & & 1,000 & 116.06 & 147.61 \\
		\hline
	\end{tabular}
\end{table}
\addtocounter{footnote}{-2}
\stepcounter{footnote}
\footnotetext{The numbers $n$ are chosen such that the start configuration in each instance will be as square-like as possible.}
\stepcounter{footnote}
\footnotetext{The results are averaged over 10 different instances, one of which is depicted.}

We implemented the heuristic algorithm for finding an approximate shortest valid translation for {\rm \ust} instances, as outlined in \Subsection~\ref{subsec:heuristics}, by choosing a random direction $\delta$ (which is generic with probability $1$), fixing a matching $M_\delta$ and itinerary $\Pi(M_\delta)$ accordingly, and then finding the shortest translation (in direction $\delta$) that is valid in this more restricted setting.
Our program is written in Python~3.7, and the experiments that we report below were carried out on an Intel Core i7-7500U CPU clocked at~2.9~GHz with~24~GB of RAM.

We consider the following input types:
\begin{enumerate}
    \item \textbf{Circle}: the points of the configurations are densely placed on the circumference of a circle.
    The discs of $D(T)$ are slightly rotated (by $\frac{\pi}{n}$) in order to avoid an easy matching.
    
    \item \textbf{Packing}:
    the discs of $D(S)$ are placed in a squared grid.
    The discs of $D(T)$ are placed in a Kepler's packing.

    \item \textbf{Cross}: the discs of $D(S)$ (resp., $D(T)$) are tightly placed along a vertical (resp., horizontal) line.
    \item \textbf{Random}: both configurations are sampled uniformly at random%
    \footnote{The random choice of each configuration is modified so as to ensure that they are valid---no two points are at distance smaller than $2$.
    Random choices that violate this property are discarded and replaced by other random choices.}
    from a square of size $\left(2.6 \sqrt{n} \times 2.6 \sqrt{n}\right)$.
    For each configuration size, we average the results over 10 different runs.
\end{enumerate}
Table~\ref{tab:experiments} shows the results obtained with our implementation for four different types of input, with the number of discs per type ranging between $100$ and $1{,}024$.
For each input, we tried $1{,}000$ different directions $\delta$, and in the table we compare the shortest valid translation that the algorithm produced (over all different directions) with the asymptotically optimal value $r(S)+r(T)$.

In Figure~\ref{fig:times} we present the running times of the implementation for all four kind of input instances.
For the Circle, Packing, and Cross inputs, for $n=100i$ discs, $i=1,\ldots,10$, we optimize in $10$ random directions $\delta$. For the Random input, for $n=100i$ discs ($i=1,\ldots,10$), we choose $10$ random configurations of $n$ discs (see Table~\ref{tab:experiments}), and for each of them, we optimize in $10$ random directions $\delta$, for a total of $100$ runs per input size $n$.
As expected, the running time of the implementation is slightly super-quadratic.
The running times for the other three kinds of input behave remarkably similarly.
Our program runs in about 30 seconds on inputs with $1{,}000$ discs; notice that the number of bad vippodromes in such instances is $999{,}000$.

\newcommand{\timesubfigurewidth}{1}
\newcommand{\timescaption}[1]{\sf The #1 input type.}
\begin{figure}
	\centering
	\begin{subfigure}{.45\linewidth}
		\centering
		\includegraphics[trim=0 0 0 0,clip,width=\timesubfigurewidth\textwidth]{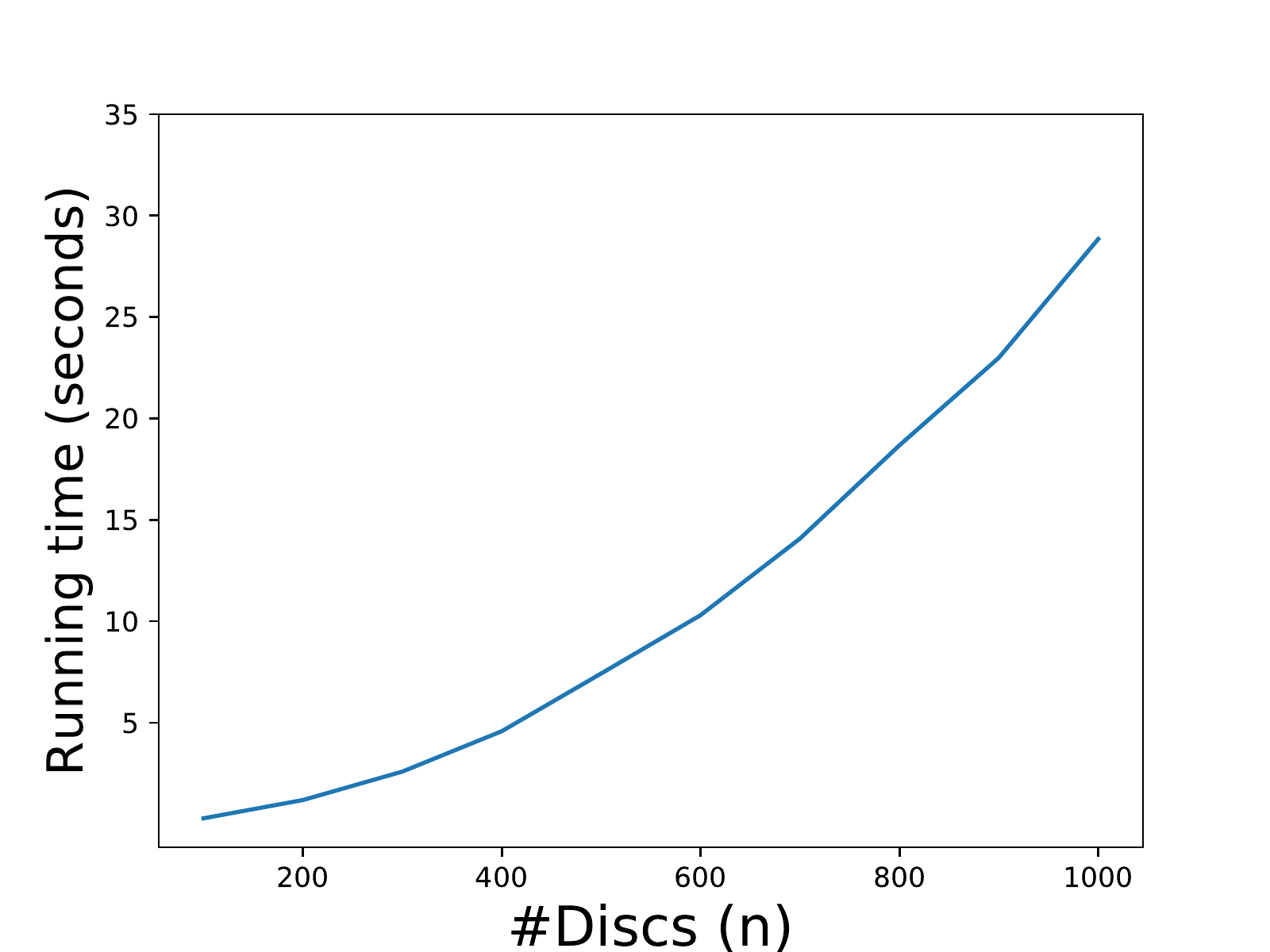}
		\caption{\timescaption{Circle}}
	\end{subfigure}%
	\begin{subfigure}{.45\linewidth}
		\centering
		\includegraphics[trim=0 0 0 0,clip,width=\timesubfigurewidth\textwidth]{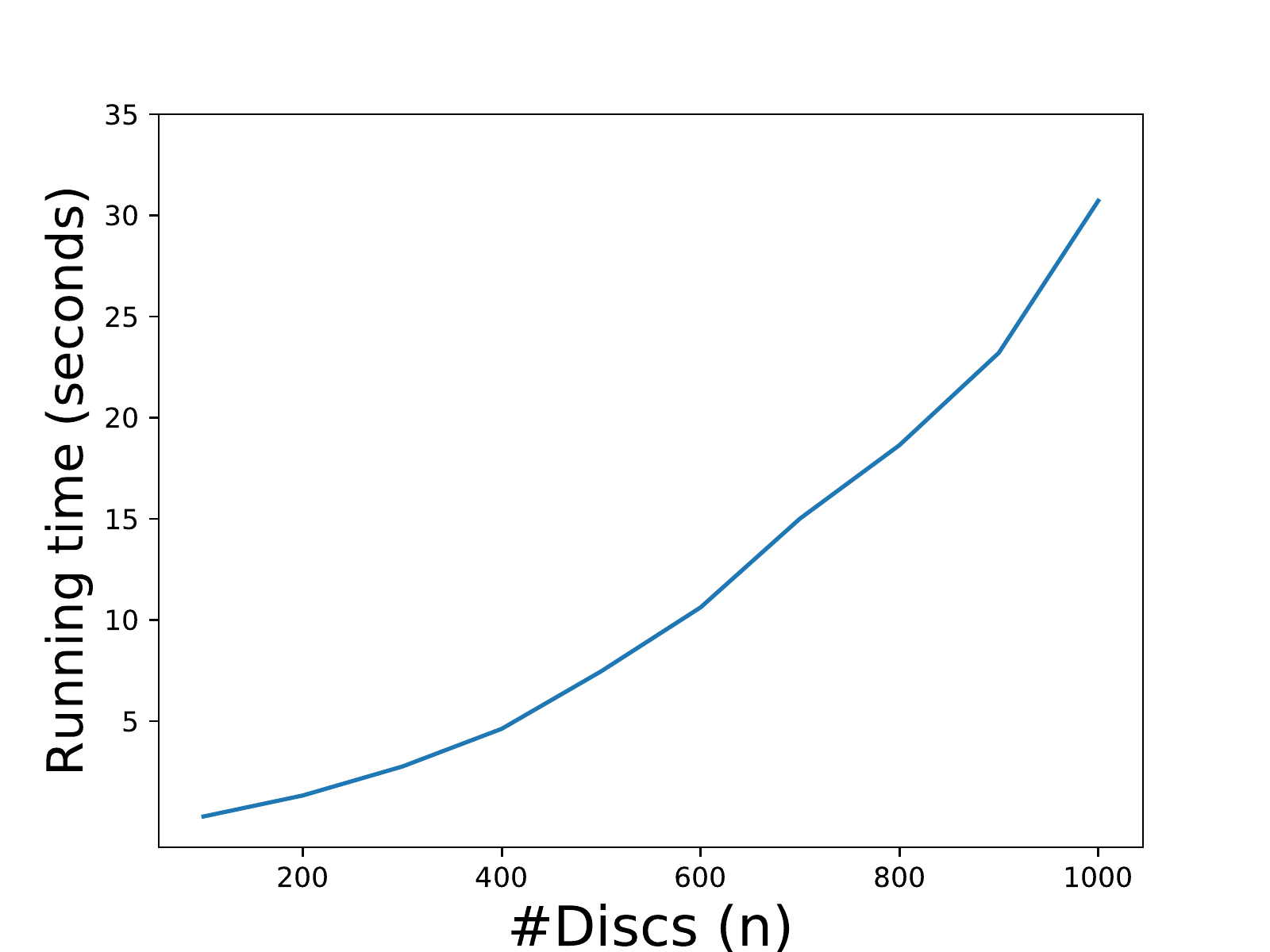}
		\caption{\timescaption{Packing}}
	\end{subfigure}\\[1ex]
	\begin{subfigure}{.45\linewidth}
		\centering
		\includegraphics[trim=0 0 0 0,clip,width=\timesubfigurewidth\textwidth]{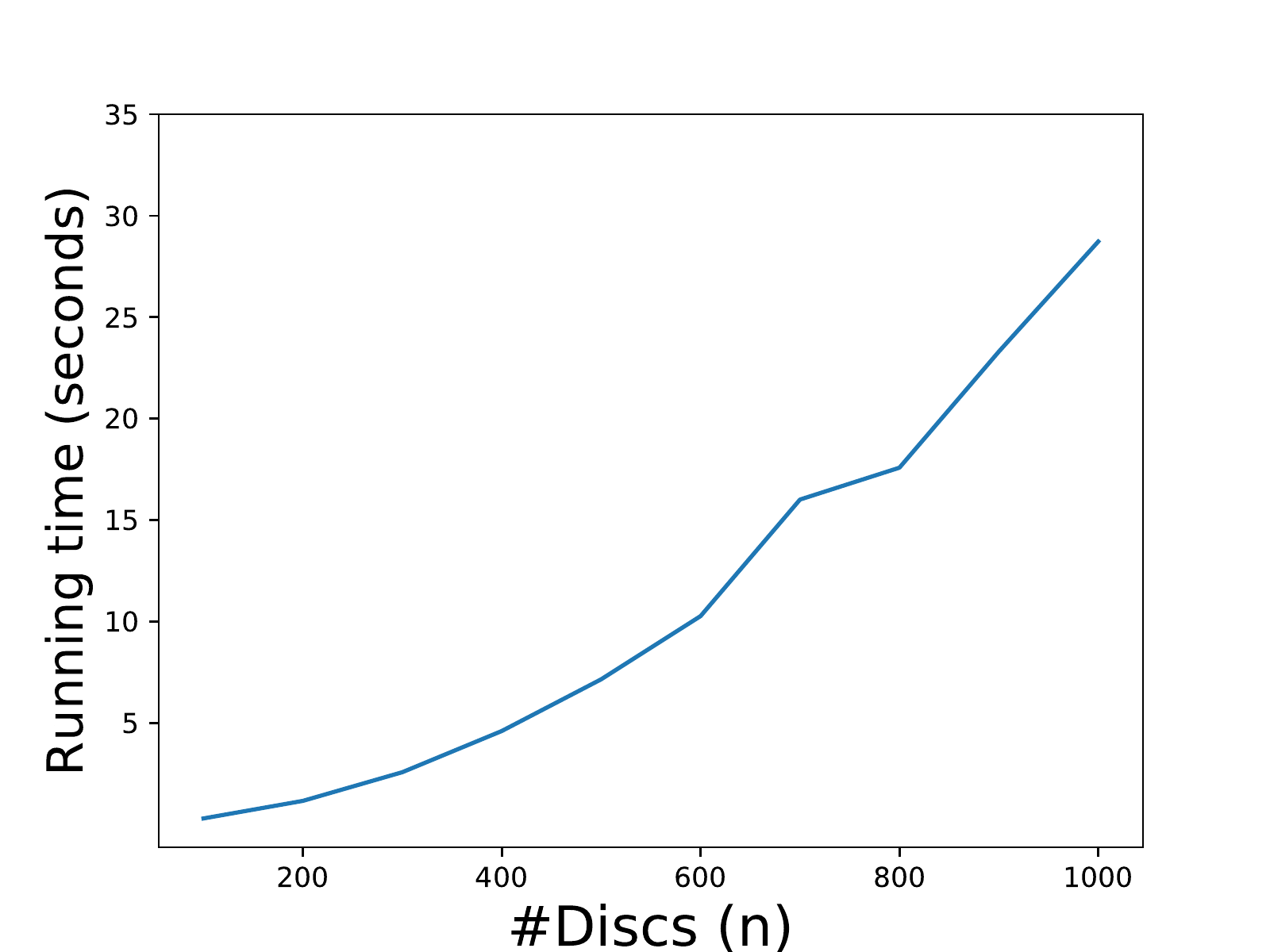}
		\caption{\timescaption{Cross} \newline \newline}
	\end{subfigure}
	\begin{subfigure}{.45\linewidth}
		\centering
		\includegraphics[trim=0 0 0 0,clip,width=\timesubfigurewidth\textwidth]{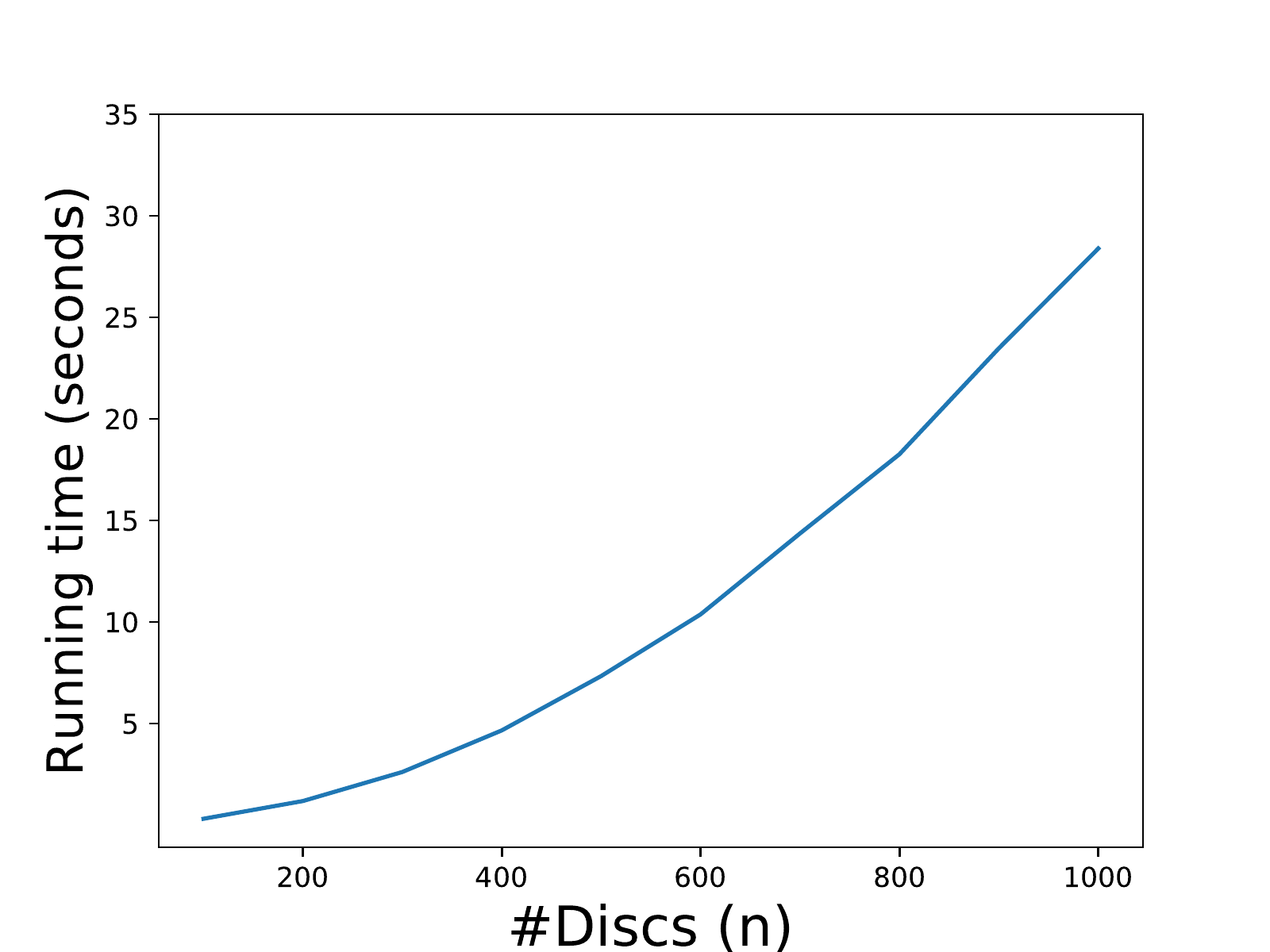}
		\caption{\timescaption{Random} The results are averaged on 10 different random input configurations.}
	\end{subfigure}
	\caption{
		\sf Running time of the heuristic as a function of the number of discs in the start (and hence also the target) configuration, for the different input types.
		Each entry (\#Discs$=100, 200, \ldots, 1{,}000$) in the first three graphs is the average running time of $10$ different runs, in each of which a different direction $\delta$ is sampled uniformly at random. In the fourth graph (Random input), each entry is the average on 10 random inputs, each run 10 times on different directions.}
	\label{fig:times}
\end{figure}

%% file: general/conclusion.tex
\section{Conclusion}  \label{section:conclusion}
\ifthesis
\vspace{-30pt}
\fi
This \paper presented results for unit disc reconfiguration, with special attention to physical space usage.
We gave reasonably efficient algorithms, for both the labeled and unlabeled versions, for computing a single translation for the target configuration that admits a valid reconfiguration, and that reduces the space used by the reconfiguration.
For the labeled version, the translation can be made optimal for space usage according to three natural optimization criteria.
For the unlabeled version, we gave heuristics for reducing the space usage, and upper-bounded the resulting space usage.
Finally, we implemented one of the heuristics for the unlabeled case and ran some experiments to demonstrate that our technique works very well in practice (on several rather difficult instances), in terms of the physical space size that it produced and in its computation time.

Our research can be extended in various ways within the space-awareness framework.
We could, for example, allow two translations per disc while aiming for minimal physical space (that also contains all the intermediate positions, in terms of the size of the bounding rectangle or disc).
These problems can be studied with or without a global rigid translation of the target configuration.
Alternatively, we could have considered variants where we allow an arbitrary initial rigid motion of the target configuration, or allow other motion paths instead of straight line paths.
Since these problems are more general, they are likely to be much harder to solve with optimal space usage.
One can also study space-aware reconfiguration for discs of varying sizes (the labeled version only), or for other, more complex shapes.

Viewing assembly planning from the space-aware perspective raises many challenging problems.
We aim to find the smallest space (e.g., a round tabletop of minimum radius) where we can put the separate parts that need to be assembled into the final product (given in some prespecified, translation-independent layout), and such that the entire assembly process can take place within this space.
The problem is more involved since we may need to store intermediate subassemblies, such that we can bring together some subassemblies into their relative placement in the final product, while avoiding other subassemblies, all within the same space.

%% file: chapters/appendix.tex
\appendix

\section{Interactive Examples}
In this \ssection, we present several examples of our geometric claims and observations, in an intuitive and interactive way, using the online and free tool --- GeoGebra~\cite{geogebra6}.
GeoGebra helped us to gather observations and to produce most of the figures in this \paper.
We describe three applets and we describe how to use them.
\begin{description}
	\item[A random input for the reconfiguration problem.]\leavevmode\\
	\url{https://www.geogebra.org/m/e5dn99jq}
	\begin{itemize}
		\item The start discs (in blue) are stationary.
		\item Move the vector $\vec{v}$ to translate the target discs (in orange).
		\item Change the number of disc by changing the slider denoted by \emph{n}.
		\item Change the size of the rectangle in which the configurations are created by changing the slider denoted by \emph{Size}; this also refreshes the randomness (note that sometimes discs of the same configuration can overlap --- refresh the randomness if it happens).
		\item Check the \emph{Labeled} checkbox to label the discs so the input will correspond to the labeled version.
		\item Check the \emph{AABR} or \emph{SED} checkboxes to see the criteria we aim to optimize the space for.
		\item For the labeled version, check the \emph{Hippodromes} or \emph{Vippodromes} checkboxes to see the hippodromes and vippodromes (only for the pairs $A$ and $B$) of both types defined in \Subsection~\ref{subsec:labeled_analysis}.
	\end{itemize}
	
	\ifthesis\newpage\fi
	\item[The construction of the vippodrome $\V^{(1)}_{AB}$ (same as in Figure~\ref{fig:vippo_construction}).]\leavevmode\\
	\url{https://www.geogebra.org/m/ckbmq4xx}
	\begin{itemize}
		\item Move around the centers of the discs to see how each affects the vippodrome structure.
		\item Check the \emph{Inner Tangents} checkbox to see the two inner tangents $\tau^-(B^S,A^S)$ and $\tau^+(B^S,A^S)$.
		\item Check the \emph{Vippodrome} checkbox to see the vippodrome $\V^{(1)}_{AB}$ as a dashed curve.
		\item Check the \emph{Wedge} checkbox to see the wedge $\W(A^S,B^S)$.
	\end{itemize}

	\item[The diametral disc presented in Case~(iii) of \Subsection~\ref{subsec:labeled_sed}.]\leavevmode\\
	\url{https://www.geogebra.org/m/gv6hnajw}
	\begin{itemize}
		\item The start configuration (in blue) are stationary.
		\item Move the vector $\vec{v}$ to translate the target configuration (in orange); observe how the function $\varphi(\vec{v})$ changes accordingly.
		\item Check the \emph{SED} checkbox to see the diametral disc $D_{s_0,t_0}(\vec{v})$.
		\item Check the \emph{K} checkbox to see $K_{s_0,t_0}$; observe that the disc encloses all the points if and only if $\vec{v}$ is inside the region $K_{s_0,t_0}$.
		\item Check the \emph{Vippodrome} checkbox to see an arbitrary vippodrome.
		The dark green intervals along the vippodrome boundary depict the edges of $Q$; recall Lemma~\ref{lemma:dq}.
	\end{itemize}
\end{description}

%% file: general/title.tex
\title{\bf Space-Aware Reconfiguration\thanks{Work by D.H. and G.M.L. has been supported in part by the Israel Science Foundation Grants~825/15,~1736/19, by the Blavatnik Computer Science Research Fund, and by a grant from Yandex.
Work by M.v.K. has been supported by the the Netherlands Organisation for Scientific Research under Grant~612.001.651.
Work by M.S. and G.M.L. has been supported by Grant~260/18 from the Israel Science Foundation.
Work by M.S. has also been supported by Grant G-1367-407.6/2016 from the German-Israeli Foundation for Scientific Research and Development, and by the Blavatnik Computer Science Research Fund.
An abridged version of this paper has been accepted for publication in the 14th International Workshop on the Algorithmic Foundations of Robotics, 2020~\cite{wafr}.
}}
\author[1]{Dan Halperin}
\author[2]{Marc van Kreveld}
\author[1]{Golan Miglioli-Levy}
\author[1]{Micha Sharir}
\affil[1]{%
	The Blavatnik School of Computer Science, Tel Aviv University, Israel\\
\email{\{danha@tau.ac.il, golanlevy@mail.tau.ac.il, michas@tau.ac.il\}}}
\affil[2]{%
Dept. of Information and Computing Sciences, Utrecht University, the Netherlands\\
\email{\{M.J.vanKreveld@uu.nl\}}
}
\date{}                     
\renewcommand\Affilfont{\itshape\small}

\maketitle